\documentclass[a4paper,11pt,nosumlimits,reqno]{amsart}
\usepackage{latexsym,amsmath,amssymb,enumerate,bbm,fullpage,amscd,graphicx,color}
\usepackage[english]{babel}
\usepackage{etoolbox,calc}
\usepackage{subcaption}

 \usepackage[foot]{amsaddr}

\makeatletter
\def\appendixname{Appendix}%  
\appto\appendix{%
  \addtocontents{toc}{\patch@l@section}%  
  \appto\appendixname{ }%  
}
\protected\def\patch@l@section{%
  \patchcmd{\l@section}{1.5em}{\widthof{\appendixname\space}+2.5em}{}{}%
}
\makeatother

\usepackage{tikz}
\usepackage{tikz-cd}
\usetikzlibrary{knots}
\usetikzlibrary{intersections}
\usepackage{subcaption} 
\usepackage{booktabs,array}%  
\usepackage{colortbl}% 
\usepackage{marvosym} 
\usepackage{graphicx}%  
\usepackage{textcomp}
\usepackage{simplewick}
\usepackage{float}
\usepackage{xcolor}
\usepackage{sansmath}
\usepackage{marvosym}
 
\colorlet{texto}{black!50!gray} 
\usepackage{hyperref}
\hypersetup{
     colorlinks   = true,
     citecolor    = green!50!black
} 
\hypersetup{linkcolor=blue!58!black}

\usetikzlibrary{trees}
\usetikzlibrary{decorations.pathmorphing}
\usetikzlibrary{decorations.markings}

 \usetikzlibrary{calc,decorations.pathmorphing,patterns}
\usetikzlibrary{decorations.text}
\usetikzlibrary{arrows}
\usetikzlibrary{decorations.pathmorphing,backgrounds,positioning,fit,petri}
\usetikzlibrary{automata}
\usetikzlibrary{graphs}
\usetikzlibrary{shadows}
\usetikzlibrary{trees}
\usetikzlibrary{decorations.pathmorphing}
\usetikzlibrary{decorations.markings} 

\usepackage{enumitem}

\newtheorem{thm}{Theorem}
\newtheorem{lem}[thm]{Lemma}
\newtheorem{prop}[thm]{Proposition}

\theoremstyle{definition}
\newtheorem{defn}{Definition}
\newtheorem{rem}{Remark}
 \newtheorem{example}{Example}
 \usepackage{tensor}
\renewcommand{\and}{\quad\mbox{and}\quad}
\newcommand{\dpa}{\partial}
\newcommand{\T}{\mathbb{T}}
\newcommand{\Tr}{\mathrm{Tr}}
\newcommand{\al}{\alpha}
\newcommand{\mtr}[1]{\mathrm{#1}}
\newcommand{\mtf}[1]{\mathfrak{#1}}
\newcommand{\A}{\mathcal{A}}
\newcommand{\dif}[1]{\mathrm{d}#1} 
\newcommand{\re}{\mathbb{R}}

\newcommand{\pr}[1]{\mathrm{pr}_{#1}}
\newcommand{\si}{\sigma}
\newcommand{\G}{\mathcal{G}}
\renewcommand{\H}{\mathcal{H}}
\newcommand{\B}{\mathcal{B}}
\newcommand{\C}{\mathbb{C}}

\newcommand{\ii}{\mathrm{i}}
\newcommand{\ee}{\mathrm{e}}
 
\newcommand{\inv}{^{-1}} 
\newcommand{\mtc}[1]{\mathcal{#1}} 
\newcommand{\Z}{\mathbb{Z}}
  
\newcommand{\U}{\mtc{U}}
\newcommand{\V}{\mtc{V}}
  
\newcommand{\im}{\mtr{im}}    
\newcommand{\with}{\,\,\mtr{with} \,\,}
\newcommand{\where}{\qquad\mbox{where}\,\,}

\newcommand{\mtb}[1]{\mathbb{#1}}

\newcommand{\F}{\mathcal{F}}

\newcommand{\R}{\mathcal R}

\newcommand{\Df}{\mathbb{\mathcal{D}}}

\newcommand{\Sym}{\mathfrak{S}}

\newcommand{\J}{\mathcal{J}}
\newcommand{\hp}[1]{^{(#1)}}
\newcommand{\Tg}[1]{\mathcal{T}_{#1}}
\renewcommand{\phi}{\varphi}
\newcommand{\fey}{\mathfrak{Feyn}}
\newcommand{\feyn}{\fey_3(\phi^4)}
\newcommand{\Riem}{\mathsf{Riem}_{\mathrm{c,cl,o}}}
\newcommand{\Riemd}{\mathsf{Riem}_{\mathrm{cl,o}}}
\newcommand{\Cob}{\mathsf{Cob}}

\newcommand{\Grph}[1]{\mathsf{Grph}_{\mathrm{col},#1}}
\newcommand{\tcol}{\Grph{3}}

\newcommand{\Rb}{\mathcal{R}} 
\newcommand{\tg}{\hphantom{!} \! {\#\hphantom{!}}\!}

\newcommand{\MyTo}[1][]{\mathrel{\tikz \draw [-stealth, #1] (0,0) (0,0.5ex) -- (1.0em,0.5ex);}}
\newcommand{\ToLong}[1][]{\mathrel{\tikz \draw [-stealth, #1] (0,0) (0,0.5ex) -- (1.5em,0.5ex);}}  
\newcommand{\MyMapsto}[1][]{\mathrel{\tikz \draw [-stealth, cap=round, #1] (0,0) (0,0.75ex) -- (0,0.220ex) (0,0.50ex) -- (1.0em,0.50ex);}}

\renewcommand{\to}{\MyTo} 

\renewcommand{\mapsto}{\MyMapsto}

\makeatletter
\def\moverlay{\mathpalette\mov@rlay}
\def\mov@rlay#1#2{\leavevmode\vtop{%
   \baselineskip\z@skip \lineskiplimit-\maxdimen
   \ialign{\hfil$\m@th#1##$\hfil\cr#2\crcr}}}
\newcommand{\charfusion}[3][\mathord]{
    #1{\ifx#1\mathop\vphantom{#2}\fi
        \mathpalette\mov@rlay{#2\cr#3}
      }
    \ifx#1\mathop\expandafter\displaylimits\fi}
\makeatother

\newcommand{\cupdot}{\charfusion[\mathbin]{\cup}{\cdot}}

\makeatletter
\def\leqno{\tagsleft@false}
\def\reqno{\tagsleft@false}
\def\fleqn{\@fleqnfalse}
\def\cneqn{\@fleqnfalse}
\makeatother 

 \DeclareRobustCommand{\gobblefive}[5]{}
\newcommand*{\SkipTocEntry}{\addtocontents{toc}{\gobblefive}}
\begin{document}

\title{Surgery in colored tensor models} 
\author{Carlos I. P\'erez-S\'anchez} 
\address{Mathematisches Institut der Westf\"alischen Wilhelms-Universit\"at, Einsteinstra\ss e 62,
48149 M\"unster, Germany
}
\email{perezsan@uni-muenster.de}

 \begin{abstract}
Rooted in group field theory and matrix models, random tensor models are a recent background-invariant approach to quantum gravity in arbitrary dimensions.  
Colored tensor models (CTM) generate random triangulated orientable (pseudo)-manifolds.  
We analyze, in low dimensions, which known spaces are triangulated by specific CTM interactions.  
As a tool, we develop the graph-encoded surgery that is compatible with the quantum-field-theory-structure  
and use it to prove that a single model, the complex $\varphi^4$-interaction in rank-$2$, 
generates all orientable $2$-bordisms, thus, in particular, also all orientable, closed surfaces.  
We show that certain quartic rank-$3$ CTM, the $\varphi_3^4$-theory, has as boundary sector all closed, possibly disconnected, 
orientable surfaces. Hence all closed orientable surfaces are cobordant via manifolds generated by the $\varphi_3^4$-theory.\\
\\
\textit{Keywords:} 
Random tensor models;  Feynman diagrams; matrix models; quantum gravity.  \\
\noindent
% \\
% \textit{MSC:} 83C47, 81T40, 81T18. 
\end{abstract}

\maketitle

\tableofcontents
\thispagestyle{empty}
\clearpage
%% 
%%  = = = = = = = = = = = = = = = = = = = = = = = = = = = = = = = = = =  
%% 
%%  = = = = = = = = = = = = = = = = = = = = = = = = = = = = = = = = = =  
%% 
%%  = = = = = = = = = = = = = = = = = = = = = = = = = = = = = = = = = =  

\section{Introduction}

Colored tensor models (CTM) have recently flourished as a
random-geometry framework that has proven the ability to model quantum
gravity in arbitrary dimension  $D$   
\cite{Ambjorn,Gurau:2009tw,RTM_QG,tensor_theory_space}, partially
following the line of thought of the $2$-dimensional quantum gravity
modeled by random matrices \cite{dFGZ}. 
CTMs are quantum field theories
for rank-$D$ tensors whose indices transform independently 
under given $D$ representations of unitary groups 
in a very simple way (see below). 
% % This has been called  \textit{color} and, in particular,
% coloring implies that the indices satisfy no (skew-)symmetry, no 
% cocycle condition, etc. 
The bridge to physics is, essentially,
\begin{align}   
 \label{eqn:correspondencia}
 \mbox{\textit{Quantum} }\to\mbox{\textit{Random }$\quad$ and } \quad
 \mbox{\textit{Gravity} }\to\mbox{\textit{Graph-encoded} $D$-\textit{geometry.} }  
\end{align}
The Euclidean path integral formulation of 
CTM defines a measure that facilitates the first correspondence.
The second `map' is
a simplicial version of the known General Relativity correspondence
for $D\geq 2$, whose discrete analogue is the Regge action in terms
of the deficit angles \cite{critical}, here expressed in a
graph-theoretical context.  The Feynman graphs of colored tensor
models have enough structure to encode a sensible space, thus, both
correspondences in \eqref{eqn:correspondencia} harmonically coexist.
A condensed summary for this framework 
is its ability to generate triangulations of (pseudo)manifolds\footnote{
  Pseudomanifolds are simplicial complexes that are non branching,
  pure and strongly connected. Since we aim at the construction of
  actual PL-manifolds, we do not deepen in that concept.  Moreover we
  work here in dimensions $2$ and $3$ and piecewise linear manifolds
  is all we need, by Moise's theorem \cite{moise}.  Thus, henceforth
  we only write `manifold'.} that can be averaged 
  by using a Boltzmann weight, $\exp(-S)$, that a particular model's 
  classical action functional $S$ determines. 
  \par

  In a historical vein, the term \textit{color}, introduced by di
  Francesco \cite{diFrancesco_rect}, appeared first, as many
  conceptions in tensor models do, in the context in the theory of
  matrix models. In that setting colors stand for the different sizes of
  rectangular matrix fields. The idea that `coloring' 
  prevents certain indices from being summed (contracted) with each other 
  was successfully carried on by
  Gur\u{a}u \cite{Gurau:2009tw}, who introduced several `colored'
  tensor fields, extending di Francesco's idea to the context of Group
  Field Theories \cite{Freidel,OritiGFT} in order to exclude 
  graphs that could not encode reasonable spaces.  The
  additional tensor fields can be integrated out, thus obtaining an
  effective action for a single field (see
  e.g. \cite[Sec. 5]{diFrancesco_rect}), which, however, retains the
  colored structure.  \par 
  Nowadays \textit{(random) tensor models} stands
  for a rather boarder cluster of alike theories \cite{MO,On,su2} with
  physically promising features. In particular, just as matrix models,
  they support a large-$N$ `t Hooft's expansion
   which is controlled by an integer called \textit{Gur\u{a}u's degree}
  \cite{Nexpansion,Nexpansion_coloured} that 
  replaces the genus in matrix models (see Rem. \ref{being_melon}). 
  Parallel to the fairly vivid study of the QFT-techniques of tensor
  models (e.g. renormalization \cite{4renorm,renormTFT,krajewskireiko,dine_beta,3Dbeta,su2,carrozzaPhD}), 
  a topology and geometry `quota' ---in the CTM-setting 
  is encoded in graph theory--- that leads us to a better understanding of the 
  gravitational-modeling, also deserves attention. These topics 
  for low dimensional scenarios is what this paper is all about. \par 
  At the core of the link between graph
  theory and geometry that concerns us lies \textit{Pezzana's
    theorem} \cite{pezzana} on manifold crystallization.
  It allows piecewise linear manifolds to be represented by decorated
  graphs, the so-called \textit{colored graphs}. 
  This family of graphs corresponds to the 
  Feynman graphs of colored tensor models.
  Thus, after Gur\u{a}u's work \cite{Gurau:2009tw,gurauReview},
  Pezzana's theorem yields a surjection
\[ 
 \{\mbox{\textit{All} rank-$D$ tensor model actions }S_{\mtr{int}}\} 
 \ToLong
 \{\mbox{PL-$D$-manifolds}\}.
\]
What Pezzana's theorem does not specify is the tensor-model action
that generates the graphs that represent certain class of manifolds.
In physics one commonly scrutinizes a single model.  Therefore it is 
interesting to pose the following question:
\begin{align}   
 \label{eqn:primed}
\textit{Given a class of manifolds,
which CTM-action generates it?} \tag*{$(\star)$}
\end{align}
The action should be polynomial by physical reasons.  Thus, given a
family $\mathcal C$ of manifolds (up to equivalence $\sim$), one
wishes to find a tensor-model action $S_{\mathrm{int}}$ and to prove
the surjectivity of the composition
\[
\raisebox{-.45\height}{
\includegraphics[height=.55cm]{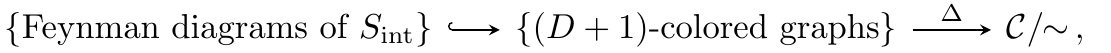}
}
\]
where $\Delta$ is a `manifold reconstruction'-scheme 
(see Sec. \ref{reconstruction_from_graphs} and  \cite{survey_cryst}).
Techniques like the bubble-homology of graphs 
\cite{Gurau:2009tw} assist in 
distinguishing spaces (see Sec. \ref{bubblehomology})
and shall be used here.  \par
We fix now the setting to answer \ref{eqn:primed} in low dimensions,
that is, we choose the right family member of tensor models\footnote{
  The generation of graphs and their characteristics are dependent on
  the type, structure and field chosen of tensor model chosen,
  i.e. whether it is colored, or hybrid, as multi-orientable tensor
  models \cite{reviewTanasa}; and the vector spaces where the tensors
  are defined can also over $\re$ or $\C$.}.  Since we want to prove 
the surjectivity of certain maps, the result is stronger if we keep the classes of graphs
emerging in that framework at its minimum, which means a
large-symmetry in the action.  The right choice is the complex 
CTM, as exposed in Section \ref{preliminaries}.  \par

Having chosen the setting, we choose now the potential
$S_{\mathrm{int}}$. We work with rank-$2$ and rank-$3$ tensor models
and in both instances we take a quartic potentials (which due to their
distinct underlying structures look somehow different). \par Our strategy is
mainly surgery: In certain categories of manifolds, by using surgery
 one is able to generate new spaces and readily compute  some 
topological invariants of them from the properties of their parts.  It
is therefore desirable to have this in the context of graphs.  The
existent concept in the context of the graph theoretical
representation of piecewise linear manifolds by Pezzana, Gagliardi,
Ferri \textit{et al.} \cite{survey_cryst}---to our knowledge the only
available concept--- unfortunately does not respect the QFT-structure
of tensor models, as we show here (Sec. \ref{gsurgery},
Rem. \ref{thm:crys}). We develop a QFT-compatible and CTM-compatible
surgery aiming at answering \ref{eqn:primed} for dimension $2$, 
going further also to dimension $3$. We stress that the methods provided
by the theory of crystallization 
do not care about the graphs being Feynman diagrams of certain
model.
\par 
Concretely,
we obtain the following: A well-defined $3$-colored graph surgery 
 (Definition \ref{thm:surg_ribb} for the connected sum and in Theorem \ref{thm:gen_bordisms}, creation 
 of boundary components), which restricts to the set of Feynman graphs
 of a given model. Remark \ref{thm:crys} explains the need of this operation. 
 These concepts lead to the parametrization of all orientable $2$-bordisms 
 by the diagrammatics of certain quartic potential 
 (Thm. \ref{thm:gen_bordisms}) and, in particular, 
  the generation of all closed, connected
   orientable surfaces from (vacuum graphs of) the rank-$2$
   \textit{quartic} potential (Lemma \ref{thm:surj_2}). This stronger
   than the $\re$-matrix model case, since, say, the following graph
 \begin{equation}\label{eqn:ribonfisico}
 \raisebox{-1.2cm}{
 \includegraphics[height=2.2cm]{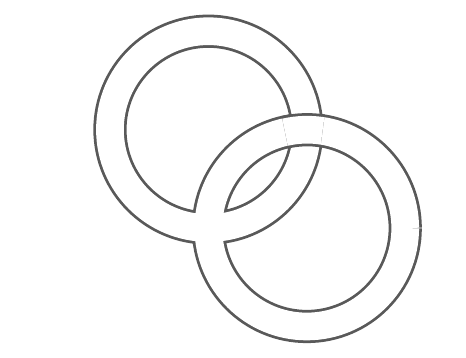}}
 \end{equation}
 of the real quartic matrix model is forbidden in any complex rank-$2$ theory.\par 
 
 Working in one dimension higher, 
 we lift the surgery of $3$-colored
 graphs to an operation on open $4$-colored graphs and use  
 this operation to prove that:
\begin{itemize}
\item[\scriptsize$\bullet$] \normalsize Boundary graphs of 
 a certain  quartic rank-$3$ model, the so-called $\phi_3^4$-theory,
   generate all closed, connected orientable surfaces.
   That is, those surfaces are null-bordant in the 
   sense of the $\phi_3^4$-theory.  
 \item[\scriptsize$\bullet$] More generally, any two compact, orientable 
  closed, \textit{possible disconnected} surfaces 
  are $3$-cobordant by (a space reconstructed from)
  certain \textit{connected} $\phi_3^4$-Feynman graph (Thm. \ref{thm:3cobord}).
\end{itemize}

This article has the following structure.  We motivate first, in
Section \ref{preliminaries}, the study of colored graphs by
introducing from scratch, albeit quite straightforwardly, colored
tensor models. A rather lengthy introduction on the graph theoretical
machinery shall be provided there. The reader which is familiar 
with CTM can skip that section. Examples there (which we
do not skimp on) shall become useful later on, though.  Sections
\ref{gsurgery} and \ref{nontriviality} are the core, where we prove our
claims above.  The reader that does not feel familiar with
graph-homology and/or ribbon graphs might find useful 
Appendices \ref{appA} and/or \ref{app_ribbons}, respectively.
 
%% 
%%  = = = = = = = = = = = = = = = = = = = = = = = = = = = = = = = = = =  
%% 
%%  = = = = = = = = = = = = = = = = = = = = = = = = = = = = = = = = = =  
%% 
%%  = = = = = = = = = = = = = = = = = = = = = = = = = = = = = = = = = =  
 
\section{Tensors models and their graph theory}
 \label{preliminaries}

 Colored tensor models are quantum field theories for tensorial
 objects specified by an integer $D$ and by so-called
 \textit{interaction vertices}. The integer $D\geq 2$ is the rank of
 tensors fields
 $\bar\phi,\phi:\H_1\otimes \H_2\otimes \cdots \otimes \H_D\to \C$ on products
 of Hilbert spaces $\H_1,\ldots,\H_D$ and the interaction vertices
 $\{\Tr_{\mathcal{ B}_\alpha}(\phi,\bar\phi)\}_\al$ are determined
 by invariance under products of unitary groups
 $\U(\H_1)\times \U(\H_2)\times \cdots \times \U(\H_D)$, as we explain next. 
 Renormalization should care for a second selection-process
 of interaction vertices consisting in suppressing those traces $\Tr_\B$ which render 
 the theory non-renormalizable (see e.g. \cite[Sec. 3.7]{krajewskireiko} 
 for a list of the vertices $\B$; those we shall deal with here are renormalizable).  
 Since the following discussion and definitions can be
 carried out without any effort to higher rank, for sake of
 concreteness we restrict ourselves to rank-$3$, thus considering
 tensors $\phi,\bar\phi:\H_1\otimes\H_2 \otimes \H_3\to \C$. We also assume for
 simplicity, that $\H_c$ are large but finite dimensional.\par

 For any integer $c=1,2,3,$ which one calls \textit{color}, take a
 basis $\{\vartheta_c^{a} \, : \,a\in I_c\}$ of the dual of
 $\H_c$. Here each $I_c \subset \Z$ serves as an index set, which
  will be often left implicit. We let
 $\phi=:\sum_{a_1,a_2,a_3} \phi_{a_1a_2a_3}\vartheta^{a_1}_1\otimes
 \vartheta^{a_2}_2 \otimes \vartheta^{a_3}_3$.
 Each color-$c$ index $a_c$ transforms independently under a change of
 basis of $\H_c$. The coordinates transform therefore under unitary
 elements $W^{(c)}\in \U(\H_c)$ like
\begin{align*}
\phi'_{a_1a_2a_3}&=\sum_{b_1,b_2,b_3}W^{(1)}_{a_1b_1}W^{(2)}_{a_2b_2}W^{(3)}_{a_3b_3}
\phi_{b_1b_2b_3}\,,  \\
\bar \phi'_{a_1a_2a_3}&= \sum_{b_1,b_2,b_3}\overline W^{(1)}_{a_1b_1}
\overline W^{(2)}_{a_2b_2}\overline W^{(3)}_{a_3b_3}
\bar \phi_{b_1b_2b_3}\,.
\end{align*}
We take as classical action only invariants under the  group 
$\U(\H_1)\times \U(\H_2) \times \U(\H_3) $. 
The only quadratic invariant, $
S_0[\phi,\bar\phi]= \Tr_{2}(\phi,\bar\phi):=\sum_{a_1,a_2,a_3} \bar\phi_{a_1a_2a_3} \phi_{a_1a_2a_3}
$ is understood as the kinetic term\footnote{
In future work it will also consider a 
slightly modified trace with a symmetry-breaking term  $E$, though:  
\begin{equation} \nonumber S[\phi,\bar\phi]=\Tr_2{(\bar \phi,E\phi)}+
\sum_\alpha\Tr_{\mathcal{
    B}_\alpha}(\phi,\bar\phi)\,, \label{generalthree} \end{equation}
with $E:\H_1\otimes \H_2 \otimes \H_3\to\H_1\otimes \H_2 \otimes \H_3$
`self-adjoint', $\Tr_2(\bar\phi,E\phi)=\Tr_2(E\bar\phi,\phi)$.  The
first term is distinguished, and being quadratic in $\phi$, it
represents the kinetic part of the action, where $E$ could be
interpreted as the Laplacian. This allows to state Group Field
Theories, via Fourier-transform, as colored tensor models.}.  
Higher order terms like
\begin{equation} \label{v2explicito}
\Tr_{\V_2}(\phi,\bar\phi)=\lambda \sum_{\mathbf{a,b,p,q}} \bar\phi_{q_1q_2q_3} \bar\phi_{p_1p_2p_3}
(\delta_{a_1p_1}\delta_{b_1q_1}\delta_{a_2q_2}\delta_{b_2p_2}\delta_{a_3p_3}\delta_{b_3q_3})
\phi_{a_1a_2a_3} \phi_{b_1b_2b_3}\, 
\end{equation}
are the interaction vertices appearing in 
$S_{\mtr{int}}=\sum_\alpha\Tr_{\mathcal{ B}_\alpha}(\phi,\bar\phi)$. 
By Schur's Lemma the tensors in the trace can be
contracted only with $\delta$'s (or multiples thereof, which can be absorbed
in the \textit{coupling constant} $\lambda$).  The explicit expression
of each one of these $\Tr_{\mathcal{ B}_\alpha}$ has certain number
(say $k$) of fields $\bar \phi_{\mathbf p^{1}},\ldots, \bar
\phi_{\mathbf p^{k}}$, which are fully contracted with same number of fields 
$  \phi_{\mathbf a^{1}},\ldots, \phi_{\mathbf a^{k}}$, where
$\mathbf{p}^{i}=(p^{i}_1,p^{i}_2,p^{i}_3),\mathbf{a}^{i}
=(a^{i}_1,a^{i}_2,a^{i}_3)\in I_1\times I_2\times I_3$ for each
$i=1,\ldots,k$.  \par 
Here it is handy, in order to avoid writing
these long expressions, to represent these vertices using either
\textit{stranded} graphs or their \textit{colored, bipartite}
version. The former is obtained as follows: each invariant trace must
contain $\phi_{a_1a_2a_3}$ and the complex conjugate field
$\bar\phi_{p_1p_2p_3}$.  We represent these graphically by associating to
them a bunch of $D=3$ white nodes with outgoing strands and a bunch of
$D=3$ dark nodes with incoming strands respectively:
\[
\vspace{-.2cm}
\phi_{a_1a_2a_3}\mapsto \tikz[scale=.75,  outer ysep=-2 pt,
    baseline=-.5ex,shorten >=.1pt,node distance=18mm, 
    semithick,auto,
    every state/.style={fill=texto ,draw=texto,inner sep=.4mm,text=black,minimum size=0},
    accepting/.style={fill=gray!50!black,text=white},
    initial/.style={white,text=black}
]{ \node[state,minimum size=0,circle, fill=white,  draw=green!58!black] (A11) at 
(0,0) {\tiny 1}; 
 \node[state,minimum size=0,circle, fill=white, draw=red] (A12) at (.5,0) {\tiny 2}; 
 \node[state,minimum size=0,circle, fill=white, draw=blue] (A13) at (1,0) {\tiny 3}; 
 \draw[green!58!black,-stealth] (A11) ..controls (.1,.7) and (.2,.8).. ($(A11)+(.6,1)$) ;
 \draw[red,-stealth] (A12) ..controls ($(.1,.7)+(A12)$) and ($(.2,.8)+(A12)$).. ($(A12)+(.6,1)$) ;
 \draw[blue,-stealth] (A13) ..controls ($(.1,.7)+(A13)$) and ($(.2,.8)+(A13)$).. ($(A13)+(.6,1)$) ;
 }\and
 \qquad 
 \bar\phi_{p_1p_2p_3}\mapsto \tikz[scale=.75,  outer ysep=-2 pt,
    baseline=-.5ex,shorten >=.1pt,node distance=18mm, 
    semithick,auto,
    every state/.style={fill=texto ,draw=texto,inner sep=.4mm,text=black,minimum size=0},
    accepting/.style={fill=gray!50!black,text=white},
    initial/.style={white,text=black}
]{ \node[state,minimum size=0,circle , draw,fill=blue!20] (A11) at
(0,0) {\tiny 3}; 
 \node[state,minimum size=0,circle, draw,fill=red!30] (A12) at (.5,0) {\tiny 2}; 
 \node[state,minimum size=0,circle,draw,fill=green!20!lightgray] (A13) at (1,0) {\tiny 1}; 
 \draw[blue,-stealth reversed] (A11) ..controls (.1,.7) and (.2,.8).. ($(A11)+(.6,1)$) ;
 \draw[red,-stealth reversed] (A12) ..controls ($(.1,.7)+(A12)$) and ($(.2,.8)+(A12)$).. ($(A12)+(.6,1)$) ;
 \draw[green!58!black,-stealth reversed] (A13) ..controls ($(.1,.7)+(A13)$) and ($(.2,.8)+(A13)$).. ($(A13)+(.6,1)$) ;
 } \mbox{(reversed order).}
 \]
We associate to each $\delta$ single lines, in such a
way that $\delta_{a_1p_1}=\tikz
[scale=.65,  outer ysep=-2 pt,
    baseline=-.5ex,shorten >=.1pt,node distance=18mm, 
    semithick,auto,
    every state/.style={fill=texto ,draw=texto,inner sep=.4mm,text=black,minimum size=0},
    accepting/.style={fill=gray!50!black,text=white},
    initial/.style={white,text=black}
]{
\node[ green!58!black] at (0.5,.25) {\tiny 1};
\draw[stealth reversed-stealth, green!58!black] (0,0) -- (1,0);
}$ connects the color-$1$
(green) strands of $\bar\phi_{p_1p_2p_3}$ and
$\phi_{a_1a_2a_3}$;
 $\delta_{a_2p_2}=\tikz
[scale=.65,  outer ysep=-2 pt,
    baseline=-.5ex,shorten >=.1pt,node distance=18mm,
    semithick,auto,
    every state/.style={fill=texto ,draw=texto,inner sep=.4mm,text=black,minimum size=0},
    accepting/.style={fill=gray!50!black,text=white},
    initial/.style={white,text=black}
]{
\draw[stealth reversed-stealth, red] (0,0) -- (1,0);
\node[ red] at (0.5,.25) {\tiny 2};
}$ connects the color-$2$
(red) strands and  $\delta_{a_3p_3}=\tikz
[scale=.65,  outer ysep=-2 pt,
    baseline=-.5ex,shorten >=.1pt,node distance=18mm,
    semithick,auto,
    every state/.style={fill=texto ,draw=texto,inner sep=.4mm,text=black,minimum size=0},
    accepting/.style={fill=gray!50!black,text=white},
    initial/.style={white,text=black}
]{
\node[blue] at (0.5,.25) {\tiny 3};
\draw[stealth reversed-stealth, blue] (0,0) -- (1,0);
}$
those of color $3$ (blue). Therefore
the trace in eq. \eqref{v2explicito} is 
\[
 \centering
  \raisebox{-0.5\height}{\includegraphics[height=3.0cm]{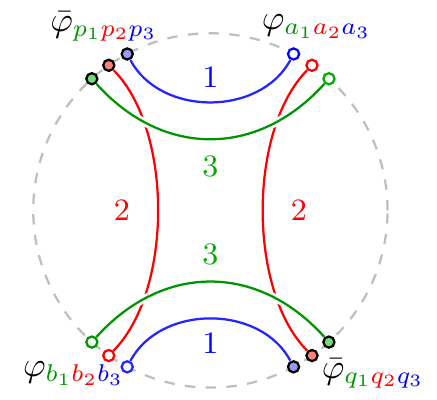}}
  \hspace*{.2in} 
\]

It turns out that these graphs are still somehow quite elaborate and
we will opt for even more simplified graphs that contain the same
information. To the interactions one associates finite regularly edge-$3$-colored
vertex-bipartite graphs.  This picture is obtained from the stranded
representation of graphs by collapsing the nodes $\tikz[scale=.65,
  outer ysep=-2 pt, baseline=-.5ex,shorten >=.1pt,node distance=18mm,
  semithick,auto, every state/.style={fill=texto ,draw=texto,inner
    sep=.4mm,text=black,minimum size=0},
  accepting/.style={fill=gray!50!black,text=white},
  initial/.style={white,text=black} ]{ \node[state,minimum
    size=0,circle, fill=white, draw=green!58!black] (A11) at (0,0) {};
  \node[state,minimum size=0,circle, fill=white, draw=red] (A12) at
  (.4,0) {}; \node[state,minimum size=0,circle, fill=white, draw=blue]
  (A13) at (.8,0) {};}$ of $\phi_{a_1a_2a_3}$ to a single white vertex
$\tikz[scale=.65, outer ysep=-2 pt, baseline=-.5ex,shorten >=.1pt,node
  distance=18mm, semithick,auto, every state/.style={fill=texto
    ,draw=texto,inner sep=.4mm,text=black,minimum size=0},
  accepting/.style={fill=gray!50!black,text=white},
  initial/.style={white,text=black} ]{ \node[state,minimum
    size=0,circle, fill=white, draw=black] (A11) at (0,0) {}; }$, and
those of $\bar\phi_{p_1p_2p_3}$, i.e.  $\tikz[scale=.65, outer ysep=-2
  pt, baseline=-.5ex,shorten >=.1pt,node distance=18mm,
  semithick,auto, every state/.style={fill=texto ,draw=texto,inner
    sep=.4mm,text=black,minimum size=0},
  accepting/.style={fill=gray!50!black,text=white},
  initial/.style={white,text=black} ]{ \node[state,minimum
    size=0,circle , draw,fill=blue!40] (A11) at (0,0) {};
  \node[state,minimum size=0,circle, draw,fill=red!50] (A12) at (.4,0)
       {}; \node[state,minimum
         size=0,circle,draw,fill=green!40!lightgray] (A13) at (.8,0)
       {}; }$, to a black vertex $\tikz[scale=.65, outer ysep=-2 pt,
  baseline=-.5ex,shorten >=.1pt,node distance=18mm, semithick,auto,
  every state/.style={fill=texto ,draw=texto,inner
    sep=.4mm,text=black,minimum size=0},
  accepting/.style={fill=gray!50!black,text=white},
  initial/.style={white,text=black} ]{ \node[state,minimum
    size=0,circle , draw,fill=black] (A11) at (0,0) {}; }$.
Accordingly, the three strands of $\phi_{a_1a_2a_3}$ join at
$\tikz[scale=.65, outer ysep=-2 pt, baseline=-.5ex,shorten >=.1pt,node
  distance=18mm, semithick,auto, every state/.style={fill=texto
    ,draw=texto,inner sep=.4mm,text=black,minimum size=0},
  accepting/.style={fill=gray!50!black,text=white},
  initial/.style={white,text=black} ]{ \node[state,minimum
    size=0,circle, fill=white, draw=black] (A11) at (0,0) {}; }$, and
those of $\bar\phi_{p_1p_2p_3}$ at $\tikz[scale=.65, outer ysep=-2 pt,
  baseline=-.5ex,shorten >=.1pt,node distance=18mm, semithick,auto,
  every state/.style={fill=texto ,draw=texto,inner
    sep=.4mm,text=black,minimum size=0},
  accepting/.style={fill=gray!50!black,text=white},
  initial/.style={white,text=black} ]{ \node[state,minimum
    size=0,circle , draw,fill=black] (A11) at (0,0) {}; }$, like
  \begin{equation}
 \phi_{a_1a_2a_3} \mapsto 
\tikz[scale=.65,  outer ysep=-2 pt,
    baseline=-.5ex,shorten >=.1pt,node distance=18mm, 
    semithick,auto,
    every state/.style={fill=texto ,draw=texto,inner sep=.4mm,text=black,minimum size=0},
    accepting/.style={fill=gray!50!black,text=white},
    initial/.style={white,text=black}
]{ 
\node[state,minimum size=0,circle, fill=white,  draw=black] (A11) at (0,0) {}; 
\path [-stealth]
(A11)  edge[bend right=50] node[above] {\tiny $a_1$} (0:1) 
(A11)  edge[bend right=50] node  {\tiny $a_3$} (120:1)  
(A11)  edge[bend right=50] node  {\tiny $a_2$} (240:1) ;
}\, ,
\qquad \qquad\bar\phi_{p_1p_2p_3} \mapsto
\tikz[scale=.65,  outer ysep=-2 pt,
    baseline=-.5ex,shorten >=.1pt,node distance=18mm, 
    semithick,auto,
    every state/.style={fill=texto ,draw=texto,inner sep=.4mm,text=black,minimum size=0},
    accepting/.style={fill=gray!50!black,text=white},
    initial/.style={white,text=black}
]{ \node[state,minimum size=0,circle , draw,fill=black] (A11) at 
(0,0) {};
\path [-stealth reversed]
(A11)  edge[bend right=50] node[above] {\tiny $p_1$} (0:1) 
(A11)  edge[bend right=50] node  {\tiny $p_2$} (120:1)  
(A11)  edge[bend right=50] node  {\tiny $p_3$} (240:1) ;
}\, 
\label{vts}
 \end{equation}
 To the $\delta$'s one associates numbered, or in the parlance \textit{colored},
 strands. Then the  $i$-colored strand
$\tikz
[scale=.65,  outer ysep=-2 pt,
    baseline=-.5ex,shorten >=.1pt,node distance=18mm, 
    semithick,auto,
    every state/.style={fill=texto ,draw=texto,inner sep=.4mm,text=black,minimum size=0},
    accepting/.style={fill=gray!50!black,text=white},
    initial/.style={white,text=black}
]{
\node[ black] at (0.5,-.35) {\tiny $i$};
\draw[stealth reversed-stealth, black] (0,0) -- (1,0);
}$ for $\delta_{a_ip_i}$ joins 
$a_i$ and $p_i$ in the vertices \eqref{vts}, so that, for instance
eq. \eqref{v2explicito} becomes 
\begin{equation}
\Tr_{\V_2}(\phi,\bar\phi)= 
  \hspace*{.2cm} \lambda\cdot\,
 \raisebox{-.45\height}{
  \includegraphics[height=2.1cm]{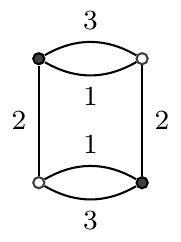}}
\end{equation}

This completes the comments on notation of interaction vertices; now
we address the corresponding notation of the Feynman graphs. The
partition function reads\footnote{ Actually the partition function has
  a factor $N^{D-1}$ ($N=|I_a|$ for $a=1,\ldots,D$) in front of the action $S[\phi,\bar\phi]$ in a
  more realistic scenario for a rank-$D$ model, and the measure should be rescaled accordingly. 
  Here we work with the graph structure of
  the theory and thus the factor can be restored anytime. See Remark
  \ref{being_melon} below.}
\begin{equation} \label{measure}
 Z[J,\bar J]=\frac{\int\Df[\phi, \bar\phi] \,\ee^{\Tr{(\bar J\phi)}+
\Tr{(\bar\phi J)}-S[\phi,\bar\phi]}}{\int\Df[\phi, \bar\phi]\,\ee^{-S[\phi,\bar\phi]}}\,, \quad 
 \Df[\phi, \bar\phi]:= \prod\limits_{\mathbf {a} \in
I_1\times I_2\times I_3} \frac{\dif\varphi_{\mathbf a} \dif\bar\phi_{\mathbf{a}}}{2\pi \ii} \,.
\end{equation}
The quantity $\dif{\mu}(\varphi,\bar\phi)= \Df[\phi, \bar\phi]\exp({-S_0[\varphi,\bar\phi]})$
defines a Gau\ss ian measure. Parenthetically, $\dif\mu$ and its perturbations, as other 
quantities can be formally studied in probability 
(see \cite{universality}, where a tensor version of the Dyson-Wigner law is 
obtained) but that is beyond our aim in this paper.  \par 
The perturbative expansion yields the Wick's contraction of products of
powers $[\Tr_{\mathcal{ B}_{\alpha_i}}(\phi,\bar\phi)]^{n_i}$ of the interaction vertices, 
i.e. all different fully Wick-contracted
terms obtained of the integrals of
\[
[\Tr_{\mathcal{ B}_{\alpha_1}}(\phi,\bar\phi)]^{n_1} [\Tr_{
\mathcal{ B}_{\alpha_2}}(\phi,\bar\phi)]^{n_2}\cdots [\Tr_{\mathcal{ B}_{\alpha_p}}(\phi,\bar\phi)]^{n_p},
\quad (n_i\in\Z_{>0}).
\]
The corresponding Feynman diagrams are $(3+1)$-colored graphs (Sec.
\ref{bubblehomology}).  It is illustrative to see how one arrives to
that result departing from the stranded representation.  First, we
associate with the propagator that contracts a field $\phi$ with a
$\bar\phi$, three parallel colored dotted lines:
\begin{equation} 
\raisebox{-.35\height}{
\tikz[scale=.65,  outer ysep=-2 pt,
    baseline=-.5ex,shorten >=.1pt,node distance=18mm, 
    semithick,auto,
    every state/.style={fill=texto ,draw=texto,inner sep=.4mm,text=black,minimum size=0},
    accepting/.style={fill=gray!50!black,text=white},
    initial/.style={white,text=black}
]{ 
\node[state,minimum size=0,circle, fill=white,  draw=green!58!black] (A11) at (0,0) {}; 
 \node[state,minimum size=0,circle, fill=white, draw=red] (A12) at (0,.4) {}; 
 \node[state,minimum size=0,circle, fill=white, draw=blue] (A13) at (0,.8) {};
\node[state,minimum size=0,circle , draw,fill=blue!40] (q13) at 
(3,.8) {}; 
 \node[state,minimum size=0,circle, draw,fill=red!50] (q12) at (3,.4) {}; 
 \node[state,minimum size=0,circle,draw,fill=green!40!lightgray] (q11) at (3,0) {};
 \foreach \c/\j in {green!58!black/1, red/2, blue/3}{
 \draw[dashed,\c] (A1\j) -- (q1\j);
 }
 } }\, \mbox{ = } \raisebox{-.30\height}{
 \tikz[scale=.65,  outer ysep=-2 pt,
    baseline=-.5ex,shorten >=.1pt,node distance=3mm, 
    semithick,auto,
    every state/.style={fill=texto ,draw=texto,inner sep=.4mm,text=black,minimum size=0},
    accepting/.style={fill=gray!50!black,text=white},
    initial/.style={white,text=black}
]{ 
\node[state,minimum size=0,circle, fill=white,  draw=gray] (A11) at (0,0) {}; 
 \node[state,minimum size=0,circle, fill=white, draw=gray] (A12) at (0,.4) {}; 
 \node[state,minimum size=0,circle, fill=white, draw=gray] (A13) at (0,.8) {};
\node[state,minimum size=0,circle , draw,fill=black] (q13) at 
(3,.8) {}; 
 \node[state,minimum size=0,circle, draw,fill=black] (q12) at (3,.4) {}; 
 \node[state,minimum size=0,circle,draw,fill=black] (q11) at (3,0) {};
 \foreach \c/\j in {black/1, black/2, black/3}{
 \draw[dashed,\c] (A1\j) -- (q1\j)
  ;
  \node[left of=A1\j] {\tiny \j};
 \node[right of=q1\j] {\tiny \j};
 }
 }}  
 \label{prop}
\end{equation}
and the corresponding Wick's contraction, 
as usual 
in the CTM-{literature}, with the $0$-color, 
$ \begin{tikzpicture}[scale=.75,
    baseline=-.5ex,shorten >=.1pt,node distance=18mm, 
    semithick,auto,
    every state/.style={fill=texto ,draw=texto,inner sep=.4mm,text=black,minimum size=0},
    accepting/.style={fill=gray!50!black,text=white},
    initial/.style={white,text=black}
]
 \node[state,minimum size=0,circle, draw,fill=white] (A) at (0,0) {}; 
 \node[state,minimum size=1,circle, draw,black,fill=texto ] (B) at (2,0) {}; 
 \path[-] 
           (A) edge[dashed, bend left=30] node[below] {\footnotesize $0$} (B);
 \end{tikzpicture}
$, which is the simplified version of the stranded
representation in eq. \eqref{prop}.
In that notation, a perturbative expansion 
in the \textit{the $(\phi_{D=3}^4)$-theory}, whose 
action is defined by the sum of the vertices
\begin{align} 
\label{vertices}
\mathcal V_1  &=\lambda \cdot \,\,  \raisebox{-.475\height}{\includegraphics[width=2cm]{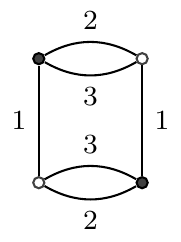}},   
& & \mathcal{V}_2 \hspace{- .63cm} 
&= \lambda\cdot \,\, \raisebox{-.475\height}{\includegraphics[width=2cm]{gfx/vtwo}}\, , & & \mathcal{V}_3 
&= \lambda \cdot \,\,\raisebox{-.475\height}{\includegraphics[width=2cm]{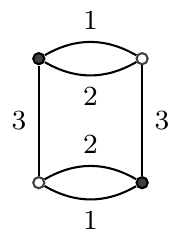}}
\, ,
\end{align}
generates $4$-colored bipartite regular graphs in colors
$\{0,1,2,3\}$ by Wick-contracting these vertices.  Among the, say,
fourth-order terms like $\int \dif \mu(\phi ,\bar \phi) (\V_1 \V_2
[\V_3]^2)$, the following is an example of vacuum graph contribution:
 \begin{equation}\hspace{-.5cm} 
\raisebox{-0.50\height}{\includegraphics[width=6.3cm]{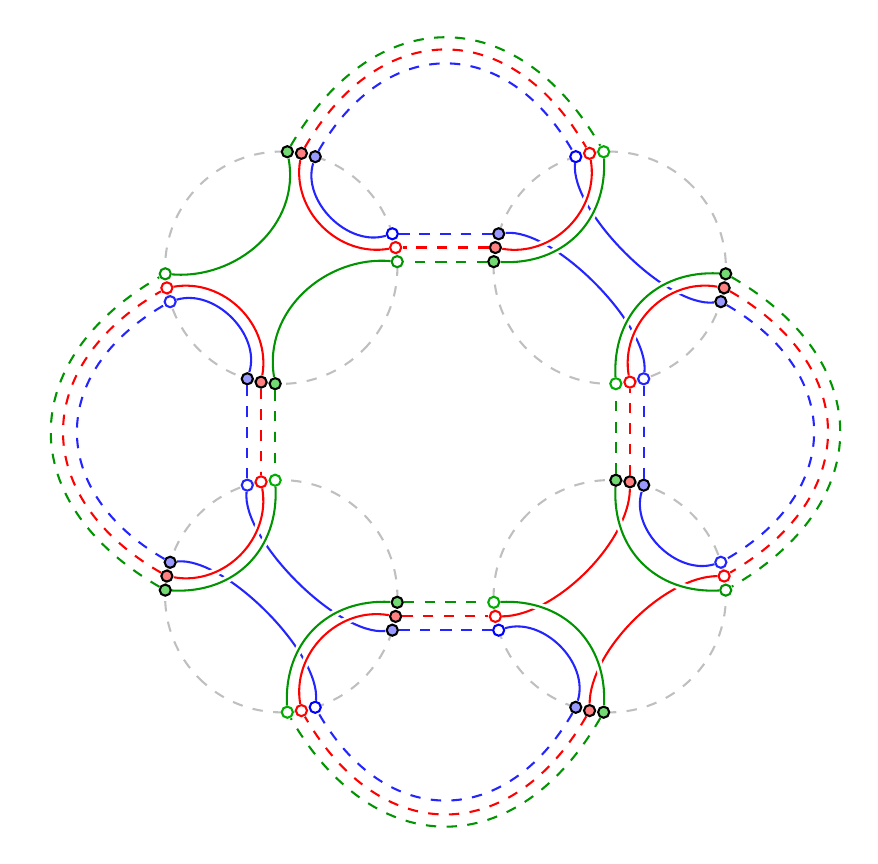}}
\mbox{ or } \raisebox{-0.50\height}{
\includegraphics[width=6cm]{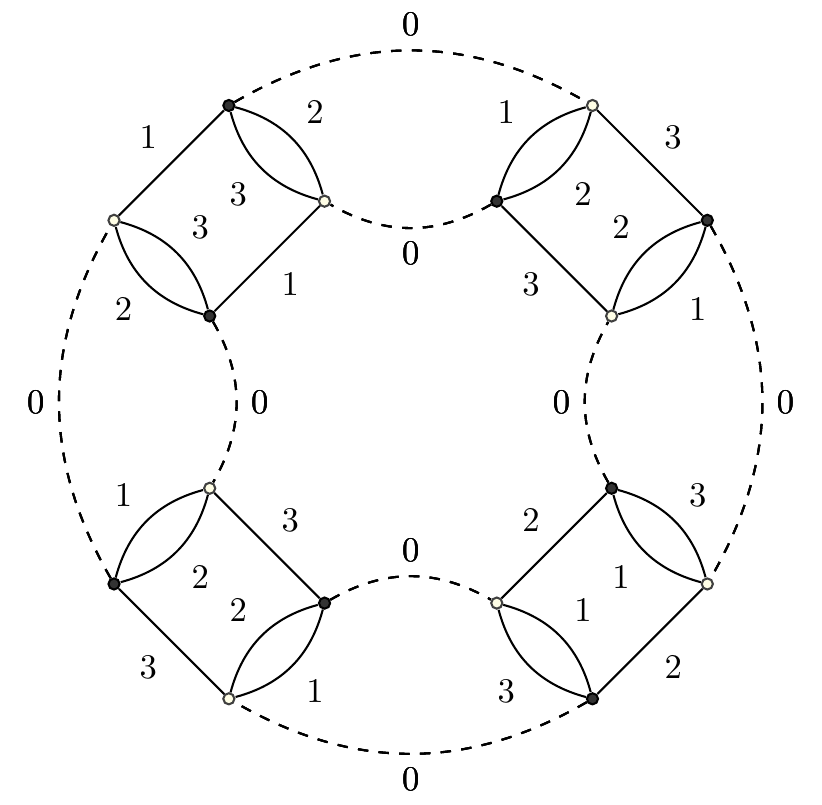}  \label{exampleforbubbles}} 
\lim_{\mtr{low\, energy}}\to \raisebox{-.4\height}{\includegraphics[width=1.5cm]{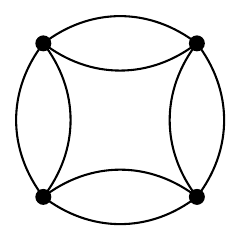}}
\end{equation}
 Both (big) graphs have the same information and are the `tensor-model' 
 version of the rightmost $\phi^4$-scalar field diagram, 
seen with higher energy resolution.
\begin{rem}\label{thm:col_uncol}
  The representations above have a slightly deceiving terminology.
  The graph $\G$ in the left of \eqref{exampleforbubbles} is called
  \textit{uncolored graph} \cite[Sec. 1.3]{wgauge}; the graph in the
  right is denoted by $\G_{\mtr{color}}$ or $\G_{\mtr{c}}$ and is the
  colored version of $\G$.  These terminology is also well-6explained
  in \cite[Def. 1 and Fig. 5]{4renorm}.  Sometimes, avoiding the
  stranded representation, $\G$ is represented as $\G_{\mtr{c}}$ with
  lines that are fainted if they are $c$-colored ($c\neq0$), but we
  refrain from doing so.  In the initial version of colored tensors a
  collection of $D$ random tensors, rather than a single one as here,
  was considered. Then $D-1$ of them were integrated out, and what one
  remains with is an `uncolored' \cite{uncoloring} tensor model with
  another (effective) action, which however, had still the colors
  encoded in their indices.  The difference is that the geometric
  realization for $\G_{\mtr{color}}$ includes a \textit{face} for
  each loop in $2$ arbitrary colors---either $0c$ or $cd$, for $c,
  d=1,\dots,D$ ($c\neq d$). For the uncolored version, one considers
  the loops of colors $0c$; the equivalence of notations, as in
  \eqref{exampleforbubbles}, explains why. We shall no longer use the
  stranded representation and prefer the `colored one', for which,
  however, we drop the label in $\G_{\mtr{color}}$ and use directly
  $\G$ instead.
\end{rem}

\subsection{Homology of colored graphs}
\label{bubblehomology} 

\begin{defn}
By a $D$-\textit{colored graph}, $\G$,  one understands here 
a graph $\G=(\G\hp 0,\G\hp1)$ with the following properties:
\begin{itemize}
\item[i)] \textit{Bipartiteness:}
The \textit{vertex set} is finite $\G^{(0)}$ and bipartite:
$\G^{(0)}=\G^{(0)}_{\mathrm{w}\vphantom{b}}\cupdot\G^{(0)}_{\mathrm{b}}$, where
$\G^{(0)}_{\mathrm{w\vphantom{b}}}$ are the \textit{white}, 
and $\G^{(0)}_{\mathrm{b}}$ the \textit{black}
vertices. Further, 
any $e\in \G\hp1$ is attached to precisely
one black vertex $w$ and one white vertex,
$b$, which we write $s(e)=b, t(e)=w$ or just
$e=\overline{bw}$.

\item[ii)] \textit{Regular coloring:} The \textit{edge set} is
  partitioned as $\G\hp1=\cupdot_{c=1}^D\G\hp1_c$.  The elements of
  $\G_c\hp1$ are said to have color $c$; this coloration is
  \textit{regular}, i.e. there are $D$ edges incident to any vertex
  having different colors.
\end{itemize}
We denote by $\Grph{D}$ 
the set of connected $D$-colored graphs and 
by $\amalg \Grph{D}$ the set of
disconnected $D$-colored graphs. Graphs in either $\Grph{D}$ or $\amalg \Grph{D}$ are also 
 called \textit{closed}, in contrast to open graphs 
 defined bellow (Sec. \ref{sec:BoundGraph}). A $p$-\textit{bubble} $\B$ of $\G
\in \Grph{D}$ is a connected subgraph
of $\G$ with edges in $p$ fixed colors ($1\leq p \leq D$), that is a 
subgraph $\B$ of $\G$ in $\Grph{p}$ (the $p$ colors
being a subset of $\{1,\ldots, D\}$). The set 
of $p$-bubbles of $\G$ is denoted $\G\hp p$. This 
is consistent with $\G\hp0$ and $\G\hp1$ being
the vertex and edge sets, respectively.
\end{defn}

\begin{example}
  Let $\G$ be the Feynman graph \eqref{exampleforbubbles}.  Then
  $\G\in \Grph{3+1}$, $\G$ having colors $\{0,1,2,3\}$.  The number of
  $0$-bubbles (indexed by $\G^{(0)}$) is 16; it also has $32$
  $1$-bubbles, $8$ of each color $c=0,\ldots,3$ (indexed by
  $\G^{(1)}$); there are $24$ $2$-bubbles sitting in $\G$: 3 bubbles
  of colors $\{01\}$ and $\{02\}$ each; 2 bubbles of colors $\{03\}$,
  6 bubbles with colors $\{12\}$ and 5 with colors $\{13\}$ and
  $\{23\}$ each.  The eight $3$-bubbles are drawn here:

\begin{align*}
\B^{012}_{\left\{\substack{\text{upper conn. component}\\\text{lower conn. component}}\right\}}&=\raisebox{-0.5\height}{
\includegraphics[width=4.2cm]{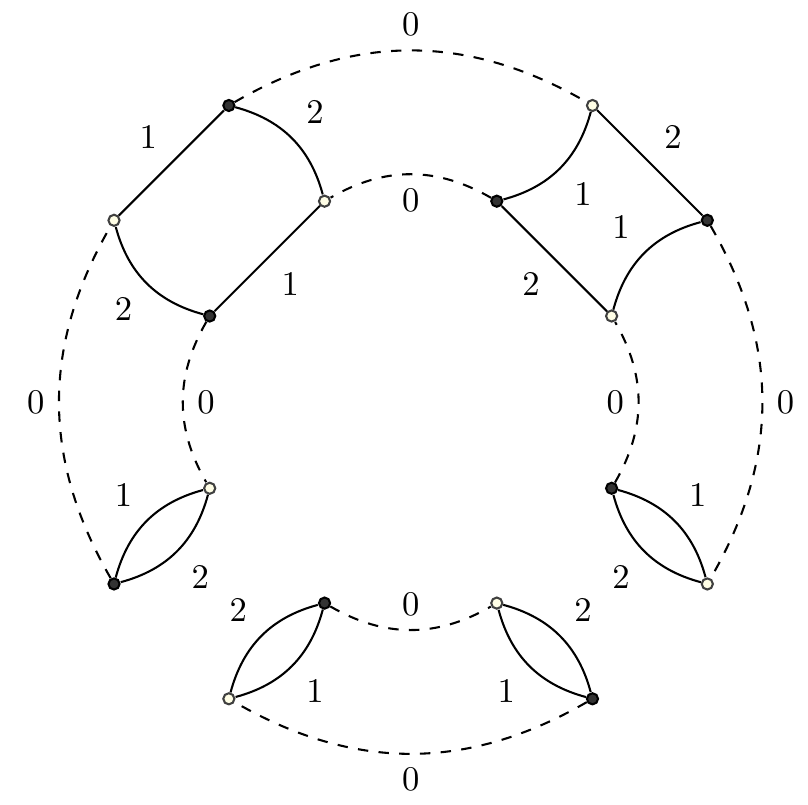}} & &\B^{013}=\raisebox{-0.5\height}{
\includegraphics[width=4.2cm]{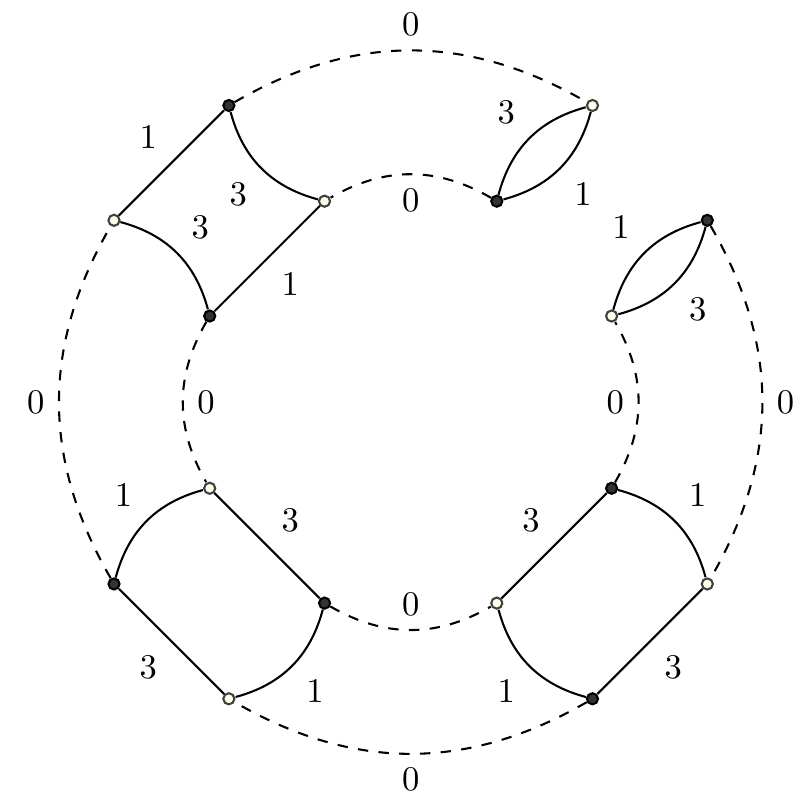}} \\
%                                                  \end{align*}
%   \begin{align*}
\B^{023}&=\raisebox{-0.5\height}{
\includegraphics[width=4.2cm]{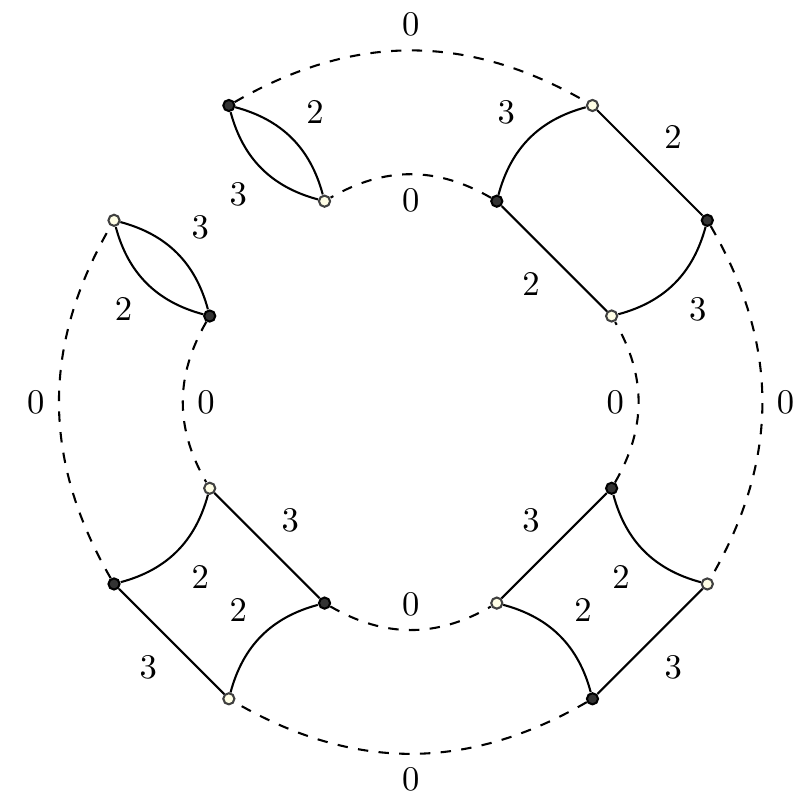}} && \qquad \begin{matrix}
                                                   \B^{123}_{1}=\V_1, & &  \B^{123}_{2}=\V_2,  \\
                                                    & \\
                                                   \B^{123}_{3}=\V_3, &  &\B^{123}_{4}=\V_3. 
                                                 \end{matrix} 
\end{align*}
 
\end{example}

The theory of homology for the Feynman
graphs of colored tensor models has been defined by Gur\u au
\cite{Gurau:2009tw} and, initially, it was referred 
to as \textit{bubble-homology}. 
The term \textit{colored homology} is also used. 
One defines the \textit{chain complex} of the graph $\G$ as the
collection of groups
$C_p(\G;\Z)=\mathrm{span}_\Z\{\B \,|\, \B \mbox{ is a $p$-bubble of }\G\}=
\mathrm{span}_\Z(\G\hp p)$ 
if $p=0,\ldots,D-1$, and $C_p(\G;\Z)=0$ otherwise. 
(Hence, in the case of the example above
$C_0(\G)=\Z^{16},\,\,C_1(\G)=\Z^{32},\,\,C_2(\G)=\Z^{24},\,\,C_3(\G)=\Z^{8}.$)
Any bubble in the generating set of $C_p(\G;\Z)$ has then the form
$\B^{I}_{\V}$ for $I\subset \{0,\ldots, D\}$, $I=(i_1,\ldots,i_p)$ fully 
ordered ($i_\alpha< i_\beta$ if $\alpha<\beta$) 
and $\V$ some vertex or number determining
the connected component. The \textit{boundary map} is
\[
\partial_p(\B^{(i_1,\ldots,i_p)}_{\V})=\sum_{q=1}^p(-1)^{q+1}
\sum\limits_{\mathcal W\subset \V} 
\B^{i_1\ldots \widehat{i_q}\ldots i_p}_{\mathcal W} \where (p\geq 2).
\]
The inner sum is performed over all the vertex-subsets $\mathcal W$ of $ \V$
with colors $i_1\ldots \widehat{i_q}\ldots i_p$.
For arbitrary $p$ one writes then:
\begin{align} \label{Rand_Kettenkomplex}
\partial_p(\B^{I^p}_{\V})=
\begin{cases} 0 & \mbox{ if } p=0, \\
 v-\bar v & \mbox{ if } p=1,\,  \mbox{and } \,\bar v=t(\B^{I}_\V),v=s(\B^{I}_{\V})\\ 
 \sum_{q=1}^p(-1)^{q+1} \sum_{\mathcal W_{\widehat{i_q}}} 
\B^{i_1\ldots \widehat{i_q}\ldots i_p}_{\mathcal{W}_{\widehat{i_q}}} & \mbox{ if }  p\geq 2
\end{cases}
\end{align}
and the restriction to $\mathcal W_{\widehat{i_q}} \subset \V $
on the sum is implicit. Thus, one orients the edges $e$ from the white
(sources $s(e)$) into the black vertices (targets $t(e)$).
\begin{defn} \label{thm:echar} The \textit{bubble homology} $H_\star(\G)$ of a 
colored graph $\G$ is the homology of the chain complex $(C_\star(\G),\partial_\star)$. 
The \textit{Euler characteristic} $\chi(\G)$ of a colored graph $\G$
is $\sum_q(-1)^{q}\,\mathrm{rnk} H_q(\G)$.\\
\end{defn}
Examples of the homology of graphs and 
the respective Euler characteristic follow.
We also refer to Appendices \ref{appA} and \ref{app_ribbons},
for more detailed computations.

%  \hyperref[fig:]{\ref*{fig: } \textsc{(a)}}

\subsection{Jackets and degree-computations} 

\begin{defn}(Jackets and degree.)
Let $\G$ be a $(D+1)$-colored bipartite graph.  Each cycle
$\tau=(\ell_0\ldots\ell_{D})\in \Sym_{D+1}$, defines a graph $\J_\tau$
called \textit{jacket} as follows: $\J_\tau$ has the same sets of
vertices and edges as $\G$,
\[ \J_\tau^{(0)}:=\G^{(0)}, \qquad \J_\tau^{(1)}:=\G^{(1)},\]
but its faces are those faces of $\G$ (i.e. two-bubbles) that have
colors $(\ell_{i+1}\ell_0)$ or $(\ell_i\ell_{i+1})$ with
$i=0,\ldots,D$:
\[ \J_{\tau}^{(2)}=\{f\in\G^{(2)} | \,f \mbox{ has colors }(\tau^q(0),\tau^{q+1}(0)), q\in \Z\}.\]
Here $\tau^q$ stands for $\tau\circ\ldots \circ\tau$ applied ($q$
times) to the color $0$.  By definition, $\J^{(k)}_\tau=\emptyset$
for $k>2$, so that jackets are ribbon graphs.  Since $\tau$ and
$\tau\inv$ lead to the same face-sets,
$\J_\tau^{(2)}=\J_{\tau\inv}^{(2)}$, those cycles are considered
equivalent. Hence $\G$ has $D!/2$ jackets. By computing their 
bubble-homology one finds, for certain non-negative integer $g_{\mathcal J}$, 
\[
H_q(\J)=\begin{cases}
\Z & \mbox{if}  \,\,q=0,2, \\
\Z^{2g_{\J}} & \mbox{if} \,\, q=1,\\
0 & \mbox{if} \,\,q>2.
\end{cases}
\]
i.e. the Euler characteristic of the geometric realization of $\J$ is
$2-2g_\J$. 
\end{defn}
\begin{defn} \label{thm:degree}
Given a closed graph $\G\in\Grph{D+1}$ ($D\geq 2$), its 
\textit{Gur\u{a}u's degree} $\omega(\G)$ is defined
as the sum of all the genera of the jackets of $\G$:
\[\omega(\G)=\sum_{\J\subset\G} g_\J.\]
A graph $\G$ with $\omega(\G)=0$ is called \textit{melon.}
\end{defn}

To fully understand the concept of jacket
the next brief examples might help.

\begin{example}\label{thm:exampleribbon} 
  The case $D=2$. There, ($2+1$)-colored graphs have a single jacket,
  the ribbon graph itself. The degree $\omega(\mathcal{R})$ is
  therefore precisely the genus $g(\mathcal{R})$. This will allow to
  treat matrix models as a rank-$2$ tensor model, as specified below
  \ref{matrix_as_tensor}.  Consider the next graphs
\begin{equation} \label{toroesferaribbon}
 \mathcal{R}_0=\hspace{-.7cm}\raisebox{-0.48\height}{\includegraphics[width=6.3cm]{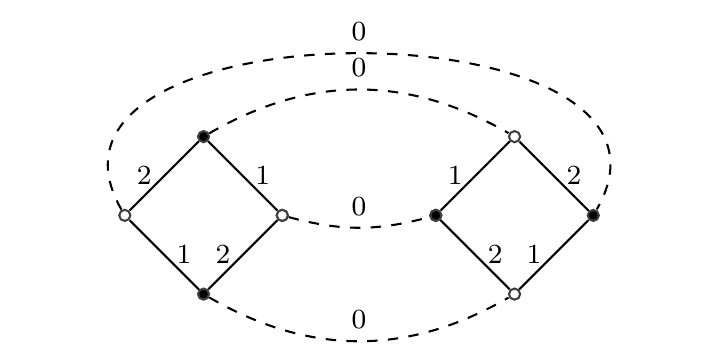}}
 \!\!\!\!\!\ \!\!\!\! \and 
 \mathcal{R}_1=\hspace{-1.4cm}\raisebox{-0.48\height}{\includegraphics[width=8.0cm]{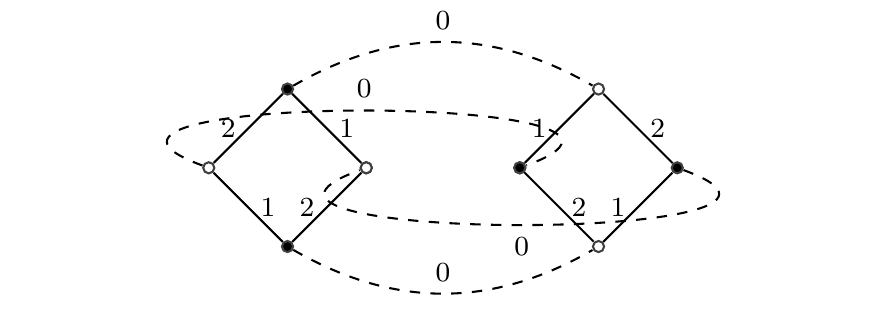}}
\hspace{-1.5cm}\end{equation}
arising in the $\mathcal O(\lambda^2)$-terms from
of the evaluation of the quartic rank-$2$ CTM, 
\[\int \Df{[\phi,\bar \phi]} \exp\big({-\Tr(\bar \phi\bar\phi)-\lambda\Tr((\bar\phi\phi)^2)}\big), \]
where the $0$-colored lines are Wick's contractions.  The homology of
$\R_0$ can be computed, but one can directly use its (dual) cell complex 
$\Delta({\mathcal R_0})$: the vertices and edges of the graph
are the $0$-cells and $1$-cells of $\Delta({\mathcal R_0})$,
respectively; to any bicolored path, one glues a $2$-cell. Then
$\Delta({\mathcal R_0})$ turns out to be a sphere.  The graph $\mathcal
R_1$ requires a slightly more detailed computation (see \ref{homtorus}
in App. \ref{appA}). From its homology $H_0(\mathcal{R}_1)=\Z ,
H_1(\mathcal{R}_1)=\Z^2, H_2(\mathcal{R}_1)=\Z$, one infers that
$\mathcal{R}_1$ is a (dual triangulation of a) torus.
\end{example}

% ------
% ------
% ------

% ------
% ------
% ------

\begin{example}(When is a graph non-melonic?) There are at least four
criteria to test non-melonicity: \label{thm:criteria}
\begin{itemize}
\item[(i)] Find a non-melonic subbubble $\B^{\hat D}_{(\kappa_0)}$ and
  use $\omega(\G)\geq D \sum_\kappa \omega(\B^{\hat D}_{(\kappa)})$
  (see \cite[Lemma 1]{universality} or \cite[Prop. 2]{uncoloring}),
  where $\kappa$ indexes all $3$-bubbles with colors $\hat D$.
  \item[(ii)] Since melonic graphs can be shown to be by necessity dual to 
 spheres \cite[Lemma 4]{Nexpansion}, one can compute the homology $H_\star(\G)$
 and show that it is not isomorphic to the (say cell-) homology
 $H_{\star}^{\mtr{cell}}(\mathbb{S}^3)$ of a sphere.  
\item[(iii)] Find a non-spherical jacket of $\G$. 
 \item[(iv)] Face-counting (cf. \cite[Prop. 1]{uncoloring} and \cite{us}).
\end{itemize}
Contrary to matrix models, tensor models turn out to encode more than
topology. That can be noted in the following example. Let
$\G$ be the following necklace-graph\footnote{The name
`necklace' has been borrowed from \cite{Carrozza_review}.
} in four colors: 
\[\includegraphics[width=2.2cm]{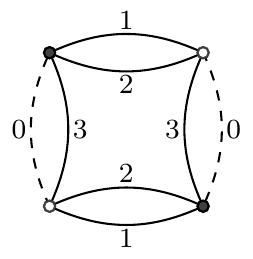} \]
For this graph, (i) above does not help. Since all four 3-bubbles
$\B^{\hat 0},\ldots,\B^{\hat 3}$ the same as the vertices $\V_i$. Thus
those bubbles are melonic, and the lower bound given by (i) is
trivially satisfied --- by definition, $\omega(\G)$ is already a
non-negative integer.  To begin
with (ii) the chain complex is
\[
\includegraphics[height=.66cm]{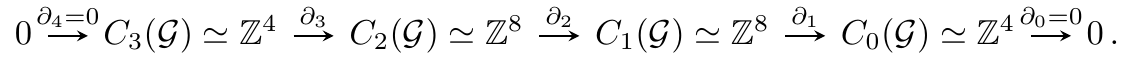}
\]
Taking homology yields (see Ex. \ref{homonecklace},  App. \ref{appA})
\[H_3(\G)=\Z,\qquad H_2(\G)=0,\qquad H_1(\G)=0, \qquad H_0(\G)=\Z.\]
That method still does not help to know whether the graph is melonic,
since $\G$ has the same homology as $\mathbb{S}^3$. Proceeding with
(iii), one has for the jackets $\J_\tau,\J_\pi,\J_\sigma$
corresponding to $\tau=(0123), \pi=(0213), \si=(0132) \in\Sym_4$, the
following groups:
\begin{align*}
H_2(\J_\tau)=\Z\,, & & H_1(\J_\tau)=&0\,, & H_0(\J_\tau)&=\Z
\,, \\
H_2(\J_\pi)=\Z\,, & & H_1(\J_\pi)=&\Z^2, & H_0(\J_\pi)&=\Z\, , \\
H_2(\J_\si)=\Z\,, & & H_1(\J_\si)=&0\,, & H_0(\J_\si)&=\Z\, . 
\end{align*}

Therefore $\omega(\G)=1$. Of course this is consistent with the
face-counting formula \cite[eq. 2.9]{critical} for $d$-colored graphs:
\[F_\G:=\#(\G^{(2)})= {d-1\choose 2}\cdot \frac{\#(\G^{(0)})}{2}+(d-1)-\frac{2\omega(\G)}{(d-2)!} \,,
\]
which in this case ($d=3+1= \#(\G^{(0)}),F_\G=8$) also yields
$\omega(\G)=1$. By Proposition 4.3 of \cite{GurauRyan} (namely,
four-colored graphs possessing a spherical jacket $J$, i.e. $g_J=0$,
are themselves spherical) one has a counterexample for the reciprocal
of `being a melon implies being dual to a sphere'.  In \cite{Ryan} it
has been found that the jackets represent embedded matrix theories in
the tensor theories.  Moreover the jackets are interpreted as the
surface along which a Heegaard splitting takes place.  Here we found
an example of this explicit splitting for the three-sphere.  The
genus-$0$ jackets $\J_\tau$ and $\J_\sigma$ correspond to the
genus-$0$ Heegaard splitting of $\mathbb{S}^3\subset \C^2$, that is,
the coordinates $(z_1,z_2)$ of $\mathbb{S}^3$ with $\Im (z_2)=0$,
i.e. $\mathbb S^2$ inside $\mathbb{S}^3$.  The jacket $\J_\pi$ is the
genus-$1$ Heegaard splitting of $\mathbb{S}^3$, the Clifford torus
$\T^2$, given by the locus $|z_1|=\sqrt{2}/2=|z_1|$.
\end{example}

\begin{rem}\label{being_melon}(On the importance of melons). 
So far we have considered the action functional as frozen, not flowing
with the energy scale, $N$.  This large integer is defined as follows.
One usually sets the Hilbert spaces $\H_c$ all to the fundamental
representation of $\mtr{U}(N)$.  One keeps in mind the color as an
artifact that forbids elements in one factor of $\mtr{U}(N)\times
\mtr{U}(N) \times \mtr{U}(N)$ to jump to another one. The action
(without sources) is then scaled as
\begin{equation} 
 Z_0=\int \Df[\phi, \bar\phi]\,\ee^{ -N^{D-1}S[\phi,\bar\phi]}\,.
\end{equation}
Melons are the dominating graphs, for 
the amplitude of a graph $\G$ is weighted \cite{gurauReview}
by Gru\u{a}u's degree as
\begin{equation}
 \A(\G) \sim N^{D-\frac{2 \omega(\G)}{(D-1)!}} \,.
 \label{weightN}
\end{equation}
This is the tensor version of the 
well-known genus-weighted amplitudes 
of ribbon graphs $\R$ in matrix models, 
$\A(\R)\sim N^{2-2g(\R)}$. Having 
ribbon graphs a single jacket, the latter 
formula is a particular case of formula
\eqref{weightN}, with $D=2$. 
\end{rem}

\subsection{The boundary graph}\label{sec:BoundGraph}

\begin{defn}
  A graph $\G$ is an \textit{open} $(D+1)$-colored graph (with colors
  $c=0,1,\ldots, D$) if its vertex-set is bipartite in the following
  two senses:
\begin{itemize}
 \item[(i)] As 
before, $\G^{(0)}=\G^{(0)}_{\mathrm{w}\vphantom{b}}\cupdot\G^{(0)}_{\mathrm{b}}$, where
$\G^{(0)}_{\mathrm{w\vphantom{b}}}$ are the \textit{white}, and $\G^{(0)}_{\mathrm{b}}$ the \textit{black}
vertices,
\item[(ii)] any vertex is either \textit{inner} or \textit{outer},
  $\G^{(0)}=\G^{(0)}_{\mathrm{inn}}\cupdot\G^{(0)}_{\mathrm{out}}$;
  further, the set $\G^{(0)}_{\mathrm{inn}}$ of inner vertices is
  regular with valence $D+1$ and outer vertices have valence $1$.
\end{itemize}
Additionally, the edge set $\G^{(1)}$ is $(D+1)$-\textit{colored} 
---that is $\G^{(1)}=\cupdot_{c=0}^D\G^{(1)}_c$,
where $ \G^{(1)}_c$ is the set of color-$c$
edges---
and satisfies the following:
\begin{itemize}
\item[(iii)] for each color $c$ and each inner vertex
  $v\in\G^{(0)}_{\mathrm{inn}}$, there is exactly one color-$c$ edge 
  $e\in \G^{(1)}_c$ attached to $v$,
 \item[(iv)] each external vertex is connected to the graph only by a
 color-$0$ edge.
\end{itemize}
Both the leaves of open graphs and the edges that are 
attached to them shall be referred to as \textit{external legs}.
Therefore alternative notations might include omission of the outer
vertices (`half lines') or their replacement by sources.  There, to
each non-contracted black (resp. white) vertex, a (tensorial) source $J$ (resp
$\bar J$ is attached). We also let 
\begin{equation}
 \Grph{D+1}\hp {\mathcal N}:= \big\{ \G \mbox{ open } (D+1) 
 \mbox{-colored }\big|\,\, \#\big(\G\hp 0_{\mathrm{out}}\big)=\mathcal N \big\}\,.
\label{eqn:open_graphs_legs}
\end{equation}
In the models we treat here, being the graphs bipartite, $\mtc N$ is
even. We allow $\mtc N=0$, setting $\Grph{D+1}\hp {0}:=\Grph{D+1}$. To
complete the notation, given an open graph $\G$ one can extract a
(generally non-regularly) colored graph $\mathrm{inn}(\G)$ defined by
\begin{equation}
 \mathrm{inn}(\G)\hp 0 = \G\hp 0_{\mathrm{inn}}   \qquad 
 \mathrm{inn}(\G)\hp 1 = \G\hp 1  \setminus \{\mbox{external legs of }\G\}\,.
\end{equation}
The graph $\mathrm{inn}(\G)$ is called \textit{amputated} graph. 
\end{defn}
For instance, the amputation of the the following open graph $\G\in\Grph{3+1}\hp{2}$ 
is represented at its right-side:
\begin{equation} \label{first_Ex_opengraph}
\!\!\!\G=\raisebox{-.38\height}{\includegraphics[height=2.4cm]{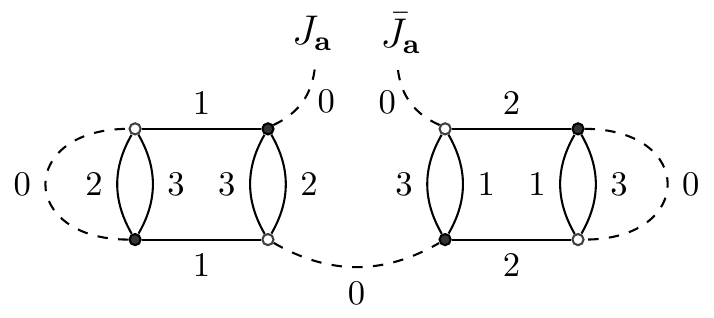}} \,\,\,\,\,\,\,
\mathrm{inn}(\G)=\raisebox{-.5\height}{\includegraphics[height=1.83cm]{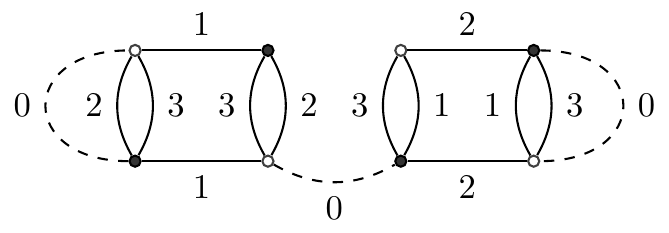}} 
\end{equation}
In the QFT-context, open graphs 
are not-fully Wick-contracted interaction vertices.
\begin{defn}
A \textit{colored tensor model} $V(\phi,\bar\phi)$ is determined by
an integer $D \geq 2$, called \textit{dimension}, which is given by
the rank of the tensors, and by an \textit{action}
\begin{equation}  %\label{potencial}
V(\phi,\bar\phi)=  \sum_{\B\in \Upsilon } \lambda_\B  \,\Tr_\B(\phi,\bar \phi)\, ,\nonumber
\end{equation} 
where $\Upsilon\subset \Grph{D}$ is a finite subset and $\lambda_\B\in \re$ for each $\B\in\Upsilon$.
The full action is then $S[\phi,\bar\phi]=\Tr_2(\bar\phi,\phi)+ V(\phi,\bar\phi)$.
The set of \textit{connected Feynman diagrams} 
of the \textit{model} $V(\phi,\bar\phi)$ is denoted by $\fey_D (V)$. It satisfies
\[
\fey_D (V)= \big\{ \G \in \cupdot _{k=0}^\infty 
\Grph{D+1}\hp{2k}\,\big|\, \mathrm{inn}(\G)^{\hat 0} \mbox{ has connected components in } \Upsilon\big\} \,.
\]
The graphs in $\Grph{D+1}\hp 0\cap \fey_D (V)=\Grph{D+1}\cap\fey_D (V):=\fey_D^{\mtr{v}}(V)$ 
are called \textit{vacuum graphs} of the model $V$. 
We write $\fey^\re_D(V)$ if the tensor field is real. 
\end{defn}

The boundary of a $(D+1)$-colored graph is a 
$D$-colored graph. It was defined in such a way that its associated
(pseudo)simplicial complex matches the 
boundary of the complex of the original graph (see \cite{GurauRyan} for this fact).  
Thus, the boundary of a graph is defined to be empty on 
vacuum (i.e. closed) graphs. Otherwise:

\begin{defn}\label{def:open}
Let $\G$ be an open Feynman diagram of a rank-$D$ colored
tensor field theory. The \textit{boundary graph} $\partial\G$
of $\G$ has by definition the 
following vertex and edges sets:
\[(\partial\G)^{(0)}=\G_{\mathrm{out}}\hp 0, \qquad (\partial\G)^{(1)}=
\cupdot _{c=1}^D (\partial\G)^{(1)}_c ,\with
(\partial\G)^{(1)}_c=\{(0c)\mbox{-colored paths in }\G\}\,. \]

Another conventions define $(\partial\G)^{(0)}$ as the set of external
lines. Since an external line goes to an outer vertex, both
definitions are equivalent.  The definition also says that the outer
vertices are the leaves of the graph. 
Two vertices of $\partial\G$ are connected by an $i$-colored 
edge if and only if there exists a $(0i)$-bicolored path between 
them in $\G$. 
\end{defn}

\begin{example}
  The graph $\G$ in eq. \eqref{first_Ex_opengraph} has two outer
  vertices, and by the previous construction, $\partial\G$ will have
  two vertices. For each $c=1,2,3$, an edge of color $c$ connects
  those, since there is a $(0c)$-bicolored path between the two outer
  vertices $\G$. Then $\partial\G$ is
$\raisebox{-.38\height}{
\includegraphics[height=4ex]{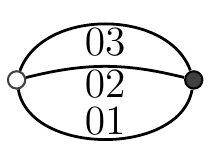}}
$
which does not differ (but apparently in labelling) from 
the graph of the propagator. A slightly 
more interesting example is 
\begin{equation} \label{eq:integraciondek33}
\raisebox{-0.45\height}{\includegraphics[height=4.1cm]{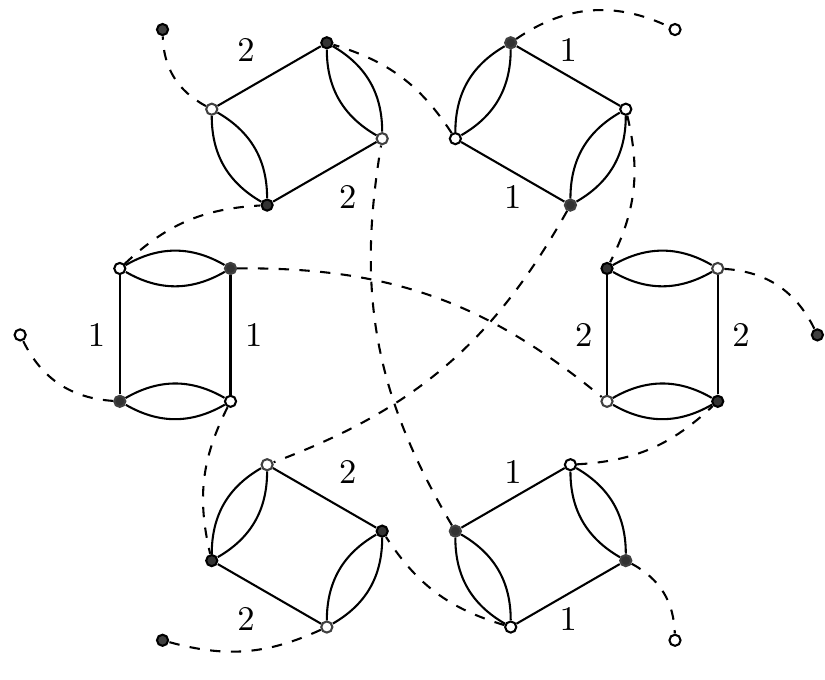}}
\quad \stackrel{\partial}{\mapsto } \quad\raisebox{-0.4\height}{
 \includegraphics[height=1.7cm]{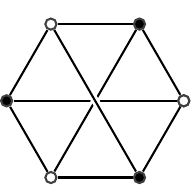}} 
\end{equation}
\end{example}
Boundary graph might be disconnected, as we will easily see
(in Lemma \ref{thm:sepa}, constructively).  \\

\subsection{The geometric realization of colored graphs}\label{reconstruction_from_graphs}

One can construct a colored triangulation $K(\G)$ of a compact
piecewise-linear manifold $|K(\G)|$ departing from a $(D+1)$-colored graphs $\G$
as we now describe. 
The simplicial (pseudo)complex $K(\G)$ is 
assembled as follows \cite{survey_cryst}:
\begin{itemize}
 \item[\scriptsize$\bullet$] for each $v\in \G\hp 0$, $K(\G)$ has a $D$-simplex $\sigma_v$
 \item[\scriptsize$\bullet$] one labels the vertices $\sigma_v$ by the colors $\{0,1,\ldots,D\}$
 \item[\scriptsize$\bullet$] for each color $c$ and each edge $e_c\in\G_c\hp1$ one identifies 
  the faces $\sigma_{s(e_c)}$ and $\sigma_{s(t_c)}$ that do not contain the color $c$ (i.e.
  the $(D-1)$-simplices that lie opposite to the vertex labelled by $c$)
 \end{itemize}
Because colored graphs allow multiple edges,
$D$-simplices might intersect at more than one face (whence the prefix \textit{pseudo}).
We write $|\Delta(\G)|$ for the manifold that 
$K(\G)$ triangulates, but we abuse on notation and abbreviate it as $\Delta(\G)$. 
We say that $\G$ \textit{represents} $\Delta(\G)$.  

\begin{rem}
 In the 2-dimensional case the cell complex associated to ribbon
 graphs and $3$-colored graphs exposed in Appendix \ref{app_ribbons}
 is the Poincar\'e dual of this $\Delta$-construct.  Since we
 basically analyze the Euler characteristic of graphs, this subtlety
 is not important.\\
  \end{rem} 
% In the late 70s and 80s colored triangulations 
% with a minimum of vertices, so called 
% called \textit{crystallizations} (see \ref{thm:crys}),
% played an important role in discrete geometry.

Colored graph theory also incarnates the cone $X\mapsto CX$ of a
topological space \cite{gurauReview}.  If $\B$ is a $D$-colored graph,
then one defines $C\B$ as the open $(D+1)$-colored graph with
$(C\B)^{(0)}_{\mathrm{inn}}=\B^{(0)}\times \{0\}$ and
$(C\B)^{(0)}_{\mathrm{out}}=\B^{(0)}\times \{1\}$ (a second copy of
the vertices). Moreover, if $v\in \B^{(0)}$ is white (resp. black)
then $(v,0)$ is white (resp. black) as well, but $(v,1)$ is black
(resp. white). 
The edges are defined by
\[(C\B)^{(1)}=\B^{(1)} \cupdot \{ \mbox{single color-$0$ edge
  between} (v,0) \mbox{ and } (v,1)\, | \,v \in \B^{(0)}\}\]
The cone is defined so that $\Delta(C\B)=C\Delta(\B)$. The relation
$\partial (C \B)=\B$ holds also for each graph.
\\

\subsection{Ribbon and $3$-colored graphs}
Ribbon graphs are also known as \textit{fat graphs}.  We choose mainly 
the definition of \cite{nlab} with a notation inspired by \cite{mulase}, but we will
really need only a subset of those graphs, which arises either as
Feynman diagrams in matrix models or as boundary-graphs of Feynman
graphs of rank-$3$ tensor field theories. In all generality, though:

\begin{defn} \cite{mulase} 
A \textit{ribbon graph} is a finite graph 
without isolated vertices nor leaves, together with a cyclic 
ordering of the set of half-edges at each vertex.
\end{defn}

The definition includes \cite{nlab}, implicitly, the following set of data and conditions:
\begin{itemize}
  \item[\scriptsize$\bullet$] two finite sets: the \textit{vertex-set} $\Rb\hp 0$, 
 and the  \textit{half-edges} set $\Rb\hp{1/2}$.
  \item[\scriptsize$\bullet$] a map $p:\Rb\hp {1/2} \to \Rb\hp 0$. The picture of `$p(h)=v$'
   is that the half-edge $h$ emanates from $v$.
 \item[\scriptsize$\bullet$] $n_v:=|p\inv(v)|$ is called the \textit{valence} of a vertex $v$. 
 Furthermore the condition $n_v>1$ is imposed. Thus
$\Rb$ has neither isolated vertices ($n_v\neq 0$) nor leaves ($n_v\neq 1$). 
 \item[\scriptsize$\bullet$] a cyclic orientation of $\Rb\hp{1/2}.$
 \item[\scriptsize$\bullet$] an involution $j$ on the set of half-edges; 
here $j(h)=h'$ means that $\{h,h'\}$ is 
a full edge ---if so then, of course, $j(h')=j^2(h)=h$.
Moreover, it is imposed that $j$ has no fixed point.
\end{itemize}
The usual graph notion in terms of vertices and (full) edges,
$\Rb=(\Rb\hp 0,\Rb\hp1)$, is then recovered by defining $\Rb\hp 1$,
the \textit{edge-set}, as the set of cycles of $j$ \cite{nlab}.  Without loss of
generality we write
\[\Rb\hp{1/2}=\{(v,\alpha): v\in \Rb\hp 0, \alpha=1,\ldots,n_v\} \and
p=\pr 1. \]
Before formally constructing the cell-complex for a ribbon graph in Appendix \ref{app_ribbons}, we
motivate in an informal vein their usual notation from the abstract
definition.  For a vertex $v\in \Rb\hp 0$ of valence $n_v$, the cyclic
ordering sees the following operation:
\begin{equation} \label{prohibida}
\raisebox{-.45\height}{\includegraphics[width=7.5cm]{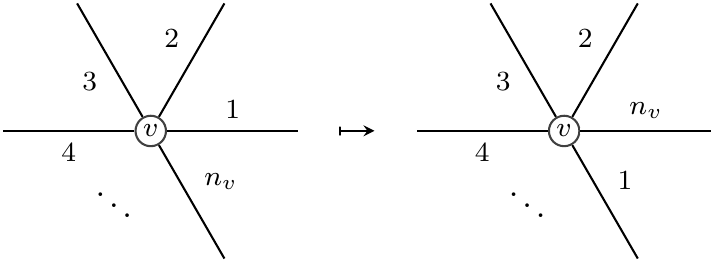}}
\end{equation}

\begin{figure}%[H] 
\begin{subfigure}{0.59\textwidth}
\includegraphics[width=8.9cm]
{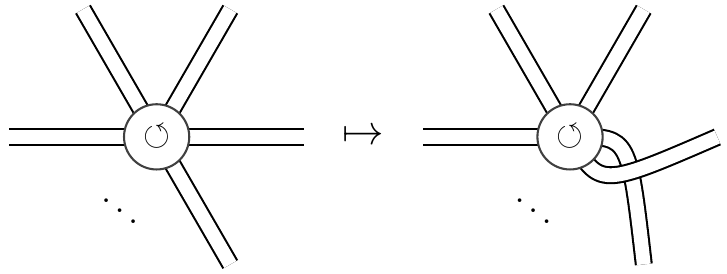}
\caption{   Cyclic ordering on the vertices
represented by a disk.} \label{fig:listonesA}
\end{subfigure}
\hspace*{\fill}  
\begin{subfigure}{0.29\textwidth}
\includegraphics[width=4.1cm]
{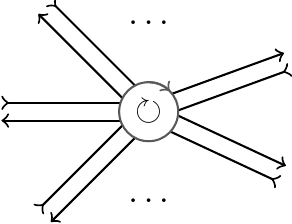}
\caption{   How the orientation of the edges is determined.} \label{fig:listonesB}
\end{subfigure}
\caption{    	
}  \label{fig:listones}
\end{figure}
In order to keep track of the order, the incidence relations are
usually graphically represented as follows: edges are
\textit{ribbons}, that is rectangles (topological disks $D^2$);
vertices are \textit{disks}. The map $j$ represents the attachment of
one side of a half-edge to a disk, as in
Fig. \hyperref[fig:listonesA]{1\textsc{ (a)}}, thus keeping track of
the operation \eqref{prohibida}.  The cyclic ordering of the vertex
determines an orientation on half-edges--rectangles as shown in
Fig. \hyperref[fig:listonesB]{1\textsc{ (b)}} and the ribbons should
be drawn taking into account the orientation on both ends. Moreover,
the ribbons do not intersect and the way they are attached to the
disks must respect the orientation. If we represent the graph on the
plane, mismatch of orientations is represented by lines
$\raisebox{-.2\height}{\includegraphics[height=0.48cm]{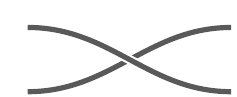}}$,
which do cross. Nevertheless, crucially, the graph can be drawn
without intersections on other surfaces.  The lowest-genus closed, orientable surface on
a ribbon graph $\Rb$ can be planarly drawn on is its geometric
realization, $\Sigma(\Rb)$ (see App. \ref{app_ribbons} for its construction).

\begin{defn} \label{genus_ribbon}
We write $\chi (\Rb)$ for the \textit{Euler characteristic} of the geometric
realization of $\Rb$, that is $\chi (\Rb)=\chi(\Sigma(\Rb))$.  In
turn, this also defines the \textit{genus} $g(\Rb)$ of $\Rb$.
\end{defn}

\begin{example}\label{thm:ribbons_concrete}
We illustrate the concepts in the last 
paragraph for the following simple
ribbon graphs: \vspace{-.21cm}
\begin{align*}
\raisebox{-.1\height}{
$\Rb =$ 
\raisebox{-.3\height}
{
\includegraphics{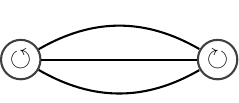}
},
$\,\,\,
\mtc Q=$
\raisebox{-.3\height}{ 
\includegraphics{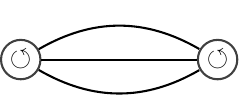}
}$,\,\,\,
\mtc W=$ 
\raisebox{-.23\height}{
\includegraphics[height=1.38cm]{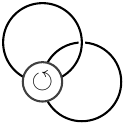} 
}.}
\end{align*}
Their ribbon representation is the following (thought of as 
filled vertices and ribbons):
\[
\Rb = 
\raisebox{-.45\height}
{
\includegraphics[height=1.5cm]{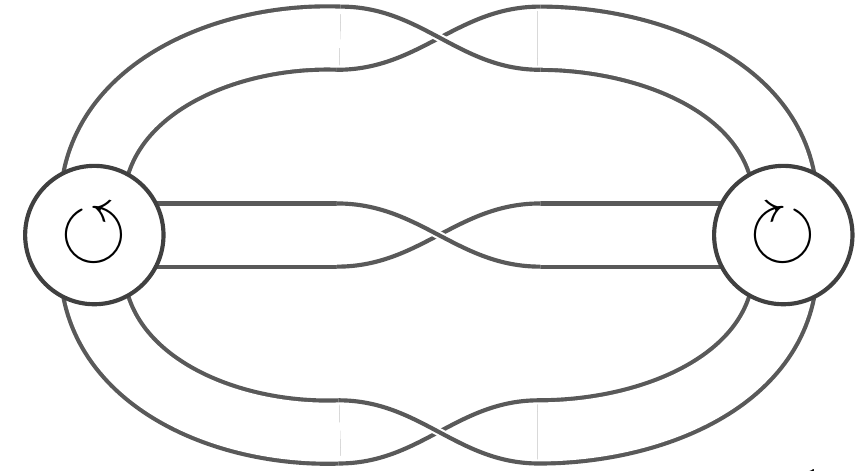}
},
\,\,\,
\mtc Q=
\raisebox{-.45\height}{ 
\includegraphics[height=1.6cm]{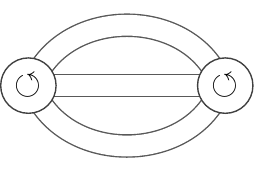}
},\,\,\,
\mtc W= 
\raisebox{-.4\height}{
\includegraphics[height=2.08cm]{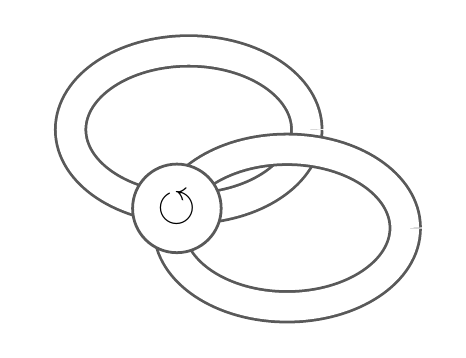} 
}.
\]
Thus, the fat graph $\Rb$ has only one boundary component, so
$\chi(\Sigma(\Rb))=0$. Thus $\Rb$ can only be planarly drawn on a
torus. Also $\mathcal W$ has genus $1$, as it has only one boundary
component, one vertex and two edges (ribbons). The graph $\mathcal Q$
has genus $0$.  The notation we will choose from now on is the
omission of the disks, usual in the physics literature, 
as well as disregarding 
crossings $\raisebox{-.2\height}{\includegraphics[height=0.40cm]{gfx/mismatch.pdf}}$. 
With that notation, $\mathcal W$ is shown in
graph \eqref{eqn:ribonfisico}. This does not affect the previously
defined quantities, because they are homotopy invariant.
\end{example}

\begin{lem}\label{thm:col_rib}
Regularly edge-$3$-colored, vertex-bipartite graphs
are ribbon graphs.
\end{lem}
\begin{proof}
  Let $\G=(\G\hp0,\G\hp1)$ be a $3$-colored graph. We exhibit the 
  ribbon graph structure of $\G$. The set of vertices of the ribbon graph is
  the same, $\G\hp0$. Define the set of half-edges
  $\G\hp{1/2}:=\G\hp{0}\times \Z_3$, where $\Z_3=\{1,2,3\}.$ The map
  $p:\G\hp{1/2}\to \G\hp0$ is the projection $p=\pr1$, which satisfies
  $n_v>1$ for each vertex $v\in\G\hp 0$, since $p\inv(v)=\Z_3$.  We
  let the cyclic order for white [resp. black] vertices be $(123)$
  [resp. $(321)$]. Finally, the involution $j$ on $\G\hp{1/2}$ is
  defined as follows: given $h=(v,c)\in\G\hp{1/2}$, let $e\in \G\hp 1$
  be the edge of color $c$ at $v$ (because of regularity and
  coloring, $e$ is uniquely determined) and $w$ the other vertex $e$
  is attached to. Then let $j(h):=(w,c)$. The map $j$ is an
  involution, since any two vertices can be connected only by one
  edge.
\end{proof}

The converse of the previous lemma does not hold. For instance,
consider the graph $\mathcal W$ in Example \ref{thm:ribbons_concrete}
(or in \ref{eqn:ribonfisico}). That ribbon graph is not bipartite,
since it is a graph of a real model,  $\mtc W \in \fey^{\re}_2( \phi ^4)$.  Regular colored
bipartite graphs in more colors can be also given the structure
of a ribbon graph, however the cell-attachment does not stop at
dimension $2$ (see App. \ref{app_ribbons}).  This explains why having exactly $3$ colors is
important.   
%% 
%%  = = = = = = = = = = = = = = = = = = = = = = = = = = = = = = = = = =  
%% 
%%  = = = = = = = = = = = = = = = = = = = = = = = = = = = = = = = = = =  
%% 
%%  = = = = = = = = = = = = = = = = = = = = = = = = = = = = = = = = = =  

\section{Graph-encoded surgery} \label{gsurgery} We develop elementary 
colored-graph-encoded surgery.  The aim of this concept is
twofold.  The physical motivation is to see that we can expand the
free energy $\log(Z[J,\bar J])$ of the model in sources indexed by
ribbon graphs, having as goal the Ward Identity of the
$\phi_3^4$-theory \cite{fullward}.  We will see here that this
expansion is optimal after the identification of those ribbon
boundary-graphs with closed, possibly disconnected Riemannian
surfaces.  The second aim, also for future work, is a macroscopic
realization of the theory. This surgery shall become useful as for
computing the space the final gluing of a large number of known 
`chunks of space' represents.  \par

An obstacle to perform this surgery is that one might have not enough
simplices; in that case, by removing a simplex (or more), the space
might fall apart into a topologically simpler one and information
about its topology would be lost.  In the same line of thought, there
are subtleties concerning disk excision of open graphs (manifolds with
boundary).  With graphs, by doing what one could naively call
`removing a disk' $\mathring D^n$ in the wrong place might not create
a boundary component $\mathbb S^{n-1}$ in the way one expects to do
so. This phenomenon is better illustrated by example. Here, in the
$\phi^4_{D=2}$ (matrix)-theory, consider the graph $\R_1$ of
Ex. \ref{thm:exampleribbon}.  By cutting two color-$0$ edges one
arrives to Figure \ref{fig:toro_fronteras}, when one realizes it as a
surface. But if one follows any of the two lines of the `boundary of the
ribbon' both connect the two boundary components of the surface.

\begin{figure}
\centering
\includegraphics[width=4.6cm]{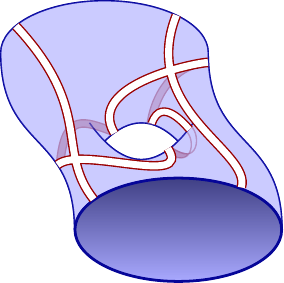} 
 \caption{   If one caps the two boundary components,
 one obtains $\R_1$ (Ex. \ref{thm:exampleribbon})  \label{fig:toro_fronteras} \normalsize}
\end{figure} 
\subsection{Colored graph surgery}

To prove our statements we need to see how to cut a $2$-disk, in
this realm, a two-bubble of a graph\footnote{ For instance, take the
  surface in Figure \ref{fig:toro_fronteras} and put a cap in one of
  the boundary components (join one ribbon). If one takes two copies
  of this and tries to glue them along the boundary, and then
  represent this back as a colored graph, one arrives at definition
  \ref{thm:surg_ribb}}.

\begin{defn}\label{thm:surg_ribb}
Let $\Rb$ and $\mathcal Q$ be (closed) $D$-colored graphs,
 $\Rb, \mathcal Q\in\Grph{D}$. 
Let $c$ be any color, and $e$ and $f$ be 
color-$c$ edges in $\Rb$ and $\mtc Q$, respectively, i.e.  
$e\in \Rb\hp 1_c$ and $f\in \mathcal Q_c\hp 1$. 
We define the graph $\Rb  \tensor*[_e]{\#}{_{\!f}} \mtc Q  $
as follows:
\begin{align*}
(\Rb  \tensor*[_e]{\#}{_{\!f}} \mtc Q) \hp 0 & = \Rb\hp 0\cup \mtc Q\hp 0  \quad \mbox{and} \quad
(\Rb  \tensor*[_e]{\#}{_{\!f}} \mtc Q) \hp 1   = 
(\Rb\hp 1 \setminus \{e\})\cup (\mtc Q \hp 1 \setminus\{ f\}) \cup \{ e',f'\},
\end{align*}
where $e'$ and $f'$ are $c$-colored edges defined by $s(e')=s(e)$,
$t(e')=t(f)$ and $s(f')=s(f)$, $t(f')=t(e)$ (see Figure
\ref{source_target}), which makes $\Rb \tensor*[_e]{\#}{_{\!f}} \mtc
Q$ a connected graph in $\Grph{D}$. We will often obviate the edges and
just write $\Rb \# \mtc Q$ if this simplification does not lead to
confusion.
\end{defn}
\begin{figure}[H] 
\centering
\includegraphics[height=2.3cm]{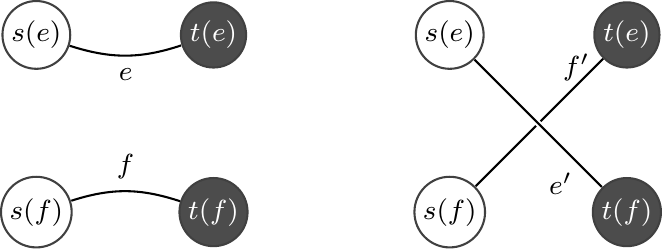}
\caption{   On Definition \ref{thm:surg_ribb}. Here $s$ and $t$ are source and target. The 
orientation of the edges is given by setting the \label{source_target} 
`white vertices' as sources.}
\end{figure}
\begin{lem} \label{thm:ECrib}
For any $\R$, $\mathcal Q \in \tcol$ and edges $e$ in $\R_c\hp1 $ and $f\in\mathcal Q_c\hp 1$
of the same color $c$, the graph 
$\Rb \tensor*[_e]{\#}{_{\!f}} \mtc Q\in \Grph{3} $ satisfies 
\[
\chi(\Rb \tensor*[_e]{\#}{_{\!f}} \mtc Q) = \chi (\Rb) + \chi (\mtc Q) - 2 \,.
\]
\end{lem}
\begin{proof}
 The operation $ \tensor*[_e]{\#}{_{\!f}} $
 is additive with in vertices and edges. Thus, only
 the $2$-bubbles might change:
 \begin{align*}
 \chi(\Rb \tensor*[_e]{\#}{_{\!f}} \mtc Q)&=
 \chi (\Rb) + \chi (\mtc Q) + \mbox{(created $2$-bubbles $-$ deleted $2$-bubbles)}.
  \end{align*}
  Now, the change in the $2$-bubbles can only take place in those
  containing the edges $e$ or $f$.  Since there are three colors,
  there are two $2$-bubbles of $\Rb$, $\B^{cd}_e(\Rb)$, containing
  $e$, namely those with colors $(cd)$, $d\in \{ \hat c\}$.
  Similarly, there are two bubbles $\B^{cd}_f(\mtc Q)$ containing $f$.
  The removal of the edges $e$ and $f$ has as consequence the
  elimination of the $2$-bubbles containing $e$ and $f$, whence four
  $2$-bubbles are eliminated in the new graph.  Now, for each color
  $d\neq c$, the new edges $e'$ and $f'$ lie on the same $2$-bubble of
  $\Rb \tensor*[_e]{\#}{_{\!f}} \mtc Q$. There is exactly one new
  bubble $\B^{cd}_{e'f'}( \Rb \tensor*[_e]{\#}{_{\!f}} \mtc Q)$ for
  each $d$, thus $2$ are created in total. Therefore the $2$-bubbles
  decrease in two and the result follows.
 \end{proof}

The previous lemma justifies the notation in previous definition, since
for compact, closed  $n$-manifolds $M$ and $N$ one has
$
\chi (M\# N)=\chi(M)+\chi(N)-\chi(\mtb S^n)=\chi(M)+\chi(N)+(-1)^{n+1}-1
$.

\begin{example} \label{thm:ex_toro_tres}
We perform this first  
operation on the graphs for a torus and a sphere graphs, namely 
$\mathcal R_0$ and $\mathcal R_1$ of eq. \eqref{toroesferaribbon},
respectively:
\[
\raisebox{-.469\height}{\includegraphics[height=3.1415cm]{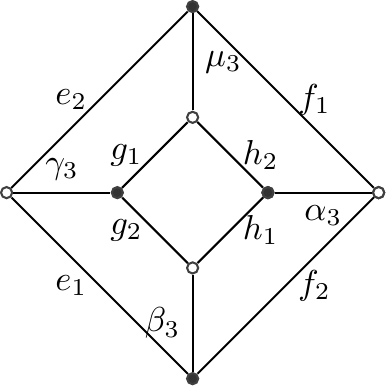}}
\,\,  \tensor*[_{\mu_3}]{\#}{_{\!{\nu_3}}} \,\,
\raisebox{-.469\height}{\includegraphics[height=3.1415cm]{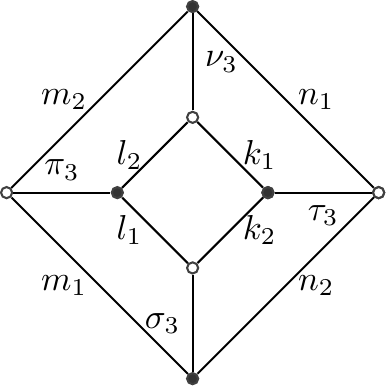}} \,\,=\,\, 
\raisebox{-.389\height}{\includegraphics[height=3.69cm]{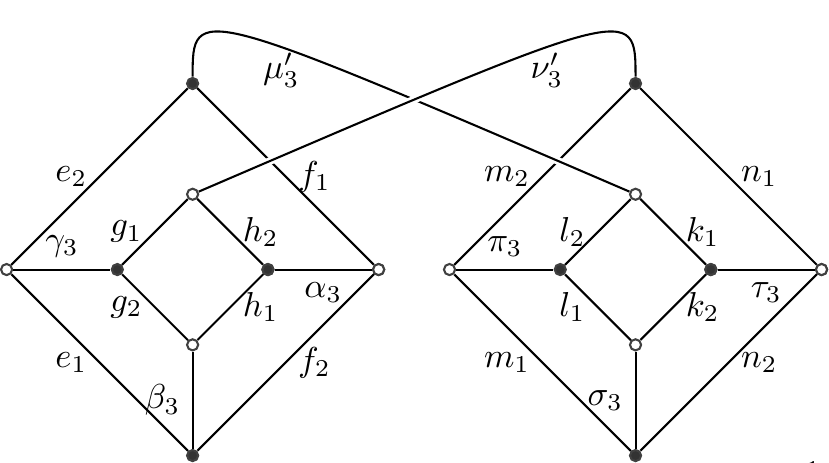}} 
\]
Here each subindex of the edges corresponds with 
its coloring. Moreover, 
the color $0$ has been replaced by $3$, since the result holds 
for graphs in abstract (not only in the QFT context). 
According to the lemma, $\chi(\R_0\# \R_1)=\chi(\R_0)+\chi(\R_1)-2$.
The same happens if we contract $ \mathcal R_0 \# \mathcal R_1$ with 
another copy $ \mathcal R_0'$ of $\R_0$ 
along the edges $\sigma_3'$ and $\beta_3$, respectively. 
We get then
$\Delta ( (\mathcal R_0 \# \mathcal R_1) \# \mathcal R_0)\simeq \T^2$.
This very graph will be used for 
the construction in Section \ref{nontriviality}.
\end{example}

\begin{rem}\label{thm:crys}
If certain $(D+1)$-colored graph $\mathcal B$ represents a ($D$-dimensional) manifold $M$ one
says that $\mathcal B$ is a \textit{crystallization} of $M$ if moreover the
number of $D$-bubbles in $\mathcal B$ is exactly $D+1$; or,
rephrasing that condition, if by removing a single arbitrary color
$i=0,\ldots, D$, one gets a connected graph $\mathcal B^{\hat i}$.
For example, the necklace graph in Example \ref{thm:criteria} is a
crystallization, but the graph \eqref{exampleforbubbles} is not.
\par The school of graph-theoretical representations of manifolds
found a crystallization $\B_M \#_{\mathrm{crys}} \B_N $ of the
connected sum $M\#N$ two manifolds, from the crystallizations of
$\B_M$ and $\B_N$ of the summands. In the orientable case (i.e. for
bipartite graphs), in order to obtain $\B_M \#_{\mathrm{crys}} \B_N $
one deletes two vertices $p \in (\B_M)\hp0_{\vphantom b\mathrm w}$ and
$q \in (\B_N)\hp0_{\mathrm b}$ and puts together, by color, the
half-edges at $s\inv(p)$ and $t\inv(q)$ created by said
vertex-removal.  \par
If one wants to use directly $ \#_{\mathrm{crys}}$ 
a first issue is that crystallizations are not that abundant in
$\fey_D(V)$. Furthermore, the serious drawback is that 
one can always find Feynman diagrams $\B$ and
$\B'$ of a model $V(\Phi)$, such that 
$\B\#_{\mtr{crys}}\B'\in\fey_D(V)$ always lies outside 
the set of Feynman graphs $\fey_D(V)$. 
Nothing forces $(\B\#_{\mtr{crys}}\B')^{\hat
  0}$ to be in the interaction potential $V$. Proposition
\ref{thm:qftcompatible} shows the advantage of using 
the operation $\#$ defined above instead. As a last reason
to prefer $\#$ over $\#_{\mtr{crys}}$ is simplicity. Both 
are related by a $1$-dipole insertion\footnote{The referee 
is acknowledged for this remark}:
\[
\#_{\mtr{crys}}= (1\mbox{-dipole contraction}) \circ \# \,.
\]
Although this relation can be inverted, it is $\#_{\mtr{crys}}$ 
which factors through a simpler operation, namely $\#$, and not 
the other way around. In the crystallization theory of manifolds
it is understandable that $\#_{\mtr{crys}}$, which gets rid of 
two vertices, is natural, for it leads a colored graph towards a simpler
one (`totally contracted', 
or properly a `crystallization'). But here, precisely we do need those vertices
to stay in the same model. 
\end{rem}

\subsection{Matrix models as tensor-models}\label{matrix_as_tensor}
The perturbative expansion of the partition function of 
the matrix model
\begin{equation}
\int \Df{[M,\bar M]} \ee^{-\Tr(\bar M M)-\lambda V(M,\bar M)} \label{pevalente}\,,
\end{equation}
as is well-known, generates ribbon graphs, which are canonically given 
the structure of a triangulated surface by taking 
the dual complex of the construction $\Sigma(\R)$ in  
Appendix \ref{app_ribbons}. Thus the interaction 
vertices of 
\begin{equation}
S[M,\bar M]=\Tr(\bar M M)+ V(M,\bar M) =\Tr(\bar M M) + \lambda\Tr((\bar M M)^p)  
\end{equation}
contribute with $(2p)$-agonal vertices
\begin{equation}
  \raisebox{-.48\height}{\includegraphics[height=2.8cm]{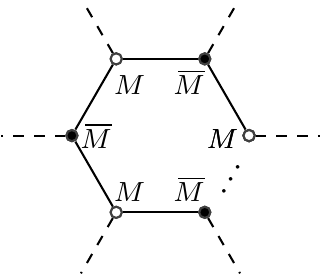}} 
  \quad = \quad \mbox{dual triangulation to }  
  \raisebox{-.475\height}{\includegraphics[height=2.8cm]{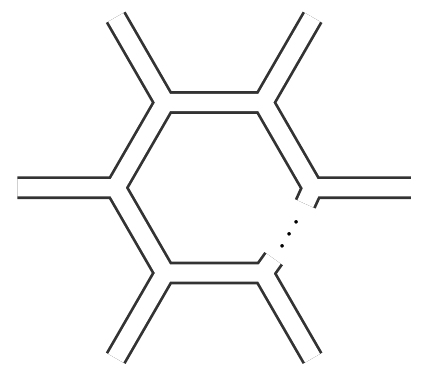}} \,. \label{con_interno}
\end{equation}
 ribbon graph in the LHS is due to the construction in Lemma
 \ref{thm:col_rib}.  We denote by $\fey_2(\phi^{2p})$ the set of
 Feynman diagrams of the theory defined by the functional
 \eqref{pevalente}. Other conventions differ from the one given so
 far. There, the loop inside the vertex, that is the ($12$)-colored
 bubble, is not drawn.  For instance, if $p=2,3,4$ one would have the
 following representation of the vertices:
\begin{equation}\raisebox{-.5\height}{
\includegraphics[height=2.25cm]{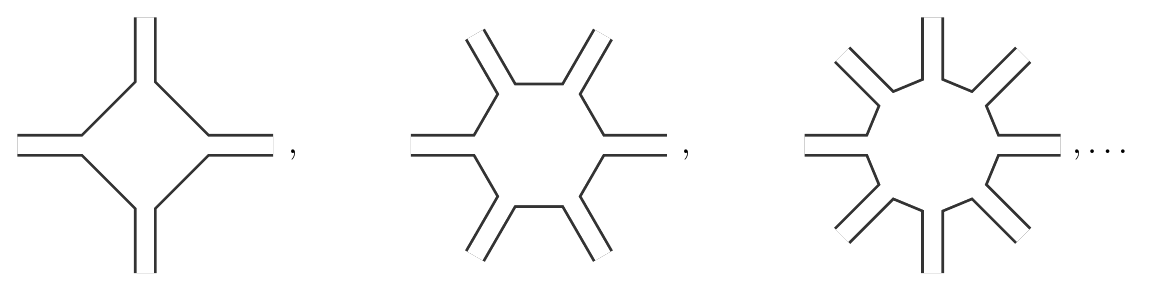}}
\label{sin_interno}
\end{equation}
that is, the colors $i=1,2$ are drawn as simple lines, and the
$0$-color double. The geometric realization is of course unaltered,
since in the only difference is to decide whether one adds a vertex,
as in the latter case (cf. \eqref{sin_interno}), or a face, as in the
representation \eqref{con_interno}.\\

A (complex) matrix model here is, in the 
tensor model context, a polynomial interaction:
\[
V(\phi, \bar\phi)= \sum_{p\in I} \lambda_p\Tr ((\phi\bar\phi)^p)) 
\qquad (\mbox{finite } I\subset  \Z_{>2}).
\]

\begin{prop} \label{thm:qftcompatible}
Fix a rank-$2$ CTM (or a complex matrix model) interaction $V$.
Let $\R $ and $\mathcal Q$ be Feynman diagrams
of that model. Let $e\in \R_0\hp0$ and $f\in \mathcal Q_0\hp0$.  Then
$\R \tensor*[_{e}]{\#}{_{\!f}} \mathcal Q \in \fey_2(V)$ as well.
\end{prop}
\begin{proof}
 It is trivial by noticing that
 $0$-colored edges are Wick contractions.
\end{proof}

Let $\Riem$ be the homeomorphism-classes of connected, closed
orientable surfaces.  We consider the empty surface also as an element
of $\Riem$.  Then we claim that the only quartic model
$\fey_2^{\C}((\phi\bar\phi)^2)$ has enough graphs to generate all of
$\Riem$.  A weaker version of the following result corresponding to
the real matrix $\phi^4$-theory might be known. In the complex theory
with potential $(\phi\bar\phi)^2$ some graphs of the real theory,
$\fey_2^{\re}(\phi^4)$, are forbidden; nonetheless:
\begin{lem} \label{thm:surj_2}
There is a surjection $\fey_2((\phi\bar\phi)^2)\to \Riem$.
\end{lem}
\begin{proof}
Any vacuum graph $\G_0$ in $\Grph{2+1}\hp 0 \subset
\fey_2^{\mtr{v}}((\phi\bar\phi)^2)$ yields the empty surface $\beta(\G_0)\in
\Riem$; we exclude this trivial case from now on.  Because of the
classification of orientable, closed surfaces, $\Riem \simeq \Z_{\geq
  0}$, one has to construct a graph for each $g\in \Z_{\geq 0}$.  For
$g=0$, the graph $\mtc{R}_0$ has been shown to triangulate $\mathbb S^2$. For
$g>0$ we proceed differently. Let $\mathcal O$ be the following graph:
 \begin{equation}
\mathcal O= 
 \raisebox{-.475\height}{
   \includegraphics[height=3.60cm]{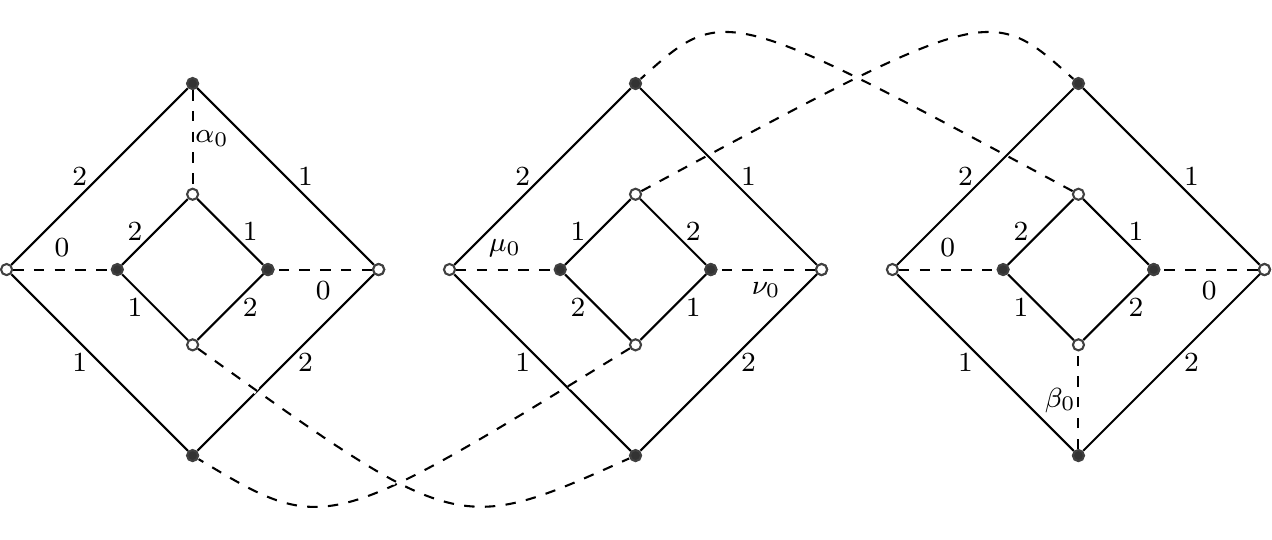}}
\label{block2cob} 
 \end{equation}

 Then consider $g$ copies $\mtc{O}\hp1,\ldots, \mtc{O}\hp g$, of that
 graph $\mathcal O$ with distinguished $0$-colored edges $\mu_0^ i$
 and $\nu_0^ i$, $i=1,\ldots,g$ as shown above.  Then we claim
 that\footnote[9]{ One can also do cell-counting: the number of
   vertices of this graph is $3\cdot 8\cdot g$; the number of edges is
   $3\cdot 12 \cdot g$ and the number of faces is
   $|\mtc{Q}_g\hp{12}|+|\mtc{Q}_g\hp{01}|+|\mtc{Q}_g\hp{02}|=2+10\cdot
   g$.  Indeed, one sees trivially that the $(12)$-colored bubbles
   are $3\cdot 2 \cdot g$. The only non-trivial part is to count the
   $(0i)$-bubbles. We see now by induction in the number $g$ of sums
   $\#$ that the number of $(01)$-bubbles $|\mtc{Q}_g\hp{01}|$ is
   $2g+1$ and it is also evident that
   $|\mtc{Q}_g\hp{01}|=|\mtc{Q}_g\hp{02}|$. }
 \begin{equation}
  \label{sum_muchos}
 \mathcal{Q}_g:=\mathcal O\hp1 \tensor*[_{\mu_0^ 1}]{\#}{_{\!\nu_0^2}} (\mathcal O\hp2  
  \tensor*[_{\mu_0^ 2}]{\#}{_{\!\nu_0^3}} ( \mathcal O\hp 3 \tensor*[_{\mu_0^ 3}]{\#}{_{\!\nu_0^4}}\ldots 
  \tensor*[_{\mu_0^ {g-2}}]{\#}{_{\!\nu_0^{g-1}}}(
 \mathcal O \hp{g-1} \tensor*[_{\mu_0^{g-1}}]{\#}{_{\!\nu_0^g}}  \mathcal O \hp{g})_{\ldots}))
 \end{equation}
 has genus $g$. Indeed, after Lemma
 \ref{thm:ECrib} each sum $\#$ in \eqref{sum_muchos} decreases the
 Euler characteristic in $2$. Since each summand $\mathcal O$ has
 $\chi(\mathcal O)=0$ (cf.  Example \ref{thm:ex_toro_tres}),
 \[\chi(\mathcal{Q}_g)=(g-1)\cdot \chi(\mathcal O)-2(g-1)=2-2g\,.   \] 
 \end{proof}

\begin{rem} \label{alphabeta}
By the same token, one can also glue by $\alpha_0$ and $\beta_0$ instead of by $\mu_0$ and $\nu_0$.
The resulting graph 
$
 \mathcal{K}_g:=\mathcal O\hp1 \tensor*[_{\beta_0^ 1}]{\#}{_{\!\alpha_0^2}} (\mathcal O\hp2 \tensor*[_{\beta_0^ 2}]{\#}{_{\!\alpha_0^3}}  (\ldots
 \mathcal O \hp{g-1} \tensor*[_{\beta_0^{g-1}}]{\#}{_{\!\alpha_0^g}}  \mathcal O \hp{g})_{\ldots}))
$ has genus $g$.
\end{rem}
 
\begin{example}\label{keineSphere}
In view of Lemma \ref{thm:surj_2},
the rank-$3$ model with interaction vertex set to $\mathcal Q_g$, for 
$g\geq 1$ (after
properly changing the color $0$ into $3$),   
generates no melons at all. A lower-order
polynomial interaction with the same characteristic is 
\[V(\phi,\bar\phi)=\raisebox{-.46\height}{
  \includegraphics[height=2.25cm]{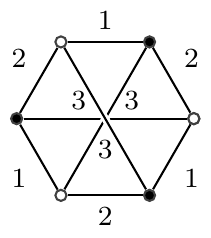}}\]
This is consequence of the lower-bound for the degree (i), mentioned
in example \ref{thm:criteria}.  Thus in rank-$3$ theories, an interaction
vertex with suitable high degree (e.g. $g=3$) has degree
$\omega(\G)\geq 3$. Thus generation of spheres is not guaranteed (at least not
before renormalization, if one does not introduce quadratic
counterterms).  
\end{example}

For any graph 
$\mathcal F\in\Grph{D+1}\hp {2k} \setminus \mathrm{im}\, C$, i.e. 
for each graph $\mathcal F$
that is not the cone of a $D$-colored graph, 
the set of inner propagators is not empty.  For such graphs one can 
increase the number of external legs as follows:
\begin{defn} Let $k\in \Z_{\geq 0}$ and $\mathcal F\in\Grph{D+1}\hp {2k} \setminus \mathrm{im}\, C$. Let $e_0=(\overline{ap})$ an internal edge of $\mtc F$. 
We denote by $\mathcal F \curlywedge e_0$ or $\mathcal F  \curlywedge(\overline{ap})$
the graph obtained 
from $\F$ by \textit{opening} the $0$-colored edge $e_0$. By that, we
mean that one creates two external legs (or leaves), one \textit{at} $a$ and
one \textit{at} $p$, so $\mathcal F \curlywedge e_0 \in \Grph{D+1}\hp
{2k+2}$. Pictorially,
\[
  \begin{tikzpicture}%
    [
    baseline=-1.2ex,shorten >=.1pt,%node distance=18mm, 
    semithick,auto,
    every state/.style={fill=white,draw=texto,inner sep=.3mm,text=texto,minimum size=0},
    accepting/.style={fill=white,text=black},
    initial/.style={white,text=texto}
]
 
       \draw [dashed] (-2,0) -- (-3,.6);
   %\draw [dashed] (-2,0) -- (-3,.1);
   \draw [dashed] (-2,0) -- (-3,-.6);
        \draw[dashed] (-1.4,0 ) circle (.4) ;  
        \node at (-.7,0) {$e_0$} ; 

   \node at (-3.,0.1) {\scriptsize $\vdots$};
%    \node at (3.,0.1) {\scriptsize $\vdots$};
    \node[state,accepting ] at (-2,0) {$\,\,\, \phantom{a}\mathcal{F} \phantom{b}\,\,\,$};
   
      \begin{scope}[xshift=1cm]
    
    ;
\node at (-.751,-.1250) {$\mapsto$};
\begin{scope}[xshift=3cm]
   \draw [dashed] (-2,0) to[bend left] (-1,.6);
    \draw [dashed] (-2,0) to[bend right] (-1,-.6);
       \draw [dashed] (-2,0) -- (-3,.6);
   %\draw [dashed] (-2,0) -- (-3,.1);
   \draw [dashed] (-2,0) -- (-3,-.6);
   \node at (-3.,0.1) {\scriptsize $\vdots$};
   %\node at (3.,0.1) {\scriptsize $\vdots$};
   \node[state,accepting  
    ] at (-2,0) {$\mathcal{F} \curlywedge e_0 $};
   \end{scope}  
\end{scope} 
\end{tikzpicture}
\]
\end{defn} 

\begin{thm} \label{thm:gen_bordisms}
Let $2$-$\Cob$ be the set of all orientable $2$-bordisms. 
There exists a surjection 
\[ \xi:\fey_2((\phi\bar\phi)^2)\to 2\mbox{-}\Cob\,.\]
\end{thm}
\begin{proof}
Consider the graph $\mathcal O$ in \eqref{block2cob} and let 
\begin{equation}
\mathcal{N}:=
\raisebox{-.475\height}{
\includegraphics[height=3.60cm]{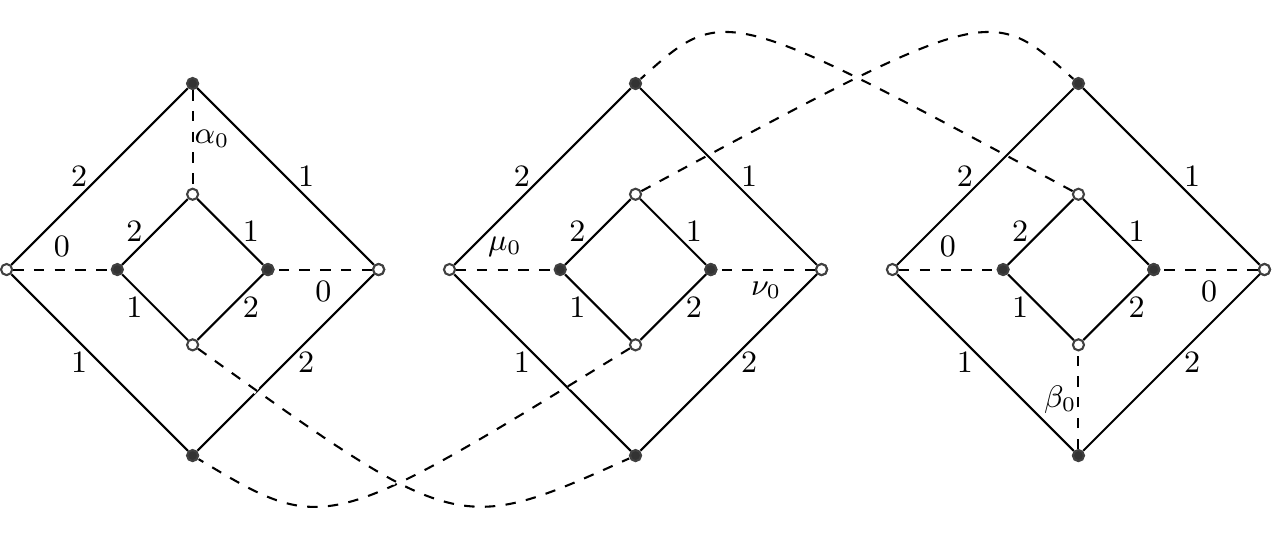}}
\end{equation}
 Obviously, both graphs are in $\fey_2((\phi\bar\phi)^2)$.  Let $M:
 \sqcup^C \mathbb S^1 \to \sqcup^B \mathbb S^1$ be an arbitrary
 element in $2$-$\Cob$.  That is, two arbitrary closed $1$-manifolds,
 $\sqcup^B \mathbb S^1$ and $\sqcup^{C} \mathbb S^1$ are cobordant via
 $M$, a genus-$g$ orientable, compact surface with boundary. We now find
 a graph $\mathcal Q_{g,B,C}$ which (dually) triangulates $M$.  \par
 Remark that the case $B=C=0$ is the statement in Lemma
 \ref{thm:surj_2}. Thus we can suppose $0<B\leq C$ and set  
 $m:=\max\{g,C\}>0$. Define the following Feynman-graph-valued
 functions:
\[
 \mathcal X: \{0,1\}\to \fey_2 ((\phi\bar\phi)^2), \qquad
 \mathcal X(\epsilon)= \begin{cases}
                         \mathcal N & \mbox{ if } \epsilon=0
                         \,, \\
                         \mathcal O & \mbox{ if } \epsilon=1 \, ,    
                       \end{cases}
 \] 
 and
 \[ \mathcal S: \{0,1\}\times \{0,1,2\} \to \fey_2 ((\phi\bar\phi)^2), \qquad
 \mathcal S(\epsilon,i)= \begin{cases}
 \phantom{(}\mathcal X(\epsilon)& \mbox{ if } i=0\,, \\
                         \phantom{(}\mathcal X(\epsilon)\curlywedge \alpha_0 & \mbox{ if } i=1\,, \\
                         (\mathcal X(\epsilon)\curlywedge \alpha_0) \curlywedge\beta_0  & \mbox{ if } i=2\,.
                       \end{cases}\]
(For example, $\mathcal S(1,2)=\mathcal O\curlywedge \alpha_0 \curlywedge\beta_0$, or explicitly, 
\[\raisebox{-.475\height}{
\includegraphics[height=3.29cm]{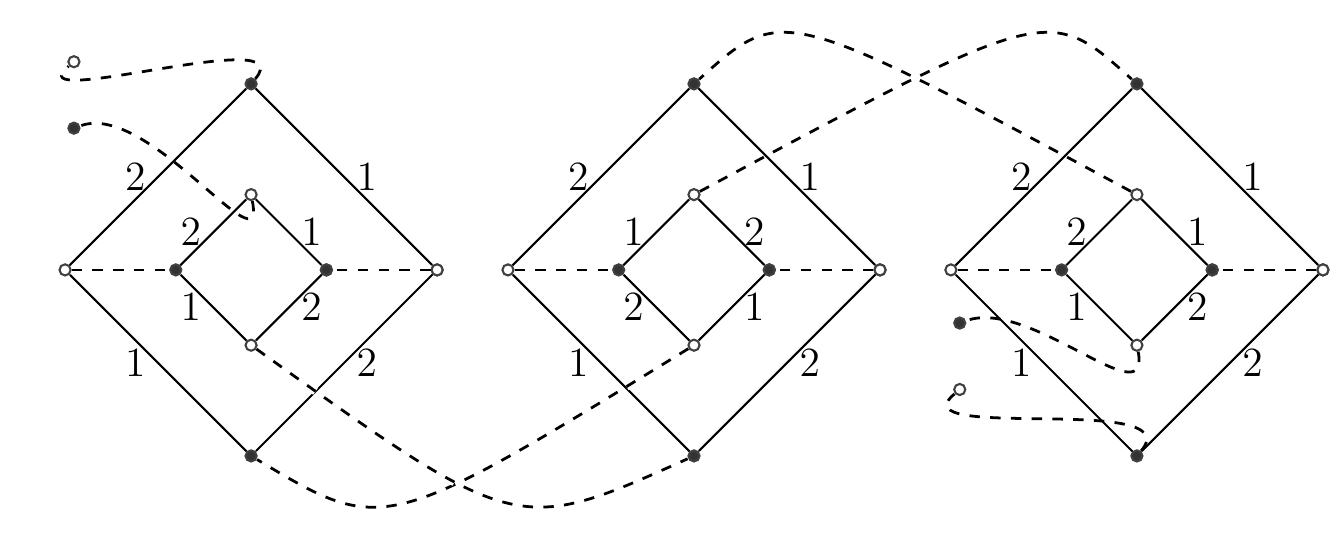}}\] where
the external lines come from cutting $\alpha_0$ and $\beta_0$.)
Notice that the only difference between $\mathcal O$ and $\mathcal N$
is the edge-coloration of four edges, namely  the central $(12)$-bicolored bubble adjacent to $\mu_0$ 
and $\nu_0$. The following connected sum is well-defined:
\[
G((\epsilon_1,\iota_1),\ldots,(\epsilon_m,\iota_m))=
\mathcal S(\epsilon_1,\iota_1) \tensor*[_{\mu_0^ 1}]{\#}{_{\!\nu_0^2}}  \mathcal S(\epsilon_2,\iota_2)  
  \tensor*[_{\mu_0^ 2}]{\#}{_{\!\nu_0^3}}  \ldots 
 \tensor*[_{\mu_0^{m-1}}]{\#}{_{\!\nu_0^m}}  \mathcal S(\epsilon_m,\iota_m)\,,
\]
where $\mu^i_0$ and  $\nu_0^i$ refer, respectively, to the $\mu_0$ and $\nu_0$ edge
of $\mathcal X(\epsilon_i)$. Define then the graph
$\mathcal Q_{g,B,C}$ by evaluating 
$G((\epsilon_1,\iota_1),\ldots,(\epsilon_m,\iota_m))$
at 
\[
\epsilon_k=\begin{cases}
            1& \mbox{ if } 1\leq k\leq g,\\
            0& \mbox{ if } k>g, 
           \end{cases}
\and
\iota_k=\begin{cases}
            2 & \mbox{ if } 1\leq k\leq B, \\
            1 & \mbox{ if } B<k\leq C, \\
            0 & \mbox{ if } k>C.
           \end{cases}
\]
Each $0$-colored-edge removal creates exactly a boundary component
$\mathbb S^1$, for none of the $2$-bubbles of $\alpha_0$ and $\beta_0$
implies the edges $\mu_0$ and $\nu_0$. Hence one has indeed created
$C+B$ boundaries $\mathbb S^1$.  If we cap them (closing all the
broken $\alpha_0$ and $\beta_0$) we get $\mathcal{Q}_{g,0,0}$ which is
the $\mathcal Q_g$ of Lemma \ref{thm:surj_2} and hence has genus $g$.
\end{proof}

\begin{example}\label{2cob}
Consider the following genus-$2$ bordism $W: \sqcup^3\mathbb S^1 \to \sqcup^2 \mathbb S^1 \in 2$-$\Cob$,
\[\hspace{-.7cm}W=\hspace{-.2cm}\raisebox{-.45\height}{\includegraphics[height=2.5cm]{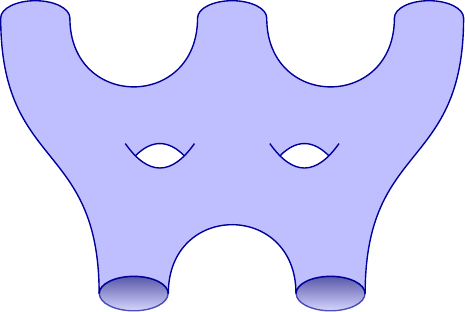}}\]
The construction for $\mathcal Q_{2,2,3}$ is, explicitly, 
the graph in Figure \ref{fig:bord_explicit}, in which we represented by \Rightscissors $\phantom{.}$the opening
of the $0$-colored edges. The theorem states that 
$\mathcal Q_{2,2,3}\in\fey_2((\phi\bar\phi^2))$
triangulates $W$.
\end{example}

  \begin{figure} 
\center
\raisebox{-.45\height}{\includegraphics[width=.8171\textwidth]{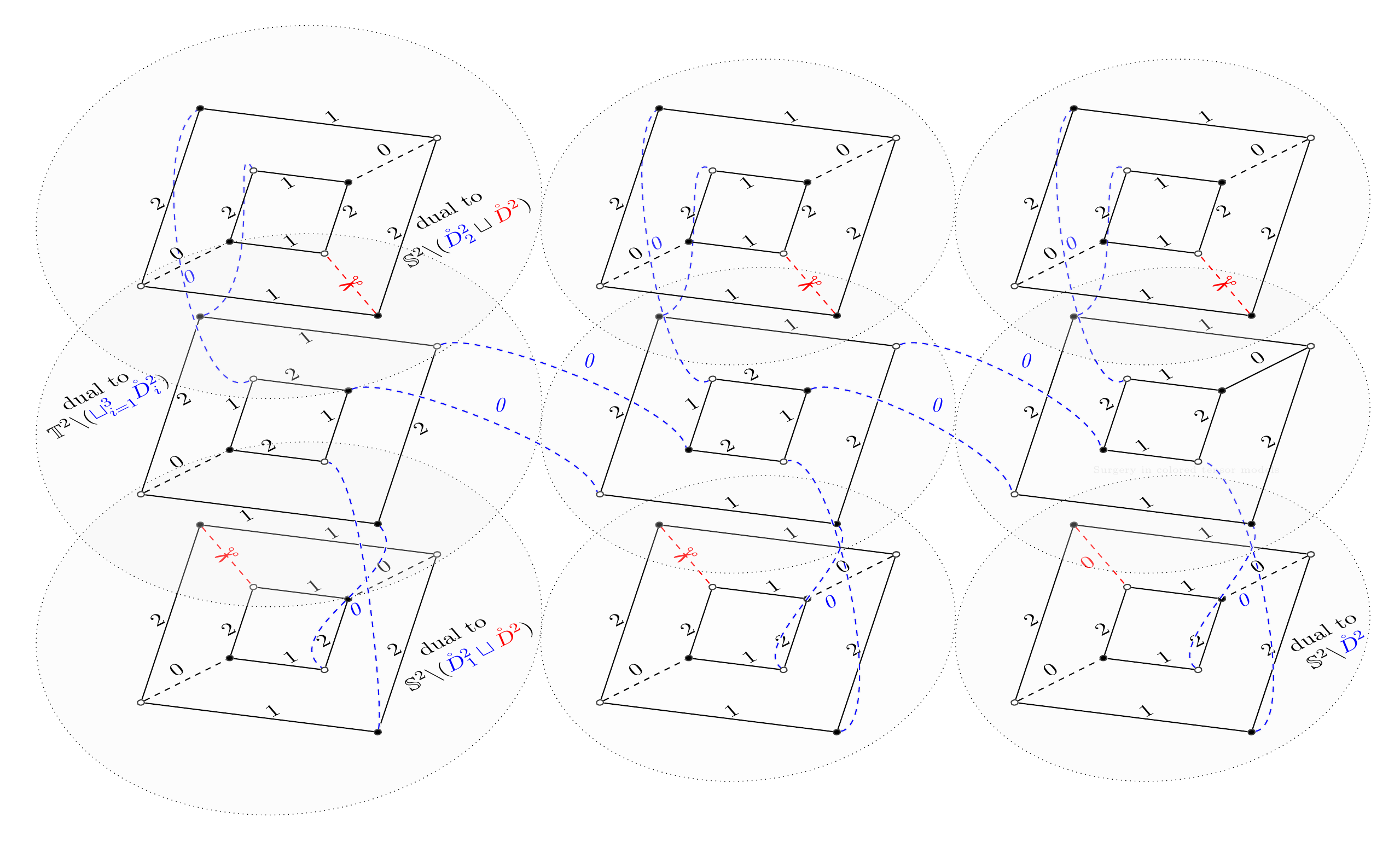}}
\caption{\label{fig:bord_explicit}The triangulation of $W$ given by the algorithm of Theorem \ref{thm:gen_bordisms}}
\end{figure}
%% 
%%  = = = = = = = = = = = = = = = = = = = = = = = = = = = = = = = = = =  
%% 
%%  = = = = = = = = = = = = = = = = = = = = = = = = = = = = = = = = = =  
%% 
%%  = = = = = = = = = = = = = = = = = = = = = = = = = = = = = = = = = =  

\section{Topological completeness of the boundary sector of the $\phi_3^4$-theory} 
\label{nontriviality}
 Let $\feyn$ be (shorthand for) the set of connected Feynman graphs of
the rank-$3$ $\phi^4$-colored tensor theory with the three 
vertices in eq. \eqref{vertices}. Throughout, $\Riem$ will be 
the set of homeomorphism-classes of closed orientable
\textit{connected} Riemannian surfaces,   $\Riem\simeq \Z_{\geq 0}$.
Further, we denote by $\Riemd$ be the set of 
\textrm{possibly disconnected} closed, orientable
Riemannian surfaces. In order to proof the main 
result of this section, we need. 
\par
It is trivial to construct Feynman graphs $\G$ which have disconnected
boundary, just by letting $\G$ itself be disconnected. This would be
rather useless, though,  for $\G$ would be cancelled out in the generating functional of
connected correlation functions.  The previous lemma says that,
nevertheless, it is possible to `separate boundaries' at
wish. Moreover, it tells us how to generate \textit{connected} graphs
with a precise disconnected boundary. 
\begin{lem}\label{thm:sepa}
The following two graphs $\mathcal M,\mathcal P 
\in \feyn$ separate boundary components:
\begin{align}
\mathcal{M}  :=   
 \raisebox{-.45\height}
{\includegraphics[height=2.00cm]{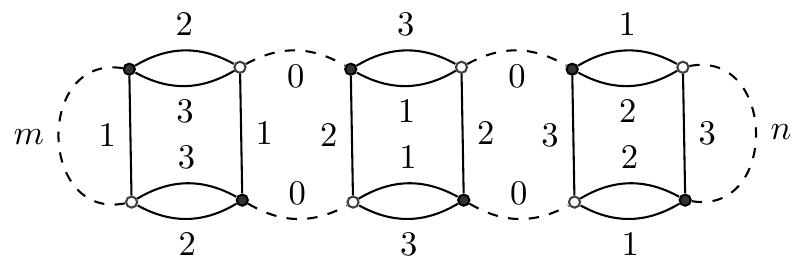}
}  \qquad
\mathcal{P} := 
\raisebox{-.45\height}{\includegraphics[height=2.0cm]{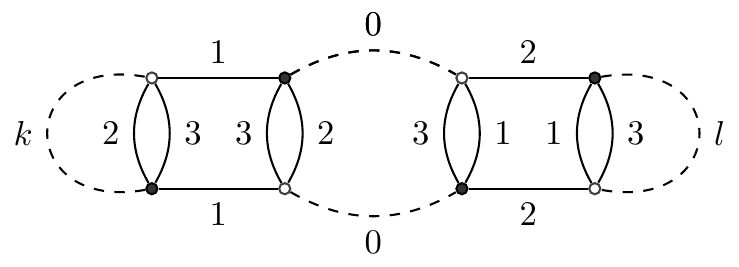}}
\end{align}
This means that, if $\G$ and $\H$ are $\phi^4_3$-Feynman graphs 
that are not in the image $\im \,C$ 
of the cone, $\G, \H \in  \feyn \setminus \im \,C$, and $g$
and $h$ are internal propagators of $\G$ and $\H$, respectively, then one has:
\begin{equation} \label{eq:separaciones}
\partial (\G  \tensor*[_{g}]{\#}{_{k}}  \mathcal P  \tensor*[_{l}]{\#}{_{h}}   \H ) = \partial \G \sqcup \partial \H\,  
\qquad \mbox{and}\qquad
\partial ( \G  \tensor*[_{g}]{\#}{_{m}} \mathcal M \tensor*[_{n}]{\#}{_{h}}  \H ) = \partial \G  \sqcup \partial \H\,.
\end{equation}
We abbreviate $ \G  \tensor*[_{g}]{\#}{_{k}} \mathcal P \tensor*[_{l}]{\#}{_{h}}  \H $
as $
\G \# \mathcal P \# \H$ and choose a similar simplification for 
the other graph. Both $\G \# \mathcal P \# \H$ and $
\G \# \mathcal M \# \H$ are in $\feyn$. 
\end{lem}

We do the proof for the equality in which 
$\mathcal P$ occurs in one of the summands (by replacement of $\mathcal P$ 
by $\mathcal M$ one can get \textit{mutatis mutandis} readily a proof). 
If 
$\mathcal{Y} $ denotes $\G \# \mathcal P\# \H $ 
then Lemma \ref{thm:sepa} can be easily grasped as follows:
\[
\partial \mathcal{Y} = \partial 
\Bigg(\raisebox{-3pt}{
  \begin{tikzpicture}%
    [scale=.9,
    baseline=-1.2ex,shorten >=.1pt,%node distance=18mm, 
    semithick,auto,
    every state/.style={fill=white,draw=texto,inner sep=.3mm,text=texto,minimum size=0},
    accepting/.style={fill=white,text=black},
    initial/.style={white,text=texto}
]
% \node at (-3.5,0) {$\mathcal{Y}\,=\,$};
\node at (-1,-.3) {$g'$};
\node at (-1,.3) {$k'$};
\node at ( 1,-.3) {$l'$};
\node at ( 1,.3) {$h'$};
    \draw [dashed] (0,0) to[bend right] (1,-.6);
    \draw [dashed] (0,0) to[bend left]  (1,.6);
    \draw [dashed] (0,0) to[bend right] (-1,.6);
    \draw [dashed] (0,0) to[bend left] (-1,-.6);
    \draw [dashed] (-2,0) to[bend left] (-1,.6);
    \draw [dashed] (-2,0) to[bend right] (-1,-.6);
    \draw [dashed] (2,0) to[bend right] (1,.6);
    \draw [dashed] (2,0) to[bend left] (1,-.6);
    \draw [dashed] (2,0) -- (3,.6);
  % \draw [dashed] (2,0) -- (3,.1);
   \draw [dashed] (2,0) -- (3,-.6);
       \draw [dashed] (-2,0) -- (-3,.6);
   %\draw [dashed] (-2,0) -- (-3,.1);
   \draw [dashed] (-2,0) -- (-3,-.6);
   \node at (-3.,0.1) {\scriptsize $\vdots$};
   \node at (3.,0.1) {\scriptsize $\vdots$};
   \node[state,accepting,  %drop shadow
   ] at (0,0) {$\,\, {\mathcal P}\,\,$};
    \node[state,accepting  
    ] at (2,0) {$\,\, {\mathcal \H}\,\,$};
    \node[state,accepting  
    ] at (-2,0) {$\,\, {\mathcal \G}\,\,$};
  \end{tikzpicture}}\Bigg)
=\partial\left(\raisebox{-.15\height}{
 \begin{tikzpicture}
    [scale=.9,
    baseline=-1.2ex,shorten >=.1pt,  
    semithick,auto,
    every state/.style={fill=white,draw=texto,inner sep=.3mm,text=texto,minimum size=0},
    accepting/.style={fill=white,text=black},
    initial/.style={white,text=texto}
]
 
     \draw [dashed] (-1.5,0) circle (.38);
     \node at (-1.0-.3,0) {$g$};
   \draw [dashed] (-2,0) -- (-3,.6);
    \node at (-3.,0.1) {\scriptsize $\vdots$};
   \draw [dashed] (-2,0) -- (-3,-.6);
   \node[state,accepting  
   ] at (-2,0) {$\,\, {\mathcal \G} \,\,$};
  \end{tikzpicture}
 }\right) \sqcup \partial
\left(\!\!\raisebox{-.15\height}{
 \begin{tikzpicture}%
    [scale=.9,
    baseline=-1.2ex,shorten >=.1pt,
    semithick,auto,
    every state/.style={fill=white,draw=texto,inner sep=.3mm,text=texto,minimum size=0},
    accepting/.style={fill=white,text=black},
    initial/.style={white,text=texto}
]
    \draw [dashed] (2,0) -- (3,.6);
   \node at (3.,0.1) {\scriptsize $\vdots$};
   \draw [dashed] (2,0) -- (3,-.6);
        \node at ( 1.0+.2,0) {$h$};
    \draw [dashed] (1.4,0) circle (.38);
  \node[state,accepting  
  ] at (2,0) {$\,\, {\mathcal \H} \,\,$};
  \end{tikzpicture}
 }\right) =  \partial \G \sqcup \partial \H\,.
\]
\begin{proof} It is obvious that $\mtc Y$ is in $\feyn$,
since the  $3$-bubbles of the amputation
$\mathrm{inn}(\mtc Y)^{\hat 0}$ are quadratic vertices. 
We verify that the edge and vertex sets of both 
$\partial \G \sqcup \partial \H$ and $\partial \mathcal{Y}$ are the same, as
well as the adjacency. If $\G$ and $\H$ have no external edges, 
then the result is a trivial equality of empty graphs. Then we assume that  
at least one of them has external legs.
\begin{itemize}
\item[\scriptsize$\bullet$] \textit{Vertices.}  
  Notice that $\#$ is additive in the number of external vertices of its graph summands.
  Since $\mathcal P$ has no external vertices, 
  if follows from the definition of boundary graph 
  that the number of vertices of $\partial{\mathcal Y} $
  is the sum of those of  $ \partial \G $ plus those of
   $ \partial \H$. 
  Thus, both vertex-sets of 
  graphs in both sides of \eqref{eq:separaciones} are identical,
  also with the same bipartiteness. 
  
\item[\scriptsize$\bullet$] \textit{Edges}.  
For any vertices 
$x,c \in  \partial \mtc Y$
and for any colour $a=1,2,3$, we prove, that joining them, there 
exist an $a$-colored edge $f_a$ in the graph $\partial \mtc Y$
if and only if there exists an $a$-colored edge $f'_a$ of 
$\partial \G \sqcup \partial\H $ between $x$ and $c$. 
The case in which $x$ and $c$ are both white or both black vertices 
is trivial, for there is no path between them. Thus, we assume w.l.o.g.
that $x$ is black and $c$ white and prove now both directions of the
equivalence:

($\Rightarrow$) If $f_a= \,\overline{cx} \, \in (\partial \mtc Y)\hp 1_a$
the first thing to notice is that, referring to 
in Fig \ref{fig:separatrixp}, there is no bicolored path through 
$(\mathcal P \curlywedge k) \curlywedge l$ 
that joins $d$ with $q$ nor $p$ with with $b$. This means 
that $c$ and $x$ are either both in $\partial \G$ 
or both in $\partial \H$. The case is symmetric in $\G$ and $\H$
and we thus suppose the former case, $c,x \in (\partial\G)\hp 0$, 
and prove that there is a (0$a$)-bicolored path \textit{entirely in} $\partial \G$ joining them.
By definition of boundary graph, there is a  
$(0a)$-bicolored path $\gamma$ in
$\mtc Y$ between $c$ and $x$, that originates the 
edge $f_a=\overline{cx}$.  
Let $g'$ be the 
$0$-colour edge created by the sum $ \tensor*[_g]{\#}{_{\!k}}$ in $\mtc Y$. The edge $g'$
is belongs to the subgraph $(\mathcal P \curlywedge k) \curlywedge l$ of $\mtc Y$ (see Fig 
 \ref{fig:separatrixp}). 
We discern two cases:

\begin{figure}%[H]
\begin{subfigure}{0.45\textwidth}
\centering
\includegraphics[width=5.5cm]{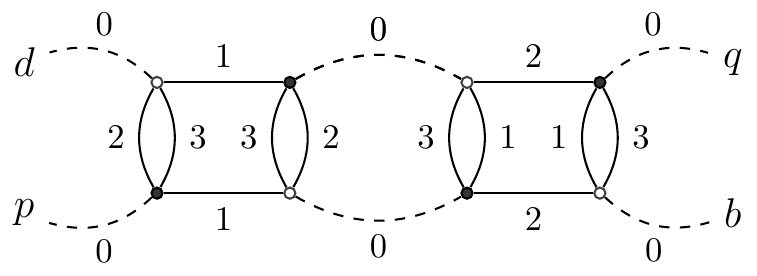} 
 \caption{\small $(\mathcal P \curlywedge k) \curlywedge l$} \label{fig:separatrixp}
\end{subfigure}
\hspace*{\fill} 
\begin{subfigure}{0.45\textwidth}
\centering
\includegraphics[width=5.5cm]{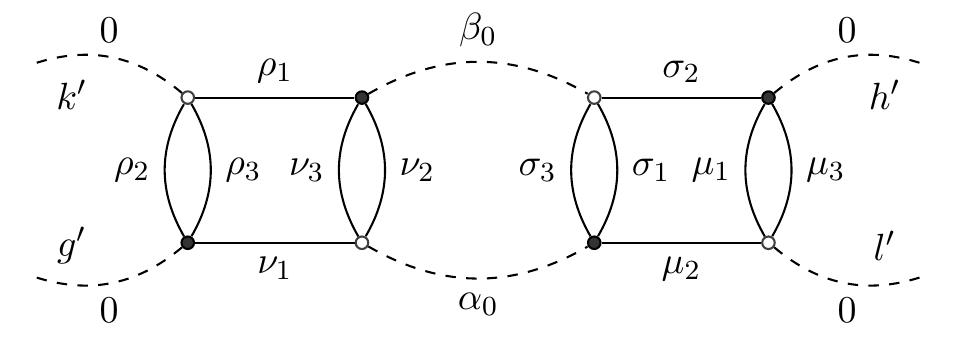}
 \caption{\footnotesize Every subscript of the Greek letters is its color} \label{fig:separatrix}
\end{subfigure}
\caption{On the proof of Lemma \ref{thm:sepa}}   
\label{fig:separatrix_both}
\end{figure}

\begin{itemize}
\item \textit{Case 1: If $\gamma$ does not pass through $g'$}. In this case, $\gamma$ itself is a 
 bicolored path in $\G$ between the given vertices $c$ and $x$. Thus $\gamma$
 also originates an $a$-colored edge between these.
 
\item \textit{Case 2: If $\gamma$ does pass through $g'$}. 
This means, as shown in Figure
 \ref{fig:separatrixp}, that for any color $a$, 
 $\gamma$ will pass through $k'$ as well. But, since $\gamma$ joins 
 $c$ with $x$, this means 
 that in $\G$, there is an $a0$-bicolored path from $c$ to $x$ (passing 
 through $g$, which differs from $\gamma$ only in that edge).	
\end{itemize} 
This shows $\partial \mathcal Y_a\hp1 \subset 
(\partial \mtc G \sqcup \partial \mtc H)\hp 1_a$ for arbitrary $a$. 
 
($\Leftarrow$) Assume that there exists an edge $f_a'\in (\partial \mtc G \sqcup \partial \mtc H)\hp 1_a$ between $c$ and $x$.
Then either both $x,c$ are external vertices of $\G$ or both of $\H$.
\begin{itemize}
\item If $x,c$ are in $\G_{\mtr{out}}\hp 0$. 
Again, $f_a'$ is originated by certain ($a0$)-bicolored path $\zeta$ in $\G$.
If the edge $g$ does not lie on this path, then the whole path is still in $\mtc Y$, so $f_a \in (\partial \mtc Y)\hp 1_a$.
On the other hand, if $g$ is one of the propagators in $\zeta$, notice that 
the same path $\zeta$ with $g$ replaced by the concatenation of the following edges
\[
 \begin{cases}
           g' \nu_1 \alpha_0 \sigma_1 \beta_0 \rho_1 k'  & \mbox{if } a=1\\
           g'\rho_a k' & \mbox{if } a\neq 1
         \end{cases}
\]
is a ($0a$)-bicolored path that lies in $(\partial \mtc Y) \hp 1_a$ and goes from $c$ to $x$.
\item If $x,c$ are in $\H_{\mtr{out}}\hp 0$. 
Similarly, $f_a$ is originated by certain ($a0$)-bicolored path $\xi$ in $\H$.
If the edge $h$ does not lie on this path, obviously $f_a \in (\partial \mtc Y)\hp 1_a$.
But if $h$ is one of the $0$-colored edges in the bicolored path $\xi$, notice that 
replacing $h$ in $\xi$ by the concatenation of
\[
 \begin{cases}
           h' \sigma_2 \beta_0 \nu_2 \alpha_0 \mu_2 l'  & \mbox{if } a=2\\
           h'\mu_a l' & \mbox{if } a\neq 2
         \end{cases}
\]
lies in $(\partial \mathcal Y) \hp 1_a$, 
and is a ($0a$)-bicolored path from $c$ to $x$.
Thus $( \partial \mtc H)\hp1_a \subset \partial \mathcal Y_a\hp1 $.
\end{itemize}
In any case, we have shown the 
direction $\partial \mathcal Y_a\hp1 \supset (\partial \mtc G \sqcup \partial \mtc H)_a\hp1$ 
and thus the lemma as well. 
\end{itemize}
\end{proof}

We use this result to prove the main result of this section. Before doing so, 
we need another result. For any non-negative
integer $g$, let $\Sigma^{g}=\#^g\T^2$ (this $\#$ is the 
usual topological connected sum. Also $\Sigma^0=\mathbb S^2$). 

\begin{lem} \label{thm:Tg}
For each $g\in \Z_{\geq 0}$ there exists a $\phi_3^4$-Feynman graph  $\mtc T_g$ 
whose boundary graph $\partial \mtc T_g $ triangulates $\Sigma^{g}$.
\end{lem}

\begin{proof}
We define the following (so-called canonical \cite{carrozzaPhD})
$3$-colored graphs $\mathcal C_g$ of genus $g$. 
For $g=0$, $\mathcal C_0$ is just the graph 
$\raisebox{-.4ex}{\includegraphics[width=2.5ex]{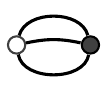}}$. 
For $g\in\Z_{>0}$, one now constructs $\mtc C_g$ form a regular $2(2g+1)$-agon whose
vertices are colorated in an alternating way: black, white, black, white,...; 
between those, the sides are given also an alternating edge-coloration (1,2,1,2,...).
To this $2(2g+1)$-agon we add its $(2g+1)$ longest 
diagonals and color these edges 
with $3$. The resulting graph is $\mathcal C_g$ (for instance, 
$\mathcal C_1$ is $K_{\mathrm c}(3,3)$). 
The terminology `genus' for these graphs is appropriate. Indeed, 
$\mtc C_{g}$ has $2(2g +1)$ vertices, 
$3(2g +1)$ edges and $3$ two-bubbles (faces)
it is a ribbon graph that can be drawn on a surface of maximal 
Euler characteristic $3-(3-2)\cdot(2g +1)=2-2g $, or minimal genus $g$. \par

The next step in the proof is to construct, for each genus $g$,
a graph $\mtc T_g \in\feyn$ with $\partial \mtc T_g=\mtc C_g$.
This graph is constructed in two stages. First, in the vertex-set $\mtc C_g\hp 0$, one replaces 
any black vertex $x$ and any white vertex $c$ by the following rule:
\[
c \quad \mapsto\quad \mathfrak a_c=\raisebox{-.45\height}{\includegraphics[width=3cm]{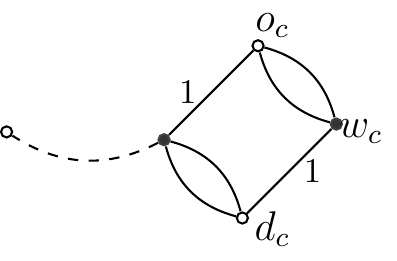}}\qquad \mbox{ and } \qquad
x \quad \mapsto \quad\mathfrak m_x=\raisebox{-.45\height}{\includegraphics[width=3cm]{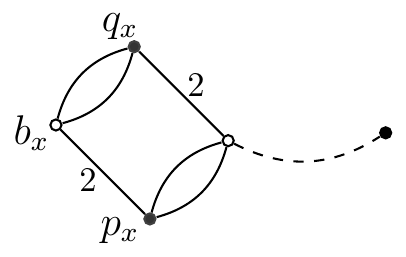}} 
\]
The second stage is to contract with propagators in order to obtain a
well defined element of $\feyn$.
For each $1$ colored line $\overline{cx}$ in $\mtc C_g$ (if it exists) we
join $o_c$ and $p_x$ with a propagator. If there is a $2$-colored edge 
$\overline{cx} $  in $\mtc{C}_g$, one joins $d_c$ with $q_x$ 
with a propagator. It is immediate to see that adjacency by 
an $a$-colored edge $\overline{cx}$ in $\mtc{C}_g$ leads to connectivity
between $\mathfrak{a}_c$ and $\mathfrak{m}_x$ by a $(0a)$-bicolored
path, for $a=1,2$. We want the same property for $a=3$, and actually,
for each $3$-colored line between the given vertices $x$ and $c$, one joins
$w_c$ and $b_x$ with a propagator. Remarkably this does \textit{not} imply the 
$(03)$-path connectedness between $c$ and $x$ but that of their succeeding
vertices of the $2(2g+1)$-gon. Thus, said connectivity between 
a 3-colored edge $\overline{cx}$ is provided by a path composed of nine
edges passing through a propagator between 
$w_d$ and $b_y$, where $d$ (resp. $y$) is the white 
(resp. black) vertex in the polygon, succeeding $c$ (resp.  $x$), as shown in Figure \ref{fig:03}.
\begin{figure}%[H] 
\centering
\includegraphics[width=.68\textwidth]{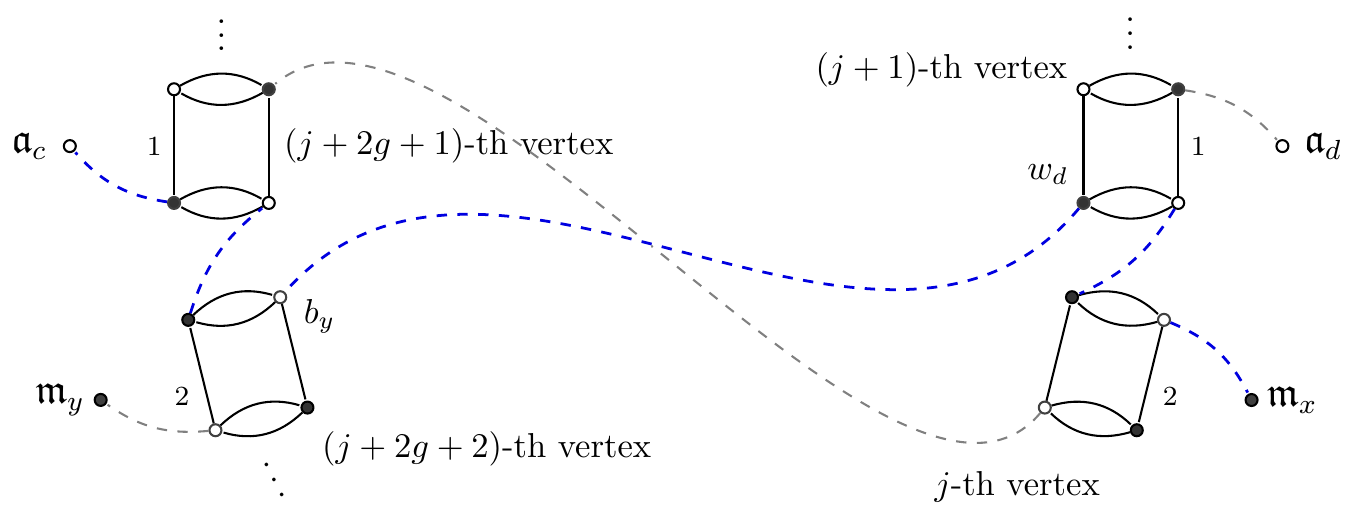}
\caption{If there is a $3$-colored edge $e\in \mtc C_g$ \label{fig:03} one contracts $w_{s(e)}$ and $b_{t(e)}$.
When one does this for all the $3$-colored edges, one guarantees that, 
given two vertices,  $c$ and $x$ joined by a $3$-colored edge in $\mtc C_g$, 
there is a propagator between $w_d$ and $b_y$, where $d$ comes after 
to $x$ and $y$ comes after $c$, anticlockwise. The (03)-colored
path between the external vertices of $\mtf {a}_c$ and $\mtf{m}_x$ passes through  $\overline{w_db_y}$. }
\end{figure}
The graph assembled by applying these two steps to each vertex and edge of $\mtc C_g$ is called 
$\mtc T_g$ and, by construction, it satisfies $\partial \mtc T_g= \mtc C_g$. This proves the result.
\end{proof}
Some $\mtc C_g$ for higher
genera are depicted in Fig. \ref{fig:canonical}. 
As example of this construction, $\mtc T_1$ is the 
leftmost graph in eq. \eqref{eq:integraciondek33}. 
The  more complex $\mtc T_4$ is shown in Figure \ref{fig:T4}. 
We observe that the set $\mathbf T=\{\Tg 0,\mathcal T_1,\mathcal T_2,\ldots\} \subset
\feyn$ endowed with the contraction $\tg$ is a monoid and the
restriction $(\Delta\circ \partial) |_\mathbf T$ to that set is a
monoid-morphism  $(\mathbf T,\tg )\to (\Riem,\#)\,.$ 

\begin{figure}[t]
\begin{subfigure}{0.48\textwidth}
\centering
\includegraphics[width=.9\textwidth]{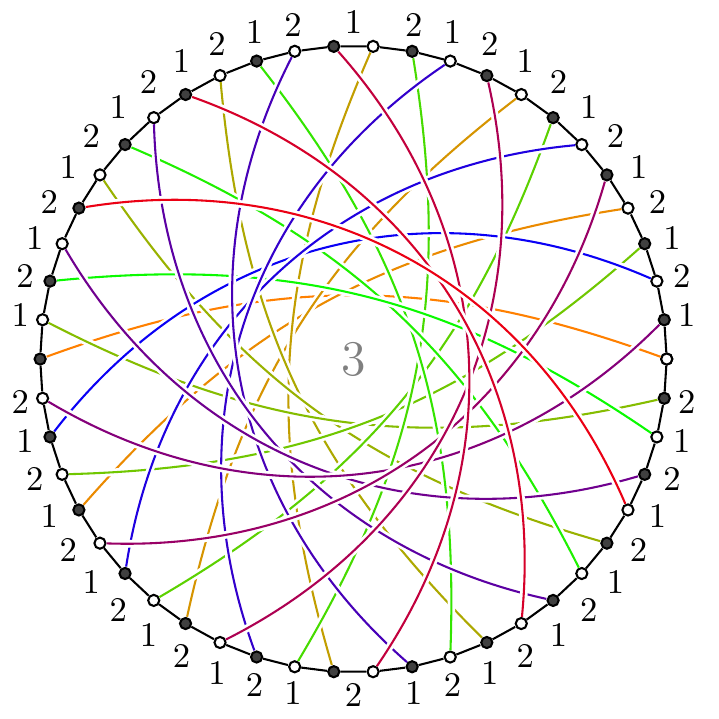}
 \caption{$\mathcal C_{12}$ (labels are on edges)} \label{fig:canonical12}
\end{subfigure}
\hspace*{\fill}
\begin{subfigure}{0.48\textwidth}
\centering
\includegraphics[width=2.9cm]{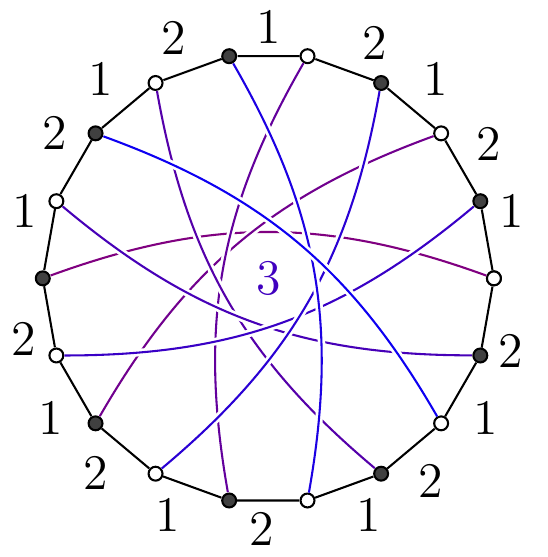} \\
 \caption{$\mathcal C_4$}   \label{fig:genus4}
\vspace{.3cm}
\includegraphics[width=.7\textwidth]{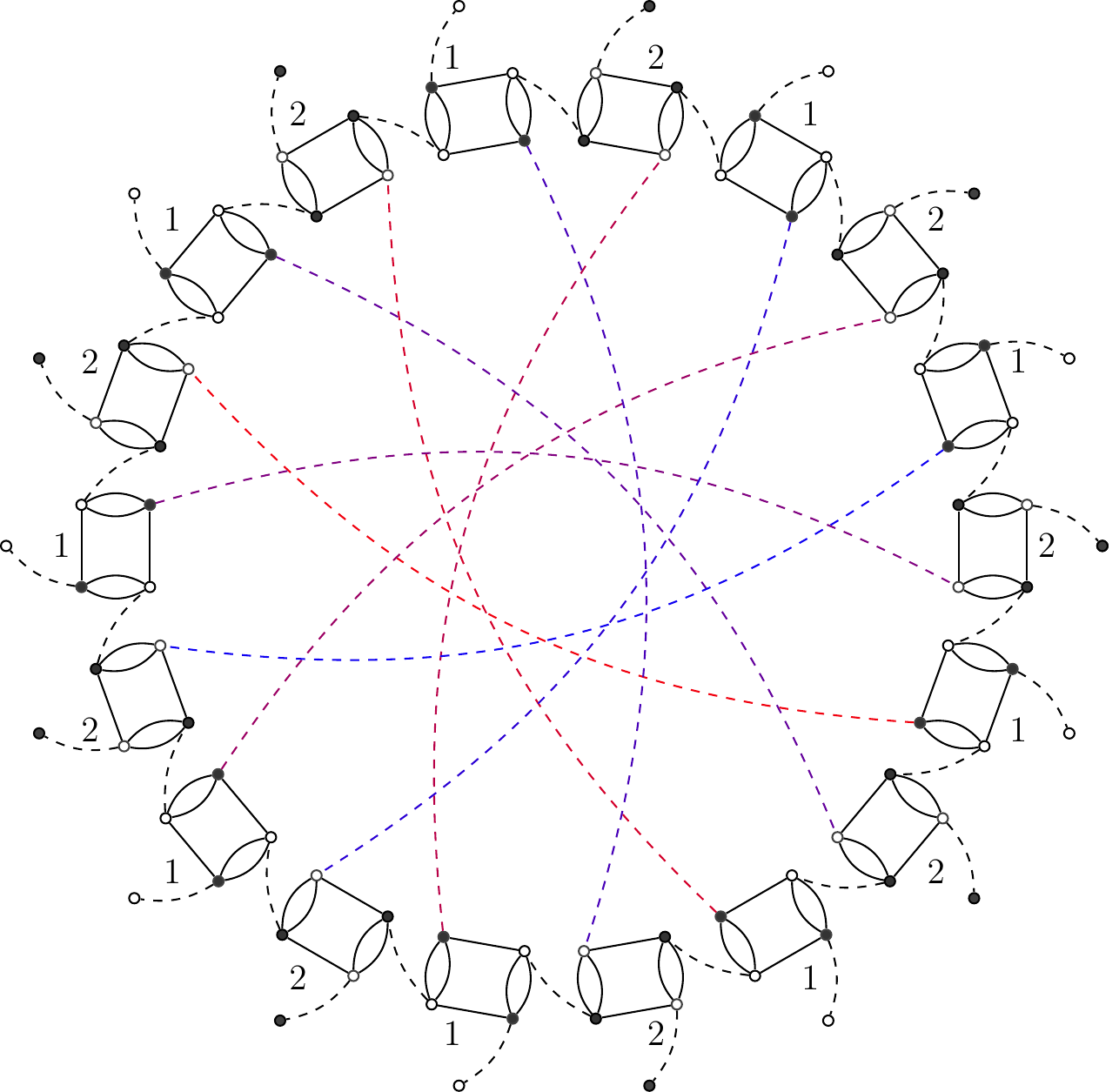} 
 \caption{$\mtc T_4$} \label{fig:T4}
\end{subfigure}

\caption{Examples of canonical graphs $\mathcal C_g$, 
for $g=4$, based on an octadecagon, and  for $g=12$ on a pentacontagon (a).
The graph in (c)  satisfies $\partial \mathcal T_4 =\mtc C_4$; see Lemma \ref{thm:Tg}
(colors in the online version only as visual guide).}
\label{fig:canonical}
\end{figure} 
\begin{thm}\label{thm:3cobord}
The map $\theta$ defined by
\[
\theta=\Delta\circ \partial:\feyn \stackrel{\partial}{\ToLong} 
\amalg \Grph{3} \stackrel{\Delta}{\ToLong}  \Riemd
\]
is surjective. That is, all closed, orientable  (possibly disconnected) Riemannian surfaces
are cobordant via a manifold triangulated by 
(\textit{connected} graphs of) the $\phi^4_3$-theory.
\end{thm} 

\begin{proof}       
If $M$ has $b$ boundary components, 
there are $b$ (not necessarily different) integers $g_1,\ldots,g_b \in \Z_{\geq 0}$ such that
\begin{equation}
M \simeq \Sigma^{g_1}\sqcup  \Sigma^{g_2}
\sqcup \cdots  \sqcup \Sigma^{g_b}\,. \label{recon}
\end{equation}  
We construct a graph $\mtc L_{g_1,\ldots,g_b}\in \feyn$ with 
$\Delta(\partial  \G)$ homeomorphic to $M$. We consider the following sum:
\[
 \mtc L_{g_1,\ldots,g_b}:= (\mtc T_{g_1})  \tensor*[_{f_1}]{\#}{_{\!k}} 
 \mtc P \tensor*[_{l}]{\#}{_{\!e_2}}  (\mtc T_{g_2}) \tensor*[_{f_2}]{\#}{_{\! k}}  \mtc P  
\ldots \mtc P \tensor*[_{l}]{\#}{_{\!e_b}}   (\mtc T_{g_b}) \,.
\]
Notice that each $\mtc T_{g_i}$ has more than three internal propagators 
and we choose two arbitrary $0$-colored edges $e_i,f_i$ of  $\mtc T_{g_i}$ to
perform the connected sum. Because all its summands are in $\feyn$, so is $\mtc L_{g_1,\ldots,g_b}$. We suppress the 
edge dependence. Finally, $\mtc L_{g_1,\ldots,g_b}$ satisfies 
\begin{align*}\partial \mtc L_{g_1,\ldots,g_b} &= \partial ((\mtc T_{g_1}) \# \mtc P \#  (\mtc T_{g_2})  \# \mtc P  
\ldots \mtc P \#  (\mtc T_{g_b})) =
\partial ( \mtc T_{g_1}) \sqcup \partial( (\mtc T_{g_2})  \# \mtc P  
\ldots \mtc P \#  (\mtc T_{g_b}))= \ldots \\
&=  \partial ( \mtc T_{g_1})\sqcup  
 \cdots \sqcup  \partial (\mtc T_{g_b})
  =
\mtc C_{g_1}\sqcup  
 \cdots \sqcup \mtc C_{g_b}\,,
\end{align*}
by applying $b-1$ times Lemma \ref{thm:sepa} and $b$ times Lemma \ref{thm:Tg}. 
Thus $\theta ( \mtc L_{g_1,\ldots,g_b})=\Delta (\partial \mtc L_{g_1,\ldots,g_b}) \simeq M$, which proves the theorem. \end{proof} 
As example of how this construction works, a bordism $\Sigma^2\to \Sigma^3$
is shown in Figure \ref{fig:ontheproof}.

\begin{figure}
\begin{subfigure}{ \textwidth}  \centering
\includegraphics[width=.8\textwidth]{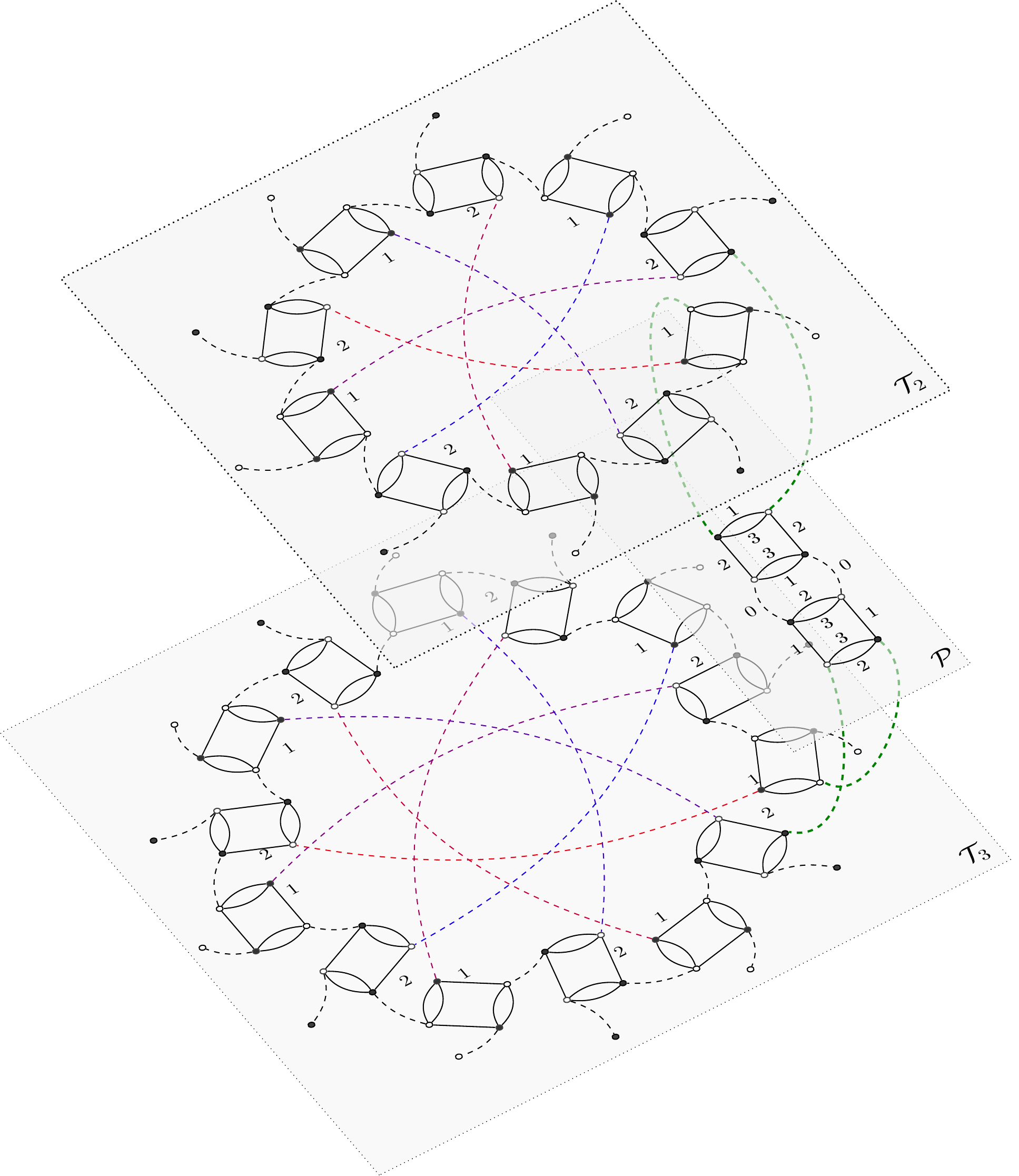}
\caption{The graph $\mathcal L_{2,3}$\label{fig:t1}  }
\end{subfigure}

\bigskip

\begin{subfigure}{ \textwidth}\centering
\includegraphics[width=7.5cm]{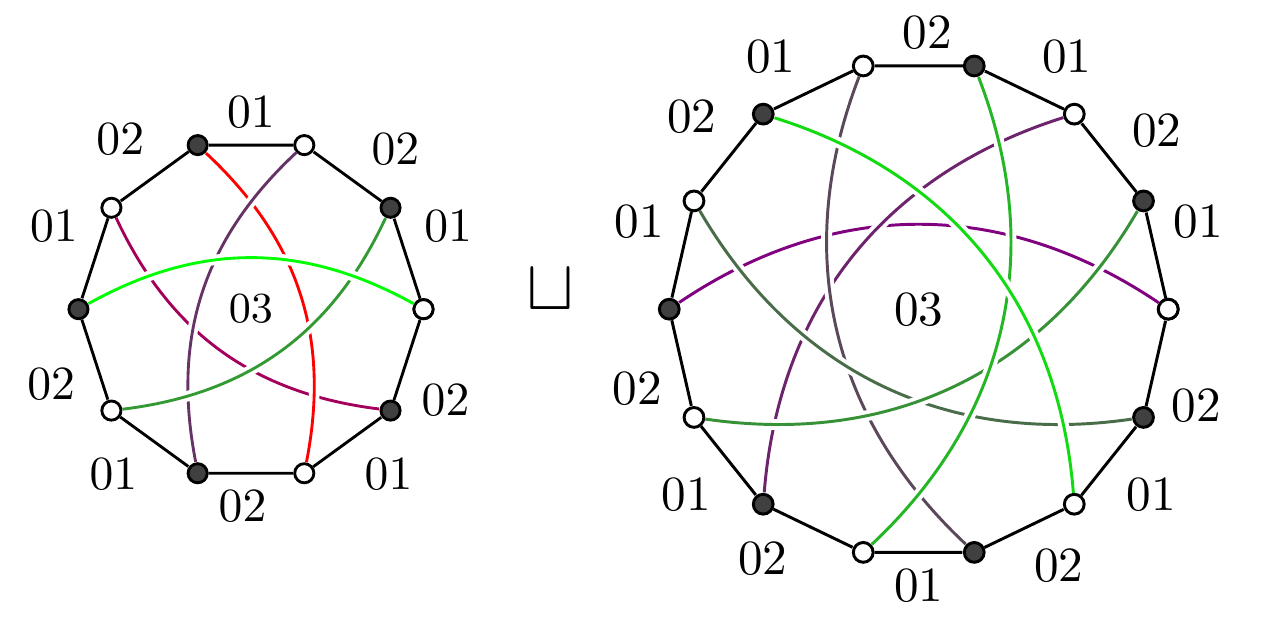}
\caption{The graph $\partial \mathcal L_{2,3} = \mathcal C_2\sqcup \mathcal C_3$
\label{fig:boundary_t1} }
\end{subfigure}

\caption{\label{fig:ontheproof} 
Example of a bordism $\Sigma^2\to \Sigma^3$ triangulated by a Feynman graph  $\mathcal L_{2,3}$ of 
the $\phi_3^4$-theory, as given in the proof of 
Theorem \ref{thm:3cobord} (coloring of the lines in the on line version only intended as visual guide).} 
\end{figure}

%% 
%%  = = = = = = = = = = = = = = = = = = = = = = = = = = = = = = = = = =  
%% 
%%  = = = = = = = = = = = = = = = = = = = = = = = = = = = = = = = = = =  
%% 
%%  = = = = = = = = = = = = = = = = = = = = = = = = = = = = = = = = = =  

\section{Conclusions}

We defined the connected sum of $3$-colored graphs that 
is a well defined operation on the set of Feynman diagrams 
of any tensor model. It differs from the existent connected 
sum in the crystallization theory by a $1$-dipole move.
There is no tensor model $V$ such that the latter operation 
restricts to a well defined binary operation on
$\fey_2(V)$, whence the need of the connected sum  
we introduced. It is used 
to prove the surjectivity of the map $\xi$ 
in the following commuting diagram:
\begin{equation} \label{condition}
\hspace{-1.5cm}
\raisebox{-.4\height}{\includegraphics[height=1.95cm]{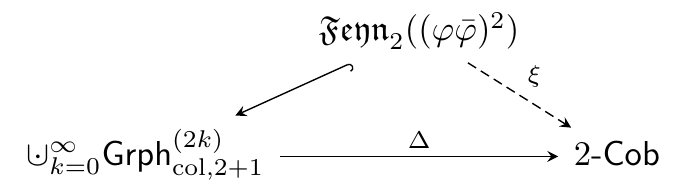}}
\end{equation}
(see Thm. \ref{thm:gen_bordisms}). Tangentially, this might provide some link between the rank-$2$
tensor models and Atiyah's Topological Quantum Field Theories \cite{Atiyah},
where one studies functors from $2$-$\Cob$ to the category of 
Hilbert spaces. A particular case of \eqref{condition} is 
$\xi|$ in the following commutative diagram:
\begin{equation}   
\raisebox{-.4\height}{\includegraphics[height=1.85cm]{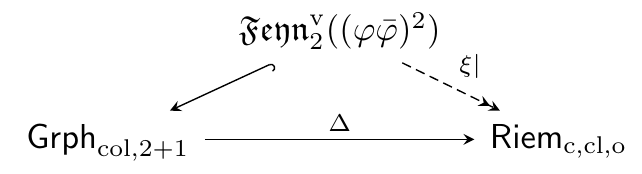}}
\end{equation}

Trivially, CTM-graphs that represent a (sub)category of boundaryless
manifolds (here $\Riem$) can be exhibited as null-bordant, by just
coning each graph. The non trivial part is showing that, in this case,
any surface in $\Riem$ is null-bordant via a suitable graph in
$\feyn$, which we constructed.  Moreover, any two surfaces in $\Riemd$
(even disconnected) are also cobordant in the sense of the
$\phi^4_3$-theory:
\begin{equation}   
\raisebox{-.4\height}{\includegraphics[height=1.85cm]{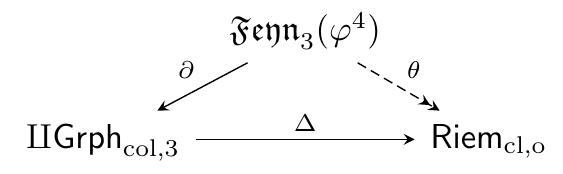}}
\end{equation}
 An immediate
consequence concerns the non-perturbative \cite{fullward} treatment of
the Ward Identity for rank-$3$ tensor models \cite{us,DineWard}.  In
order to undertake that problem, one needs the expansion of the free
energy $\log Z[J,\bar J] $ in boundary graphs, analogous to in the
matrix model case \cite[Sec. 2.3]{GW12}.  The fact that the boundary sector $\partial \feyn$, 
as proven here, generates all of $\Riemd$ facilitates our
Ansatz regarding this expansion.  \par The divergence degree of
graphs can be expressed in terms of certain degree 
$\tilde\omega$ for open graphs, 
that extends Gur\u au's degree $\omega$. The 
degree for open graphs is given in terms of the so called pinched jackets
of $\G$.  For us, it is important that the genus and Gur\u{a}u's
degree are the same for $3$-colored graphs. One has then for
an open Feynman graph $\G$,
$ \tilde\omega(\G) \geq 3\cdot \omega (\partial \G) = 3
\sum_{\R\subset \partial\G} g(\R),$
summing\footnote{The inequality has been proven in \cite[Lemma
  4]{3Dbeta} (see also \cite[Cor. B4]{wgauge}); for the equality, see
  here Example \ref{thm:exampleribbon}.} over the the connected
components of the boundary graph $\partial \G$. The present work helps
to compute the boundary graph $\partial \G$, and hence to have a lower-bound for the
degree $\tilde\omega(\G)$.  \par In \cite{fullward} it will be proven that
the correlation functions $G_\B$ of the $\phi^4_3$-model are indexed
by boundary graphs $\B\in\partial \feyn$.  It will be useful to expand
these functions in Gur\u{a}u's degree, as done in the matrix
theory-formulation of the ($\Omega=1$)-Grosse-Wulkenhaar
$\phi^4$-model \cite{GW12} in terms of the genus.  In there, using such
expansion, combined with the a full Ward identity and the
Schwinger-Dyson equations yielded a closed equations for
correlation functions and that techniques will be extended to 
the present setting.  \par
As another immediate application, the natural continuation of this
work is to relax some of the symmetry and to pose Question
\ref{eqn:primed} in the framework of multi-orientable
\cite{reviewTanasa} or $\mathrm O(N)$-tensor models \cite{On}. \par

A second, quite different application is the addition of bosonic
fields. In dimension two, for instance, using Theorem \ref{thm:gen_bordisms}
and constructs before it, one has control of the gluings' topology, even of those 
made of a large number of interaction vertices. This can both ease
computations and might be reused to define gauge theories on (computable) random spaces. 
An approach is, first, to adapt the gauge theory on
usual graphs \`{a} la Baez \cite{baez} to our colored
graphs. Secondly, one would add gauge fields \`a la Marcolli-van
Suijlekom \cite{MvS} using representation of graphs (here tensor-model
Feynman graphs used as random-`manifold base') in the category of
finite dimensional spectral triples with vanishing Dirac operator. In
that respect, the connection between noncommutative geometry and
matrix models would be based on recent results by Barrett and Glaser
\cite{barrett}, which treat the \textit{quantum} Connes-Chamseddine
spectral action \cite{SAP,CCM} as a certain matrix model.

%% 
%%  = = = = = = = = = = = = = = = = = = = = = = = = = = = = = = = = = =  
%% 
%%  = = = = = = = = = = = = = = = = = = = = = = = = = = = = = = = = = =  
%% 
%%  = = = = = = = = = = = = = = = = = = = = = = = = = = = = = = = = = =  
 
\SkipTocEntry\section*{Acknowledgement}
The author wishes to thank:
\begin{itemize}\setlength\itemsep{-.51pt}
 \item[\scriptsize$\bullet$] The \textit{Deutscher Akademischer
   Austauschdienst} (DAAD) mainly, but also 
   the \textit{Sonderforschungsbereich 878} ``Groups, Geometry \& Actions''
(SFB 878), for financial support.
 \item[\scriptsize$\bullet$] Paola Cristofori, Joseph Ben Geloun, Raimar Wulkenhaar for useful comments
 and Adrian Tanas\u{a} for pointing out the multi-orientable tensor models.
 \item[\scriptsize$\bullet$] The Erwin Schr\"odinger International
Institute for Mathematical Physics, Vienna, for 
hospitality during 
the ESI-Program {``The interrelation between mathematical physics, 
number theory and non-commutative geometry''}.
\end{itemize}

\appendix

%% 
%%  = = = = = = = = = = = = = = = = = = = = = = = = = = = = = = = = = =  
%% 
%%  = = = = = = = = = = = = = = = = = = = = = = = = = = = = = = = = = =  
%% 
%%  = = = = = = = = = = = = = = = = = = = = = = = = = = = = = = = = = =  
 
 \small
\section{Computing homology of colored graphs}\label{appA}

\begin{example} \label{homtorus} 
To compute its bubble-homology, as proposed in Ex. \ref{thm:exampleribbon} one 
chooses an (ordered) basis for each dimension according to
following labels:
\[
\begin{minipage}{.43\linewidth}
  \centering
  $ 
\R_1=\raisebox{-0.48\height}{\includegraphics[width=4.0cm]
{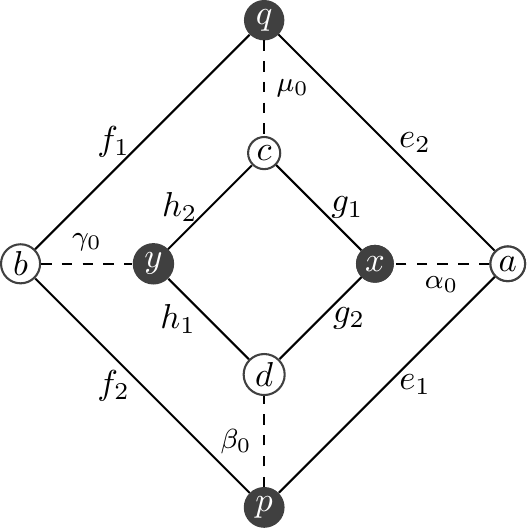} } $
\end{minipage}% 
\begin{minipage}{.55\linewidth}
  \centering
  $\begin{array}{cl}
C_0(\R_1) =&\!\!\langle a,b,c,d,p,q,x,y\rangle_\Z \,, \\
 C_1(\R_1) =&\!\!\langle e_1,e_2,f_1,f_2,g_1,g_2,h_1,h_2,
\alpha_0,\beta_0,\gamma_0,\mu_0\rangle_\Z  \,, \\
 C_2(\R_1) = & \!\!\langle \B^{01},\B^{02}, 
 \B^{12}_{\mtr{outside}},\B^{12}_{\mtr{inside}}\rangle_\Z \,.
  \end{array}$ 
\end{minipage}
\]
In the chain complex
$
 0\to \Z^{4}\simeq C_{2}(\R_1)   \stackrel{\partial_{2}}{\ToLong}  \Z^{12} 
 \simeq C_{1}(\R_1) \stackrel{\partial_{1}}{\ToLong}    C_0(\R_1)\simeq \Z^{8} {\to}\, 0
$
the non-trivial boundary operators are, in the chosen bases, given by:
\[
\dpa_1=
\left(
%\begin{array}{cccccccccccc}
\begin{smallmatrix}
-1 & -1 & 0 & 0 & 0 & 0 & 0 & 0 & -1 & 0 & 0 & 0 \\
 0 & 0 & -1 & -1 & 0 & 0 & 0 & 0 & 0 & 0 & -1 & 0 \\
 0 & 0 & 0 & 0 & -1 & 0 & 0 & -1 & 0 & 0 & 0 & -1 \\
 0 & 0 & 0 & 0 & 0 & -1 & -1 & 0 & 0 & -1 & 0 & 0 \\
 1 & 0 & 0 & 1 & 0 & 0 & 0 & 0 & 0 & 1 & 0 & 0 \\
 0 & 1 & 1 & 0 & 0 & 0 & 0 & 0 & 0 & 0 & 0 & 1 \\
 0 & 0 & 0 & 0 & 1 & 1 & 0 & 0 & 1 & 0 & 0 & 0 \\
 0 & 0 & 0 & 0 & 0 & 0 & 1 & 1 & 0 & 0 & 1 & 0 \\
\end{smallmatrix}
%\end{array}
\right)
\sim
%\begin{array}{cccccccccccc}
 \left(
 \begin{smallmatrix}
 1 & 0 & 0 & 1 & 0 & 0 & 0 & 0 & 0 & 1 & 0 & 0 \\
 0 & 1 & 0 & -1 & 0 & 0 & 0 & 0 & 0 & 0 & -1 & 1 \\
 0 & 0 & 1 & 1 & 0 & 0 & 0 & 0 & 0 & 0 & 1 & 0 \\
 0 & 0 & 0 & 0 & 1 & 0 & 0 & 1 & 0 & 0 & 0 & 1 \\
 0 & 0 & 0 & 0 & 0 & 1 & 0 & -1 & 0 & 1 & -1 & 0 \\
 0 & 0 & 0 & 0 & 0 & 0 & 1 & 1 & 0 & 0 & 1 & 0 \\
 0 & 0 & 0 & 0 & 0 & 0 & 0 & 0 & 1 & -1 & 1 & -1 \\
 0 & 0 & 0 & 0 & 0 & 0 & 0 & 0 & 0 & 0 & 0 & 0 \\
 \end{smallmatrix}
\right)
\]
%\end{array}
and 
\[
\dpa_2=\left(
%\begin{array}{cccc}
\begin{smallmatrix}
1 & 0 & -1 & 0 \\
 0 & 1 & 1 & 0 \\
 1 & 0 & -1 & 0 \\
 0 & 1 & 1 & 0 \\
 1 & 0 & 0 & -1 \\
 0 & 1 & 0 & 1 \\
 1 & 0 & 0 & -1 \\
 0 & 1 & 0 & 1 \\
 -1 & -1 & 0 & 0 \\
 -1 & -1 & 0 & 0 \\
 -1 & -1 & 0 & 0 \\
 -1 & -1 & 0 & 0 \\
\end{smallmatrix}
 %\end{array}
\right)  \sim
%\]
%The column reduction of $\dpa_2$ is
%\[
\left(
%\begin{array}{cccc}
\begin{smallmatrix}
 1 & 0 & 0 & 0 \\
 0 & 1 & 0 & 0 \\
 1 & 0 & 0 & 0 \\
 0 & 1 & 0 & 0 \\
 0 & 0 & 1 & 0 \\
 1 & 1 & -1 & 0 \\
 0 & 0 & 1 & 0 \\
 1 & 1 & -1 & 0 \\
 -1 & -1 & 0 & 0 \\
 -1 & -1 & 0 & 0 \\
 -1 & -1 & 0 & 0 \\
 -1 & -1 & 0 & 0 \\
%\end{array}
\end{smallmatrix}
\right) \, , \]
where the tilde means a change of basis, which 
in each case leads to row or column reduction. 
Since $\dpa_3=0$, $H_2(\R_1)=\ker\,\dpa_2=\Z$. On the other hand
the column reduction of $\dpa_1$ is
\[
\dpa_1M=
\left(
\begin{smallmatrix}
%\begin{array}{cccccccccccc}
 1 & 0 & 0 & 0 & 0 & 0 & 0 & 0 & 0 & 0 & 0 & 0 \\
 0 & 1 & 0 & 0 & 0 & 0 & 0 & 0 & 0 & 0 & 0 & 0 \\
 0 & 0 & 1 & 0 & 0 & 0 & 0 & 0 & 0 & 0 & 0 & 0 \\
 0 & 0 & 0 & 1 & 0 & 0 & 0 & 0 & 0 & 0 & 0 & 0 \\
 0 & 0 & 0 & 0 & 1 & 0 & 0 & 0 & 0 & 0 & 0 & 0 \\
 0 & 0 & 0 & 0 & 0 & 1 & 0 & 0 & 0 & 0 & 0 & 0 \\
 0 & 0 & 0 & 0 & 0 & 0 & 1 & 0 & 0 & 0 & 0 & 0 \\
 -1 & -1 & -1 & -1 & -1 & -1 & -1 & 0 & 0 & 0 & 0 & 0 \\
%\end{array}
\end{smallmatrix}
\right)
\]
where 
\[M=\left(
\begin{smallmatrix}
%\begin{array}{cccccccccccc}
 0 & 0 & 0 & 0 & 1 & 0 & 0 & 0 & -1 & -1 & 0 & 0 \\
 0 & 1 & 0 & 0 & 0 & 1 & 0 & 0 & 1 & 0 & 1 & -1 \\
 0 & -1 & 0 & 0 & 0 & 0 & 0 & 0 & -1 & 0 & -1 & 0 \\
 0 & 0 & 0 & 0 & 0 & 0 & 0 & 0 & 1 & 0 & 0 & 0 \\
 0 & 0 & -1 & 0 & 0 & 0 & 0 & -1 & 0 & 0 & 0 & -1 \\
 1 & 1 & 1 & 0 & 1 & 1 & 1 & 1 & 0 & -1 & 1 & 0 \\
 -1 & -1 & -1 & -1 & -1 & -1 & -1 & -1 & 0 & 0 & -1 & 0 \\
 0 & 0 & 0 & 0 & 0 & 0 & 0 & 1 & 0 & 0 & 0 & 0 \\
 -1 & -1 & 0 & 0 & -1 & -1 & 0 & 0 & 0 & 1 & -1 & 1 \\
 0 & 0 & 0 & 0 & 0 & 0 & 0 & 0 & 0 & 1 & 0 & 0 \\
 0 & 0 & 0 & 0 & 0 & 0 & 0 & 0 & 0 & 0 & 1 & 0 \\
 0 & 0 & 0 & 0 & 0 & 0 & 0 & 0 & 0 & 0 & 0 & 1 \\
%\end{array}
\end{smallmatrix}
\right)\mathrm{.\,Then \,}
%\]
%Then \[
M\inv\dpa_2=\left(
\begin{smallmatrix}
%\begin{array}{cccc}
 0 & 0 & 0 & 0 \\
 0 & 0 & 0 & 0 \\
 0 & 0 & 0 & 0 \\
 0 & 0 & 0 & 0 \\
 0 & 0 & 0 & 0 \\
 0 & 0 & 0 & 0 \\
 0 & 0 & 0 & 0 \\
 0 & 1 & 0 & 1 \\
 0 & 1 & 1 & 0 \\
 -1 & -1 & 0 & 0 \\
 -1 & -1 & 0 & 0 \\
 -1 & -1 & 0 & 0 \\
 \end{smallmatrix}
%\end{array}
\right)
%\stackrel{\mathrm{row.red.}}{\OldMapsto}%,  BM\inv\dpa_2=
\]
The last zero-columns of $\dpa_1$ correspond to the generators of $\ker \,\dpa_1=\Z^5$.
From those generators, which in 
the matrix $M\inv \dpa_2$ correspond to the five last rows,
three of them ---the non-zero rows corresponding to
the row reduction of the (non-zero lower part of $M\inv \partial_2$)
\[\left(
\begin{smallmatrix}
%\begin{array}{cccc}
 1 & 0 & 0 & -1 \\
 0 & 1 & 0 & 1 \\
 0 & 0 & 1 & -1 \\
 0 & 0 & 0 & 0 \\
 0 & 0 & 0 & 0 \\
 \end{smallmatrix}
%\end{array}
\right)\]
---lie also in the image of $\dpa_2$. It follows $H_1(\R_1)=\Z^{5}/\Z^{3}=\Z^{2}.$

\end{example}

\begin{example} \label{homonecklace}
In order to compute the homology of the the complex, 
one labels the graph
\[
\begin{minipage}{.43\linewidth}
  \centering
  $ 
\G=\raisebox{-0.48\height}{
\includegraphics[width=2.7cm]{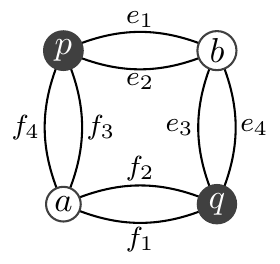}} $
\end{minipage}% 
\begin{minipage}{.55\linewidth}
  \centering
  $\begin{array}{cl}
C_0(\G) =&\!\!\langle a,p,b,q\rangle_\Z \\
\vspace{0.2cm}
 C_1(\G) =&\!\!\langle e_1,e_2,e_3,e_4,f_1,f_2,f_3,f_4\rangle_\Z  \\
 \vspace{0.1cm}
 C_2(\G) = & \!\!\langle  \B^{12}_{e},\B^{12}_{f}, 
  \B^{23},\B^{24}, \B^{13},\B^{14} ,\B^{34}_{e},\B^{34}_{f}\rangle_\Z \\
  \vspace{0.1cm}
 C_3(\G) = &\!\!\langle \B^{\hat 1} , \B^{\hat 2}, \B^{\hat 3} , \B^{\hat 4} \rangle_\Z
  \end{array}$
\end{minipage}
\]
where $\B^{\hat c}$ means omission of the color $c$. The differentials
of the chain complex 
\[
0 \to  C_{3}(\G) \simeq \Z^{4} \stackrel{\partial_{3}}{\ToLong} 
C_{2}(\G)\simeq \Z^{8} \stackrel{\partial_{2}}{\ToLong} 
 C_{1}(\G)\simeq \Z^{8} \stackrel{\partial_{1}}{\ToLong} 
 C_0(\G)\simeq \Z^{4} \to  0
\]
are explicitly:
\begin{equation*}
\dpa_1=
\left(
\begin{smallmatrix}
 0 & 0 & -1 & -1 & -1 & -1 & 0 & 0 \\
 1 & 1 & 1 & 1 & 0 & 0 & 0 & 0 \\
 -1 & -1 & 0 & 0 & 0 & 0 & -1 & -1 \\
 0 & 0 & 0 & 0 & 1 & 1 & 1 & 1 \\
\end{smallmatrix}
\right),\quad 
\dpa_2 =
\left(
 \begin{smallmatrix}
 -1 & 0 & 0 & 0 & -1 & -1 & 0 & 0 \\
 1 & 0 & -1 & -1 & 0 & 0 & 0 & 0 \\
 0 & 0 & 1 & 0 & 1 & 0 & -1 & 0 \\
 0 & 0 & 0 & 1 & 0 & 1 & 1 & 0 \\
 0 & -1 & 0 & 0 & -1 & -1 & 0 & 0 \\
 0 & 1 & -1 & -1 & 0 & 0 & 0 & 0 \\
 0 & 0 & 1 & 0 & 1 & 0 & 0 & -1 \\
 0 & 0 & 0 & 1 & 0 & 1 & 0 & 1 \\
 \end{smallmatrix}
\right),\quad 
\dpa_3 =
\left(
 \begin{smallmatrix}
 0 & 0 & 1 & 1 \\
 0 & 0 & 1 & 1 \\
 1 & 0 & 0 & 1 \\
 -1 & 0 & 1 & 0 \\
 0 & 1 & 0 & -1 \\
 0 & -1 & -1 & 0 \\
 1 & 1 & 0 & 0 \\
 1 & 1 & 0 & 0 \\
 \end{smallmatrix}
\right)
\end{equation*}
To compute $H_2$, the reduced versions are
\[
\dpa_2'=\left(
 \begin{smallmatrix}
 1 & 0 & 0 & 0 & 0 & 0 & 0 & 0 \\
 0 & 1 & 0 & 0 & 0 & 0 & 0 & 0 \\
 0 & 0 & 1 & 0 & 0 & 0 & 0 & 0 \\
 -1 & -1 & -1 & 0 & 0 & 0 & 0 & 0 \\
 0 & 0 & 0 & 1 & 0 & 0 & 0 & 0 \\
 1 & 1 & 0 & -1 & 0 & 0 & 0 & 0 \\
 0 & 0 & 0 & 0 & 1 & 0 & 0 & 0 \\
 -1 & -1 & 0 & 0 & -1 & 0 & 0 & 0 \\
 \end{smallmatrix}
\right)
\qquad
\dpa_3'=\left(
 \begin{smallmatrix}
 0 & 0 & 0 & 0 \\
 0 & 0 & 0 & 0 \\
 0 & 0 & 0 & 0 \\
 0 & 0 & 0 & 0 \\
 0 & 0 & 0 & 0 \\
 0 & -1 & -1 & 0 \\
 0 & 1 & 0 & -1 \\
 1 & 1 & 0 & 0 \\
 \end{smallmatrix}
\right) \sim   \left(
 \begin{smallmatrix}
  0 & 0 & 0 & 0 \\
 0 & 0 & 0 & 0 \\
 0 & 0 & 0 & 0 \\
 0 & 0 & 0 & 0 \\
 0 & 0 & 0 & 0 \\
 1 & 0 & 0 & 1 \\
 0 & 1 & 0 & -1 \\
 0 & 0 & 1 & 1 \\
 \end{smallmatrix}
\right)
\]
The non-zero part of $\dpa_3'$ has rank three,
whence $H_2(\G)=\Z^3/\Z^3=0$. Similarly one
finds $H_1(\G)=0$.
\end{example}

%% 
%%  = = = = = = = = = = = = = = = = = = = = = = = = = = = = = = = = = =  
%% 
%%  = = = = = = = = = = = = = = = = = = = = = = = = = = = = = = = = = =  
%% 
%%  = = = = = = = = = = = = = = = = = = = = = = = = = = = = = = = = = =  
 
\section{The cell complex of a ribbon graph} \label{app_ribbons}

Departing from the abstract definition, we construct here
the cell-decomposition of the minimum-genus-surface where a ribbon graph can be drawn on
without self-intersections. 
To begin with, we remark that vertices do not have naturally an
orientation, but only a cyclic order---so far, these are only
abstract combinatorial objects. However, when one tries to represent
graphs by drawings, thus evoking the orientation of the plane
(counterclockwise), graphs might be given a neater representation if
we invert the cyclic order on some vertices and, of course, keep track
of this action with help of a sign, $\epsilon.$ Write $\epsilon_v=+1$
if we preserve the order of a vertex $v$ as the levorotation, and
$\epsilon_v=-1$ if the cyclic order of $v$ is written as a
dextrorotary vertex. Also, for any ribbon graph $\Rb$, 
since each edge $e$ is determined by two half-edges
$h,h'\in \Rb\hp{1/2}$ that are joined (i.e. $j(h)=h'$), $e$ can be
rewritten as $e=[h,h']=[h',h]$. If $e=[(v,\al),(w,\beta)]$, think of
$e$ as being attached to the vertex $v$ at the $\al$-th place, and to
the vertex $w$ in the $\beta$-th place. \par 

The `fat graph' representation of 
$\Rb$ is in a natural way a cell-complex, $X(\Rb)$. 
The skeleta $X\hp n$ are constructed as
follows: \\ 
\begin{enumerate}
 \item[\textit{$0$-cells:}]  Let $n_v$ be the valence of the vertex $v\in \Rb\hp 0$, 
 and $\epsilon_v\in \{+1,-1\}$ its orientation. 
 We associate $v$ the following cyclic ordered set (see 
 Fig. \ref{fig:ceroCW}) of
 $0$-cells:
\begin{align}
 \qquad \qquad v\to \begin{cases}
      (P(v)_1^+,P(v)_1^-,P(v)_2^+,P(v)_2^-,\ldots,P(v)_{n_v}^+,P(v)_{n_v}^-  )
      & \mbox{if}  \,\, \epsilon_v=+1\,, \\
(P(v)_{1}^+,P(v)_1^-,P(v)_{n_v}^{+},P(v)_{n_v}^{-},\ldots,P(v)_{2}^+,P(v)_{2}^-)
& \mbox{if} \,\,  \epsilon_v=-1\,.
      \end{cases} \label{estructuraciclica}
\end{align}
The $0$-skeleton $X\hp0$ is then 
 the union of all such points $P(v)^\epsilon_\alpha$,
 with $v$ running all over $\Rb \hp0$, $\epsilon=\pm$ 
 and $\alpha=1,\ldots,n_v$.
 One has thus, 
 in total, $2\sum_{v \in \Rb\hp 0} n_v$ $0$-cells. 
 
\begin{figure}%[H]
\begin{subfigure}{0.45\textwidth}
\centering
\includegraphics[width=4.2cm]{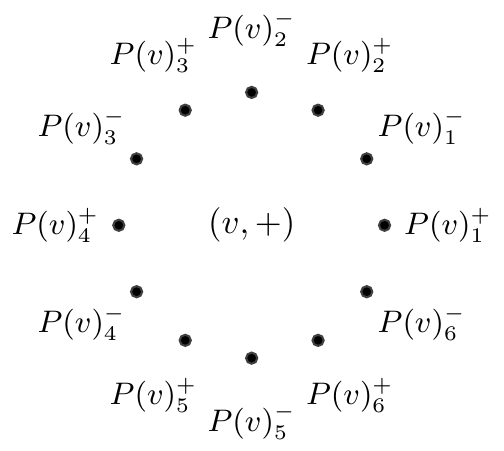}
 \caption{ \small  For a positively orientable vertex $(v,+)\in\Rb\hp0$,
 with $n_v=6$.} \label{fig:ceroCWA}
\end{subfigure}
\hspace*{\fill} 
\begin{subfigure}{0.45\textwidth}
\centering
\includegraphics[width=4.2cm]{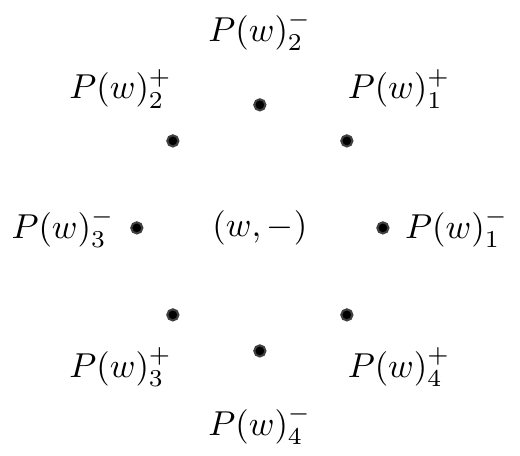}
 \caption{  \small  For a, say, valence-$4$ vertex, the
 associated $0$-cells.} \label{fig:ceroCWB}
\end{subfigure}
\caption{      \small
On the construction of the $0$-skeleton, $X\hp0$ of $\Rb$. 
}   \label{fig:ceroCW}
\end{figure}

\item[\textit{$1$-cells:}]

 In order to construct $X\hp 1$, we proceed in two 
 steps:
 \begin{itemize}
  \item[i)] First, add a $1$-cell $\overline{pq}$ to
 for each consecutive pair of points $p$ and $q$,
 with the order given by the cycles \eqref{estructuraciclica}. That is, 
 add for each vertex $v$ the following cells:
 \begin{align} \label{uno_circulo}
\qquad \overline{P(v)_1^+ P(v)_1^-},&\,\overline{P(v)_1^-P(v)_2^+}, 
\ldots, \overline{P(v)_{n_v}^+P(v)_{n_v\vphantom 1}^{-\vphantom +}}, 
\overline {P(v)_{n_v}^{-}  P(v)_1^+} &&(\epsilon_v=+1) \!\!\! \\
  \qquad  \overline{P(v)_{1}^+P(v)_1^-},&\,\overline{P(v)_1^-P(v)_{n_v}^{+}},\ldots, 
   \overline{P(v)_{2}^+  P(v)_{2}^-},\, \overline{ P(v)_{2}^-P(v)_{1}^+} && (\epsilon_v=-1)
   \!\!\!
 \end{align}
 This results in the 
 space $\sqcup_{v\in \Rb\hp0} (\mathbb{S}^1_v)^\pm$,
 where the circles is given 
 the orientation $\pm$ of the vertices.  
 \item[ii)] 
 The second step is attaching the ribbons:
 for each edge $e\in \Rb\hp1$ with $e=[(x,\al),(y,\beta)]$ 
 (viz. connecting the vertex $x$ at the $\al$-th place,
 with $y$ at the $\beta$-th place),
 attach $1$-cells from
 $P(x)_\alpha^+$ to $P(y)_\beta^-$
 and from $P(y)_\beta^+$ to  $P(x)_\alpha^-$ 
(see Fig. \hyperref[fig:listonesB]{\ref*{fig:listones}\textsc{ (b)}} and Fig. \ref{fig:cw1}). 
 Notice that the same edges are attached if
 we instead take the pair $((y,\beta),(x,\al))$
 as representative.
 Since for each vertex $v$ we attached double lines, 
 the whole number of attached $1$-cells 
 is $2\sum_{v \in \Rb\hp 0} n_v+ 2 |\Rb \hp 1| $.
 \end{itemize}

 \begin{figure}  %[H]
 \centering
 \includegraphics[height=3.1cm]{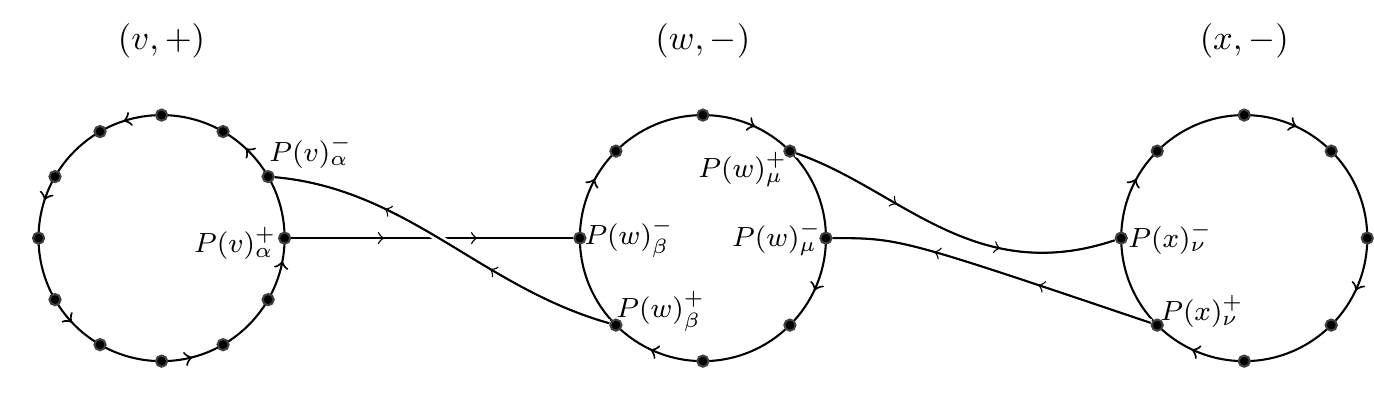}
 \caption{ \small Shows the most general case on how  
  to adjoint $1$-cells. First, to the cyclic 
  structure in eq. \eqref{estructuraciclica} (in the figure, 
  the circles). Here 
  the $1$-cells attached to an edge $e$ that 
  connects with $e=[(v,\alpha),(w,\beta)]$;
  the opposite orientation is responsible for the apparent crossing.
   On the right, the two $1$-cells associated to $f \in \Rb \hp 1$,
 if $ f=[(w,\mu),(x,\nu)]$.\label{fig:cw1}
  }
  \end{figure}
\item[\textit{$2$-cells.}] The last skeleton $X\hp2$ is 
 obtained in two steps:
 \begin{itemize}
\item[a)]filling the ribbon double lines: that is, if 
 $e\in \Rb\hp 1$ and $e=[(v,\alpha),(w,\beta)]$,
 $2$-disk attachment at the loop formed by 
 $\overline{P(v)_{\alpha}^+P(v)_\alpha^-}$,  $ \overline{P(w)_{\beta}^+P(w)_\beta^-}$  
 and the two ribbon segments constructed for $e$ in step ii) above.
 \item[b)]
 The second step is filling for all $v\in \Rb\hp 0$ the
 vertex-circles $\mathbb S^1_v$ which
 one gets by \eqref{uno_circulo}. In total, we added $|\Rb\hp 1|+|\Rb\hp 0|$ 
 $2$-cells.
  \end{itemize}
\end{enumerate}
  This exhibits the cell-structure $X\hp 2$ of a ribbon graph. But
  actually is more natural not to stop at $X\hp 2$ and to adjoint more
  $2$-cells to some loops left, namely the boundary components.  A
  \textit{boundary component} of the graph $\Rb$ here is a loop of the
  graph formed by the boundary of the ribbons' long segments and arcs
  determined by the orientation, as pictured in Figure
  \hyperref[fig:loops_ribbon]{\ref*{fig:bc} \textsc{(a)}} (see also
  Fig. \hyperref[fig:listonesB]{\ref*{fig:listones} \textsc{(b)}}).
  \\
\begin{figure}  \centering
\begin{subfigure}{0.40\textwidth}
\includegraphics[height=2.4cm]{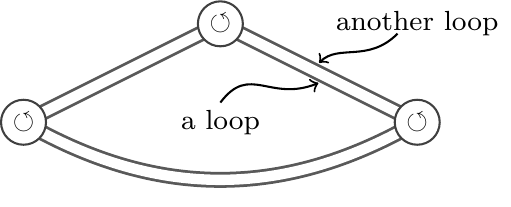}
\caption{  \small    Example of 
boundary components of a ribbon graph. 
Here $\mathrm{bc}(\Rb)=2$.\label{fig:loops_ribbon}
} 
\end{subfigure}
\hspace*{\fill} 
\begin{subfigure}{0.450\textwidth}
\includegraphics[height=4.2cm]{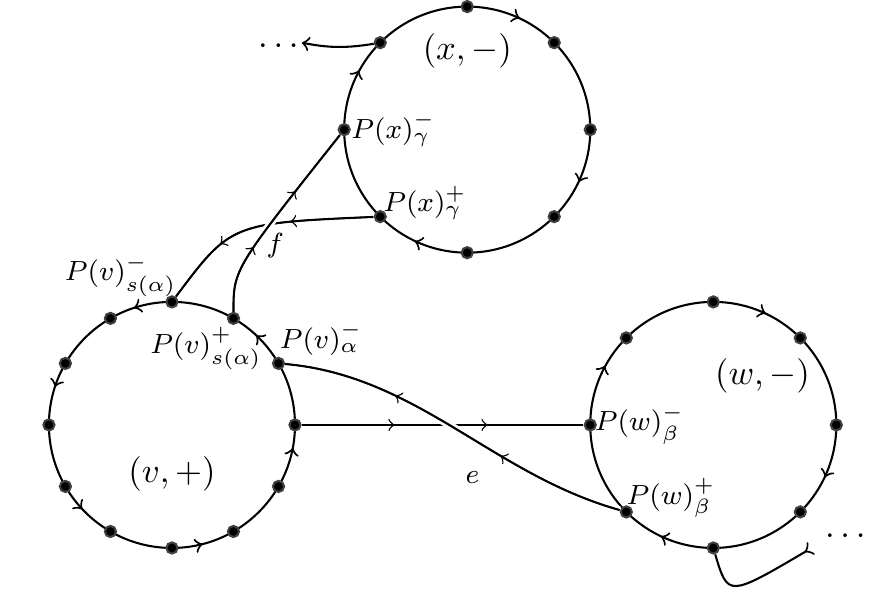}
\caption{   \small    Illustration of a boundary component.
\label{fig:external_loop} }
\end{subfigure}
\caption{        \small
 On the definition of boundary components.
\label{fig:bc}
}  
\end{figure}
Formally, these boundary components are described as follows: take an
arbitrary edge $e=[(v,\alpha),(w,\beta)]$ and let $s(\alpha)$ be the
next place according to the cyclic ordering given to $v$
(i.e. $s(\alpha)\equiv\alpha \pm 1$ (mod $n_v$) if $\epsilon_v=\pm$).
Thus, consider the path that begins with the segments
$\overline{P(w)_\beta^+ P(v)_\alpha^-},
\overline{P(v)_\alpha^-P(v)_{s(\alpha)}^+ } $.
We can juxtapose another segment, since there is a unique
$f\in\Rb\hp 1$ and a unique $x\in \Rb\hp 0$ with
$ f =[(v,s(\alpha)),(x,\gamma)]$.  The process finishes after a finite
number of steps by coming back to $P(w)_\beta^+$, at the latest, when
we run out of vertices. The loop $\ell_1$ obtained by concatenation of
these paths (see Fig. \hyperref[fig:loops_ribbon]{\ref*{fig:bc}
  \textsc{(b)}})
\[
\ell_1= \overline{P(w)_\beta^+ P(v)_\alpha^-}
\overline{ P(v)_{s(\alpha)}^+ } 
\overline{P(x)_{\gamma}^- \ldots P(w)_\beta^+} 
\]
is a boundary component. It might be that the not all vertices lie on
$\ell_1$, so pick one such a vertex and repeat the process to get
$\ell_2$. The final number of non-intersecting loops ---that is, after
each $0$-cell $p\in X\hp 0$ lies precisely in one of the
$\ell_1,\ell_2,\ldots, \ell_{\mtr{bc}(\Rb)}$ --- defines
$\mtr{bc}(\Rb)$, the \textit{number of boundary components}.  \\

Thus each boundary component
$\ell$ is, by construction, 
homeomorphic to $\mathbb S^1_\ell$ by certain
map, say $\phi_\ell$. Then we attach to the 
ribbon picture $X\hp 2$ a $2$-cell $\mathring D^2_\ell$
by such a map
 $\phi_\ell: \partial \mathring D^2_\ell \to  
\mathbb S^1_\ell $, for
each $\ell=1,\ldots, \mtr{bc}(\Rb)$. The cell-structure
of $\Sigma(\Rb)$ is that of $\Rb$ but with 
$\mathrm{bc}(\Rb)$ more $2$-cells. Thus
\begin{align*}
|0\mbox{-cells of } \Sigma (\Rb)| & = 2\sum_{v \in \Rb\hp 0} n_v,  \\
|1\mbox{-cells of } \Sigma (\Rb)| & = 2\sum_{v \in \Rb\hp 0} n_v+ 2 |\Rb \hp 1|,\\
|2\mbox{-cells of } \Sigma (\Rb)| & = |\Rb\hp 0| + |\Rb\hp 1|+ \mtr{bc}(\Rb).
\end{align*}
Since we have a finite cell-complex, we conclude from this:
\begin{align*}
\chi(\Sigma(\Rb)) = \sum_j (-1)^j| \mbox{$j$-cells of }\Sigma (\Rb)|  = |\Rb\hp 0| -| \Rb\hp 1| + \mtr{bc} (\Rb).
\end{align*}
Since $\Rb$ has been assigned a cell-structure as well,
`$\chi (\Rb)$' is now misleading. It will \textit{not}  
denote $\chi(X(\Rb))$ but $\chi(\Sigma(\R))$, as in 
stated in Definition \ref{genus_ribbon}.
\begin{rem}
Notice that, by construction, the inclusion
of the cell-complex
into a closed, compact surface $\Sigma(\Rb)$ is an embedding 
$\!\!\raisebox{-.35\height}{\includegraphics[height=.6cm]{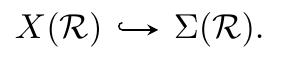}}$
\end{rem}
The next result is nothing unexpected. It shows that the two
definitions of Euler characteristic harmonically coexist:
\begin{prop} The Euler characteristic of a
  $3$-colored graph is the same, either if we compute it by its
  bubble homology or, by appealing its ribbon graph structure, via its
  geometric realization.\label{thm:eulers}
\end{prop}
\begin{proof}
Let $\G$ be a $3$-colored graph. Because of Lemma \ref{thm:col_rib},
we can consider the cell complex $X(\G)$ embedded in the ribbon graph
realization $\Sigma(\G)$.  We transform $\Sigma(\G)$ into another
cell-complex $Y(\G)$ by a deformation retraction.
\begin{figure} \centering
\begin{subfigure}{0.49\textwidth}
\hspace{1cm}
\includegraphics[height=5.8cm]{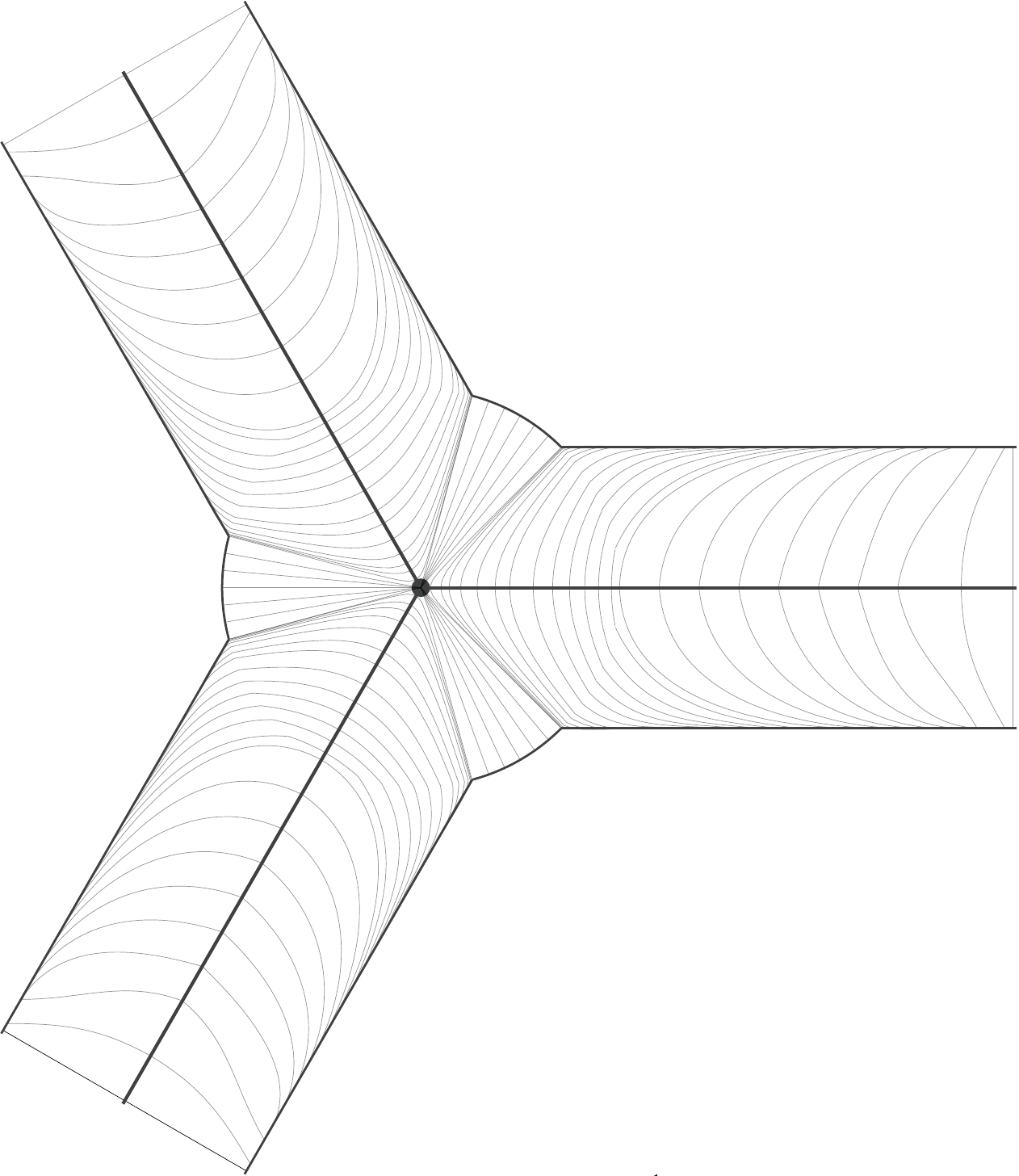}
\caption{\footnotesize The deformation retraction 
$\Sigma(\G)\to Y(\G)$ in local coordinates on a 
neighborhood of a vertex. The disk is 
continuously collapsed to a point-like vertex, 
and the ribbons sent to the ($1$-dimensional) half-edges
as shown.\label{fig:defo}
} 
\end{subfigure}
\hspace*{\fill} 
\begin{subfigure}{0.39\textwidth}
\hspace{1.5cm}
\includegraphics[width=3.2719932cm]{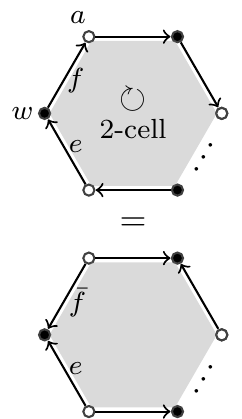}
\vspace{.5cm}
\caption{\footnotesize On the two different orientations of edges.  The
  first, determined by the cell-complex construction of ribbon graphs
  (above) and the lower one by the cell-attachment by bubble homology,
  see eq. \eqref{Rand_Kettenkomplex} with $p=2$.
} \label{fig:inversion}
\end{subfigure}
\caption{  \small On the proof of Proposition \ref{thm:eulers}.
\label{fig:inversion_edges}
}  \label{fig:para_prueba}
\end{figure}
We retract the ribbons' disks to thin edges and 
the vertices' disks to actual point-like vertices; 
see Fig. \hyperref[fig:defo]{\ref*{fig:inversion_edges} \textsc{(a)}}.
We end up with 
\begin{align*}
|0\mbox{-cells}| & = |\G\hp 0|, &&  
|1\mbox{-cells}|=|\G\hp 1|, &&  
|2\mbox{-cells}|  = \mbox{bc}(\G).
\end{align*}
We claim that the complex associated to $Y(\G)$ is the same
cell-complex obtained by bubble homology.  \par
For cells of dimensions $p=0,1$, the statement is trivially verified,
since a graph is, naturally, a cell-complex.  For dimension $2$, we
observe that each boundary component we attached $2$-cells to, was
formed by arcs on the disks, and segments on the next edge (determined
by the cyclic ordering at the vertex).  After the deformation
retraction, the arcs no longer exist.  Therefore the boundary
component is now composed by the $1$-cells determined by only the
edges as follows. Pick a boundary component and an arbitrary edge $e$
lying on it. Let $c$ be the color of $e$.  Then pick the black vertex
(name it $w$) $e$ is attached to.  The next edge $f$ has then color
$c-1$ (mod $3$), since $w$ has orientation $(321)$. By the same token,
$f$ is attached to a white vertex, say $a$, that determines the next
edge lying on the boundary component; this has color $c$ again, for
$a$ has been assigned the orientation $(123)$.  This yields then a
sequence of edges that forms connected path of edges of alternating
colors $\{c, c-1\}$. Thus, each one of the attached $2$-cells to the
boundary components of $X(\G)$ corresponds, after the deformation
retraction, to a bicolored connected path. These are, by definition,
$2$-bubbles. Therefore $ C_2(Y(\G))=C_2(\G).$\par

On the equivalence of the boundary operators:  The boundary
operator $\partial_2':C_2(Y(\G))\to C_1(Y(\G))$ is the sum of all
edges that lie on a boundary component, with `compatible orientation'
(that is, if the edges $e$ and $f$ meet at the vertex $w$, then $f$
and $e$ point in opposite directions, as seen from $w$). Since any
generator of $C_2(Y(\G))$ is of the form $\B^{(ij)}$, with $i<j$,
\begin{align*}
\partial'_2 (\B^{(ij)}) = \sum_{\mtr{all\, edges}, e_c\subset\B^{(ij)}} e_c  &= 
\sum_{e\in \G\hp1_i} e+\sum_{f\in \G\hp1_j}  f \\
&=\sum_{e\in \G\hp1_i} e-\sum_{f\in \G\hp1_j} \overline {\!f} =\partial_2(\B^{(ij)}).
\end{align*}
where $\overline {\!f}$ is the edge with the opposite orientation 
(cf. Fig. \hyperref[fig:inversion]{\ref*{fig:inversion_edges} \textsc{(b)}}).
\end{proof}

%% 
%%  = = = = = = = = = = = = = = = = =  
%% 
%%  = = = = = = = = = = = = = = = = =
%% 
%%  = = = = = = = = = = = = = = = = = 
%% 

\bibliography{surgery_CTM} 

\begin{thebibliography}{10}

\bibitem{Ambjorn}
Jan Ambj\o{}rn, Bergfinnur Durhuus, and Thordur Jonsson.
\newblock {Three-dimensional simplicial quantum gravity and generalized matrix
  models}.
\newblock {\em Mod. Phys. Lett.}, A6:1133--1146, 1991.

\bibitem{Atiyah}
M.~Atiyah.
\newblock {Topological quantum field theories}.
\newblock {\em Inst. Hautes Etudes Sci. Publ. Math.}, 68:175--186, 1989.

\bibitem{baez}
John~C. Baez.
\newblock {Spin network states in gauge theory}.
\newblock {\em Adv. Math.}, 117:253--272, 1996.
\newblock arXiv:gr-qc/9411007.

\bibitem{barrett}
John~W. Barrett and Lisa Glaser.
\newblock {Monte Carlo simulations of random non-commutative geometries}.
\newblock {\em J. Phys.}, A49(24):245001, 2016.
\newblock arXiv:1510.01377.

\bibitem{4renorm}
Joseph Ben~Geloun and Vincent Rivasseau.
\newblock {A Renormalizable 4-Dimensional Tensor Field Theory}.
\newblock {\em Commun. Math. Phys.}, 318:69--109, 2013.
\newblock arXiv:1111.4997.

\bibitem{3Dbeta}
Joseph Ben~Geloun and Dine~Ousmane Samary.
\newblock {3D Tensor Field Theory: Renormalization and One-loop
  $\beta$-functions}.
\newblock {\em Annales Henri Poincare}, 14:1599--1642, 2013.
\newblock arXiv:1201.0176.

\bibitem{critical}
Valentin Bonzom, R\u{a}zvan Gur\u{a}u, Aldo Riello, and Vincent Rivasseau.
\newblock {Critical behavior of colored tensor models in the large N limit}.
\newblock {\em Nucl. Phys.}, B853:174--195, 2011.
\newblock arXiv:1105.3122.

\bibitem{uncoloring}
Valentin Bonzom, R\u{a}zvan Gur\u{a}u, and Vincent Rivasseau.
\newblock {Random tensor models in the large N limit: Uncoloring the colored
  tensor models}.
\newblock {\em Phys. Rev.}, D85:084037, 2012.
\newblock arXiv:1202.3637.

\bibitem{carrozzaPhD}
Sylvain Carrozza.
\newblock {\em {Tensorial methods and renormalization in Group Field
  Theories}}.
\newblock PhD thesis, Orsay, LPT, 2013.

\bibitem{Carrozza_review}
Sylvain Carrozza.
\newblock {Flowing in group field theory space: a review}.
\newblock {\em SIGMA}, 2016.
\newblock arXiv:1603.01902.

\bibitem{su2}
Sylvain Carrozza, Daniele Oriti, and Vincent Rivasseau.
\newblock {Renormalization of a SU(2) Tensorial Group Field Theory in Three
  Dimensions}.
\newblock {\em Commun. Math. Phys.}, 330:581--637, 2014.
\newblock arXiv:1303.6772.

\bibitem{On}
Sylvain Carrozza and Adrian Tanas\u{a}.
\newblock {$O(N)$ Random Tensor Models}.
\newblock {\em Lett. Math. Phys.}, 106(11):1531--1559, 2016.
\newblock arXiv:1512.06718.

\bibitem{SAP}
Ali~H. Chamseddine and Alain Connes.
\newblock {The Spectral action principle}.
\newblock {\em Commun. Math. Phys.}, 186:731--750, 1997.

\bibitem{CCM}
Ali~H. Chamseddine, Alain Connes, and Matilde Marcolli.
\newblock {Gravity and the standard model with neutrino mixing}.
\newblock {\em Adv. Theor. Math. Phys.}, 11(6):991--1089, 2007.

\bibitem{diFrancesco_rect}
P.~Di~Francesco.
\newblock {Rectangular matrix models and combinatorics of colored graphs}.
\newblock {\em Nucl. Phys.}, B648:461--496, 2003.
\newblock arXiv:cond-mat/0208037.

\bibitem{dFGZ}
P.~Di~Francesco, Paul~H. Ginsparg, and Jean Zinn-Justin.
\newblock {2-D Gravity and random matrices}.
\newblock {\em Phys. Rept.}, 254:1--133, 1995.
\newblock arXiv:hep-th/9306153.

\bibitem{survey_cryst}
M.~Ferri, C.~Gagliardi, and L.~Grasselli.
\newblock A graph-theoretical representation of {PL}-manifolds --- a survey on
  crystallizations.
\newblock {\em {Aequationes Mathematicae}}, 31(1):121--141, 1986.

\bibitem{Freidel}
Laurent Freidel.
\newblock {Group field theory: An Overview}.
\newblock {\em Int. J. Theor. Phys.}, 44:1769--1783, 2005.
\newblock arXiv:hep-th/0505016.

\bibitem{renormTFT}
Joseph~Ben Geloun.
\newblock {Renormalizable Tensor Field Theories}.
\newblock In {\em {18th International Congress on Mathematical Physics
  (ICMP2015) Santiago de Chile, Chile, July 27-August 1, 2015}}, 2016.
\newblock arXiv:1601.08213.

\bibitem{GW12}
Harald Grosse and Raimar Wulkenhaar.
\newblock {Self-Dual Noncommutative $\phi^4$ -Theory in Four Dimensions is a
  Non-Perturbatively Solvable and Non-Trivial Quantum Field Theory}.
\newblock {\em Commun. Math. Phys.}, 329:1069--1130, 2014.
\newblock arXiv:1205.0465.

\bibitem{Gurau:2009tw}
R\u{a}zvan Gur\u{a}u.
\newblock {Colored Group Field Theory}.
\newblock {\em Commun. Math. Phys.}, 304:69--93, 2011.
\newblock arXiv:0907.2582.

\bibitem{Nexpansion_coloured}
R\u{a}zvan Gur\u{a}u.
\newblock {The $1/N$ expansion of colored tensor models}.
\newblock {\em Annales Henri Poincare}, 12:829--847, 2011.
\newblock arXiv:1011.2726.

\bibitem{gurauReview}
R\u{a}zvan Gur\u{a}u.
\newblock {A review of the large-$N$ limit of tensor models}.
\newblock {\em Symmetries and Groups in Contemporary Physics}, 2012.
\newblock arXiv:1209.4295.

\bibitem{Nexpansion}
R\u{a}zvan Gur\u{a}u.
\newblock {The complete $1/N$ expansion of colored tensor models in arbitrary
  dimension}.
\newblock {\em Annales Henri Poincare}, 13:399--423, 2012.
\newblock arXiv:1102.5759.

\bibitem{universality}
R\u{a}zvan Gur\u{a}u.
\newblock {Universality for Random Tensors}.
\newblock {\em Ann. Inst. H. Poincare Probab. Statist.}, 50(4):1474--1525,
  2014.
\newblock arXiv:1111.0519.

\bibitem{GurauRyan}
R\u{a}zvan Gur\u{a}u and James~P. Ryan.
\newblock {Colored Tensor Models - a review}.
\newblock {\em SIGMA}, 8:020, 2012.
\newblock arXiv:1109.4812.

\bibitem{krajewskireiko}
Thomas Krajewski and Reiko Toriumi.
\newblock {Exact Renormalisation Group Equations and Loop Equations for Tensor
  Models}.
\newblock {\em SIGMA}, 12:068, 2016.

\bibitem{MvS}
Matilde Marcolli and Walter~D. van Suijlekom.
\newblock Gauge networks in noncommutative geometry.
\newblock {\em Journal of Geometry and Physics}, 75:71 -- 91, 2014.
\newblock arXiv:1301.3480.

\bibitem{moise}
Edwin~E. Moise.
\newblock {Affine Structures in 3-Manifolds: V. The Triangulation Theorem and
  Hauptvermutung}.
\newblock {\em Annals of Mathematics}, 56(1):96--114, 1952.

\bibitem{mulase}
M.~Mulase and M.~Penkava.
\newblock Ribbon graphs, quadratic differentials on {R}iemann surfaces, and
  algebraic curves defined over {$\overline{ Q}$}.
\newblock {\em Asian J. Math.}, 2(4):875--919, 1998.

\bibitem{nlab}
nLab.
\newblock
  {\href{https://ncatlab.org/nlab/show/ribbon+graph}{\color{blue}https://ncatlab.org/nlab/show/ribbon+graph}}.

\bibitem{OritiGFT}
Daniele Oriti.
\newblock {Group Field Theory and Loop Quantum Gravity}.
\newblock 2014.
\newblock arXiv:1408.7112.

\bibitem{dine_beta}
Dine Ousmane~Samary.
\newblock {Beta functions of $\mathrm{U}(1)^d$ gauge invariant just
  renormalizable tensor models}.
\newblock {\em Phys. Rev.}, D88(10):105003, 2013.
\newblock arXiv:1303.7256.

\bibitem{us}
Dine {Ousmane Samary}, Carlos~I. {P\'erez-S\'anchez}, Fabien
  {Vignes-Tourneret}, and Raimar Wulkenhaar.
\newblock {Correlation functions of a just renormalizable tensorial group field
  theory: the melonic approximation}.
\newblock {\em Class. Quant. Grav.}, 32(17):175012, 2015.
\newblock arXiv:1411.7213.

\bibitem{fullward}
Carlos~I. P\'erez-S\'anchez.
\newblock {The full Ward-Takahashi Identity for colored tensor models}.
\newblock 2016.
\newblock arXiv:1608.08134.

\bibitem{pezzana}
Mario Pezzana.
\newblock Sulla struttura topologica delle variet\`a compatte.
\newblock {\em Ati Sem. Mat. Fis. Univ. Modena}, 23(1):269--277, 1975.

\bibitem{tensor_theory_space}
Vincent Rivasseau.
\newblock {The Tensor Theory Space}.
\newblock {\em Fortsch. Phys.}, 62:835--840, 2014.
\newblock arXiv:1407.0284.

\bibitem{RTM_QG}
Vincent Rivasseau.
\newblock {Random Tensors and Quantum Gravity}.
\newblock {\em SIGMA}, 12:069, 2016.
\newblock arXiv:1603.07278.

\bibitem{Ryan}
James~P. Ryan.
\newblock {Tensor models and embedded Riemann surfaces}.
\newblock {\em Phys. Rev.}, D85:024010, 2012.
\newblock arXiv:1104.5471.

\bibitem{DineWard}
Dine~Ousmane Samary.
\newblock {Closed equations of the two-point functions for tensorial group
  field theory}.
\newblock {\em Class. Quant. Grav.}, 31:185005, 2014.
\newblock arXiv:1401.2096.

\bibitem{wgauge}
Dine~Ousmane Samary and Fabien Vignes-Tourneret.
\newblock {Just Renormalizable TGFT's on $U(1)^{d}$ with Gauge Invariance}.
\newblock {\em Commun. Math. Phys.}, 329:545--578, 2014.
\newblock arXiv:1211.2618.

\bibitem{MO}
Adrian Tanas\u{a}.
\newblock {Multi-orientable Group Field Theory}.
\newblock {\em J. Phys.}, A45:165401, 2012.
\newblock arXiv:1109.0694.

\bibitem{reviewTanasa}
Adrian Tanas\u{a}.
\newblock {The Multi-Orientable Random Tensor Model, a Review}.
\newblock {\em SIGMA}, 12:056, 2016.

\end{thebibliography}
\bibliographystyle{plain}

%% 
%%  = = = = = = = = = = = = = = = = =  
%% 
%%  = = = = = = = = = = = = = = = = =
%% 
%%  = = = = = = = = = = = = = = = = = 
%% 

\end{document}